\documentclass[sigconf, nonacm]{acmart}

\usepackage[ruled,vlined,linesnumbered]{algorithm2e}
\SetKwProg{Fn}{Function}{}{}
\SetKw{Break}{break}
\SetKwComment{Comment}{\{}{\}}

\usepackage{multirow}
\usepackage{subcaption}
\usepackage{enumitem}
\usepackage{framed}
\graphicspath{{figs/}}

\DeclareMathOperator*{\argmin}{arg\,min}
\DeclareMathOperator*{\argmax}{arg\,max}

\newcommand{\ie}{{i.e.,}\xspace}
\newcommand{\eg}{{e.g.,}\xspace}

\newcommand{\AlgNWF}{\textsc{Greedy}\xspace}
\newcommand{\AlgSPH}{\textsc{Sphere}\xspace}
\newcommand{\AlgDMM}{\textsc{DMM}\xspace}
\newcommand{\AlgBG}{\textsc{BiGreedy}\xspace}
\newcommand{\AlgIBG}{\textsc{BiGreedy+}\xspace}
\newcommand{\AlgTwoD}{\textsc{IntCov}\xspace}

\begin{document}
\title{Happiness Maximizing Sets under Group Fairness Constraints}

\author{Jiping Zheng}
\authornote{Jiping Zheng is also with the State Key Laboratory for Novel Software Technology, Nanjing University, Nanjing, China.}
\affiliation{%
  \institution{Nanjing University of Aeronautics and Astronautics}
  \city{Nanjing}
  \country{China}
}
\email{jzh@nuaa.edu.cn}
\author{Yuan Ma}
\affiliation{%
  \institution{Nanjing University of Aeronautics and Astronautics}
  \city{Nanjing}
  \country{China}
}
\email{mayuancs@nuaa.edu.cn}
\author{Wei Ma}
\affiliation{%
  \institution{Nanjing University of Aeronautics and Astronautics}
  \city{Nanjing}
  \country{China}
}
\email{mawei@nuaa.edu.cn}
\author{Yanhao Wang}
\authornote{Corresponding author}
\affiliation{%
  \institution{East China Normal University}
  \city{Shanghai}
  \country{China}
}
\email{yhwang@dase.ecnu.edu.cn}
\author{Xiaoyang Wang}
\affiliation{%
  \institution{The University of New South Wales}
  \city{Sydney}
  \state{NSW}  
  \country{Australia}
}
\email{xiaoyang.wang1@unsw.edu.au}

\begin{abstract}
  Finding a happiness maximizing set (HMS) from a database, i.e., selecting a small subset of tuples that preserves the best score with respect to any nonnegative linear utility function, is an important problem in multi-criteria decision-making.
  When an HMS is extracted from a set of individuals for assisting data-driven algorithmic decisions such as hiring and admission, it is crucial to ensure that the HMS can fairly represent different groups of candidates without bias and discrimination.
  However, although the HMS problem was extensively studied in the database community, existing algorithms do not take \emph{group fairness} into account and may provide solutions that under-represent some groups.
  
  In this paper, we propose and investigate a fair variant of HMS (FairHMS) that not only maximizes the minimum happiness ratio but also guarantees that the number of tuples chosen from each group falls within predefined lower and upper bounds.
  Similar to the vanilla HMS problem, we show that FairHMS is NP-hard in three and higher dimensions.
  Therefore, we first propose an exact interval cover-based algorithm called \AlgTwoD for FairHMS on two-dimensional databases.
  Then, we propose a bicriteria approximation algorithm called \AlgBG for FairHMS on multi-dimensional databases by transforming it into a submodular maximization problem under a matroid constraint.
  We also design an adaptive sampling strategy to improve the practical efficiency of \AlgBG.
  Extensive experiments on real-world and synthetic datasets confirm the efficacy and efficiency of our proposal.
\end{abstract}

\maketitle

\section{Introduction}
\label{sec:intro}

Scoring, ranking, and selecting tuples from a database\footnote{The terms ``database''/``dataset'' and ``tuple''/``point'' will be used interchangeably throughout this paper.} based on a combination of multiple criteria is an essential function of modern data management systems.
Since the number of tuples matching a query can be large, it may be impossible for a decision-maker to search within it entirely.
This gives rise to the database operators to select a small yet representative subset of tuples for supporting decision-making.
In the database community, the most widely used operators to fulfill this task are \emph{top-$k$ queries}~\cite{Ilyas:2008}, \emph{skyline queries}~\cite{Borzsony:2001}, and \emph{regret minimizing} (or \emph{happiness maximizing}) \emph{sets}~\cite{Nanongkai:2010,Xie:2020,Qiu:2018}.

A top-$k$ query outputs $k$ tuples with the highest scores based on a predefined utility function.
Unfortunately, in many cases, we may not provide the utility function exactly in advance or should consider different trade-offs among multiple criteria that cannot be captured by any single utility function.
Skyline queries do not need any explicit utility function anymore.
They are based on the concept of ``dominance'': a tuple $p$ dominates another tuple $q$ if and only if $p$ is no worse than $q$ in all attributes and strictly better than $q$ in at least one attribute.
The results of skyline queries contain all the tuples that are not dominated by any other tuple and thus include the best tuples w.r.t.~all possible utility functions for different trade-offs among attributes.
However, the output sizes of skyline queries are uncontrollable, particularly so when the dimensionality of the database is high.
Therefore, the regret minimizing set (RMS) problem~\cite{Nanongkai:2010} is proposed to avoid the deficiencies of both top-$k$ and skyline queries.
In the RMS problem, the notion of \emph{regret ratio} is used to quantify how regretful a user is if she/he obtains the best tuple in the selected subset instead of the best tuple in the whole database with regard to a utility function.
And an RMS is a set of $k$ tuples such that the maximum of regret ratios among all nonnegative linear utility functions is minimized.
Furthermore, the minimization of \emph{maximum regret ratios} (MRR) is often transformed into the maximization of \emph{minimum happiness ratios} (MHR)~\cite{Luenam:2021, Xie:2020, Qiu:2018} for ease of theoretical analysis, as the MHR of any subset is always equal to one minus its MRR.
This equivalent form of RMS is called the \emph{happiness maximizing set} (HMS) problem.
Although there have been extensive studies on RMS and HMS (\eg~\cite{Nanongkai:2010,Asudeh:2017,Agarwal:2017,Peng:2014,Chester:2014,Qi:2018,Xie:2018,Xie:2019,Xie:2020,Luenam:2021,Faulkner:2015}, see~\cite{Xie:2020VLDBJ} for a survey), they only consider the numerical attributes in scoring, ranking and representative subset selection, but ignore the remaining attributes.

\begin{table}[t]
	\centering
	\footnotesize
	\caption{Example of tuples in the LSAC database}
	\label{tab:example}
	\begin{tabular}{ccccc}
		\toprule
		\textbf{Applicant ID} & \textbf{Gender} & \textbf{Race} & \textbf{LSAT} (140-180) & \textbf{GPA} (0-4) \\
		\midrule
		$a_1$ & \emph{Female} & \emph{Black}    & $164$ & $3.31$ \\
		$a_2$ & \emph{Male}   & \emph{Black}    & $163$ & $3.55$ \\
		$a_3$ & \emph{Female} & \emph{White}    & $165$ & $3.09$ \\
		$a_4$ & \emph{Male}   & \emph{White}    & $160$ & $3.83$ \\
		$a_5$ & \emph{Male}   & \emph{Hispanic} & $170$ & $2.79$ \\
		$a_6$ & \emph{Female} & \emph{Hispanic} & $161$ & $3.69$ \\
		$a_7$ & \emph{Male}   & \emph{Asian}    & $153$ & $3.89$ \\
		$a_8$ & \emph{Female} & \emph{Asian}    & $156$ & $3.87$ \\
		\bottomrule
	\end{tabular}
\end{table}

In many real-world applications, a database is composed of a set of individuals grouped by sensitive categorical attributes such as \emph{gender} and \emph{race}, from which a subset is extracted for assisting data-driven algorithmic decisions~\cite{DBLP:journals/datamine/Zliobaite17} in employment, banking, education, and so on.
The societal impact of these algorithmic decisions has raised concerns about \emph{fairness}.
Several pieces of evidence have indicated that an algorithm, if left unchecked, could ``\emph{systematically and unfairly discriminate against certain individuals or groups of individuals in favor of others}''~\cite{Friedman:1996} and such discrimination might be further passed to algorithmic decisions~\cite{Chouldechova:2020, Stoyanovich:2020}.
For RMS/HMS problems, the selected subset would be biased by over-representing particular groups while under-representing the others.
As an illustrative example, Table~\ref{tab:example} shows a database of eight applicants, each of whom is denoted by two numerical attributes (\ie \emph{LSAT} and \emph{GPA}) and two demographic attributes (\ie \emph{Gender} and \emph{Race}), in the Law School Admission Council (LSAC) database\footnote{\url{http://www.seaphe.org/databases.php}}.
Typically, the combination of \emph{LSAT} and \emph{GPA} is used to score and rank the applicants for admission decisions.
Since all the applicants are in the skyline but the quota for admission is limited, we should select a subset of $k = 3$ applicants to admit from the database to best represent all combinations of both attributes.
If an HMS algorithm is used for this task, the subset $\{a_4,a_5,a_7\}$ will be returned since it achieves the highest minimum happiness ratio of $0.9984$ among all size-$3$ sets of tuples.
We observe that the solution only contains \emph{male} applicants, which implies discrimination against \textit{female} applicants because they constitute half of all applicants in the database and are equally eligible as none of the selected applicants can dominate them.
Such a biased selection would further lead to unfairness in education opportunities.
Therefore, ensuring a fair representation of each group in the solution is necessary for RMS and HMS.
However, to our knowledge, none of the existing algorithms have taken group-level fairness into account.

To fill the gap, we propose and investigate a fair variant of HMS in this paper.
Specifically, we introduce an additional \emph{group fairness constraint}, a well-established definition of fairness widely adopted in many real-world problems including top-$k$ selection~\cite{Stoyanovich:2018,DBLP:conf/fat/MehrotraC21}, data summarization~\cite{DBLP:conf/icml/KleindessnerAM19,Celis:2018}, submodular maximization~\cite{DBLP:conf/nips/HalabiMNTT20,Wang:2021a}, and elsewhere~\cite{DBLP:conf/ijcai/CelisHV18,Moumoulidou:2021,DBLP:conf/wsdm/MaGTW22}, into HMS.
We consider that the tuples in a database $\mathcal{D}$ are partitioned into $C$ disjoint groups $\{\mathcal{D}_1, \mathcal{D}_2, \ldots, \mathcal{D}_C\}$ by a certain sensitive attribute, \eg \emph{gender} or \emph{race}.
To fairly represent each group $\mathcal{D}_c$ for $c \in [C]$, we restrict that the number $k_c$ of tuples selected from $\mathcal{D}_c$ falls within predefined lower and upper bounds, \ie $l_c \leq k_c \leq h_c$. Consequently, we define the \emph{fair happiness maximization set} (FairHMS) problem that asks for a set of $k$ tuples from $\mathcal{D}$ to maximize the minimum happiness ratio (MHR) while ensuring the satisfaction of the above-defined group fairness constraint.
Similar to vanilla HMS, FairHMS is NP-hard in three or higher dimensions.

We propose an exact interval cover-based algorithm called \AlgTwoD for FairHMS on two-dimensional databases.
In the \AlgTwoD algorithm, the space of nonnegative linear utility functions in $\mathbb{R}^{2}_{+}$ is denoted as an interval $[0, 1]$ based on their angles from $0$ to $\frac{\pi}{2}$~\cite{Chester:2014,Asudeh:2017,Cao:2017}.
Then, based on geometric transformations in the plane, FairHMS is converted into the problem of determining whether there is a feasible set of points under the fairness constraint whose corresponding sub-intervals fully cover the interval $[0, 1]$.
By solving decision problems with dynamic programs, we can find the optimal solution of FairHMS in polynomial time when $C = O(1)$.

For databases in $\mathbb{R}^d_{+}$ with $d \geq 2$, since the happiness ratio with respect to any non-negative linear utility function is submodular, and the fairness constraint is a case of \emph{matroid constraints}~\cite{Korte2012, DBLP:conf/nips/HalabiMNTT20}, we transform FairHMS into a multi-objective submodular maximization problem under a matroid constraint based on the notion of $\delta$-nets.
We propose the \AlgBG algorithm that provides a bicriteria approximate solution for FairHMS based on existing methods for multi-objective submodular maximization~\cite{Anari:2019, Krause:2008}.
Furthermore, we exploit an adaptive sampling strategy to reduce the sizes of $\delta$-nets and thus improve the efficiency of the \AlgBG algorithm.

Finally, we conduct extensive experiments to evaluate the performance of our proposed algorithms on real-world and synthetic datasets.
It is shown that the solutions of existing RMS/HMS algorithms do not satisfy the group fairness constraints in almost all cases if left unchecked.
Moreover, the \emph{price of fairness} indicated by the decreases in MHRs caused by fairness constraints is low in most cases.
Our proposed algorithms provide solutions of higher quality than existing algorithms adapted for ensuring fairness.
Meanwhile, they have comparable efficiencies to existing algorithms and scale to massive datasets with millions of tuples.

To sum up, the main contributions of this paper are as follows.
\begin{itemize}
  \item We formally define FairHMS and show its NP-hardness in $\mathbb{R}^{d}_{+}$ when $d \geq 3$. (Section~\ref{sec:problem})
  \item We propose an exact algorithm called \AlgTwoD for FairHMS in $\mathbb{R}^{2}_{+}$. (Section~\ref{sec:2d})
  \item We propose a bicriteria approximation algorithm called \AlgBG for FairHMS in $\mathbb{R}^{d}_{+}$ for any $d \geq 2$ and a practical strategy for speeding up \AlgBG. (Section~\ref{sec:algorithm})
  \item We compare our proposed algorithms against the state-of-the-art RMS and HMS algorithms to show their effectiveness, efficiency, and scalability. (Section~\ref{sec:experiments})
\end{itemize}

\section{Preliminaries}
\label{sec:problem}

In this section, we first introduce the basic concepts, then define the FairHMS problem, and finally discuss the properties of FairHMS.

\smallskip\noindent\textbf{Data Model:}
We consider a database $\mathcal{D}$ of $n$ tuples, each of which has $d$ nonnegative numeric attributes and is represented as a point $p = (p[1], p[2], \ldots, p[d]) \in \mathbb{R}^d_{+}$. Without loss of generality, we assume that each numeric attribute is normalized to the range $[0,1]$ and larger values are preferred. This assumption is mild because of the scale invariance of regret/happiness ratios~\cite{Nanongkai:2010}.

\smallskip\noindent\textbf{Scoring Model:}
We focus on the class of nonnegative linear functions for score computation, where a $d$-dimensional utility vector $u = (u[1], u[2], \ldots, u[d]) \in \mathbb{R}^d_{+}$ represents a user's preference over the $d$ attributes. Note that linear functions are adopted by most of the existing literature on RMS/HMS problems~\cite{Nanongkai:2010,Peng:2014,Chester:2014,Asudeh:2017,Xie:2018,Asudeh:2019,Xie:2019,Shetiya:2020,Xie:2020,Wang:2021a,Agarwal:2017} for modeling users' preferences. Specifically, given a point $p$ and a utility function $f_u$ with regard to a vector $u$, the utility of point $p$ is computed as the inner product of $u$ and $p$, \ie $f_u(p) = \sum_{i=1}^{d} u[i] \cdot p[i]$. Because of the scale invariance of regret/happiness ratios~\cite{Nanongkai:2010}, we assume that the utility vector is rescaled to be a unit vector based on the $l_1$-norm or $l_2$-norm, \ie $\sum_{i=1}^{d} u[i] = 1$ or $\sqrt{\sum_{i=1}^{d} u^2[i]} = 1$. Note that the class of all nonnegative linear utility vectors after the $l_2$-normalization corresponds to the unit $(d-1)$-sphere $\mathbb{S}^{d-1}_{+}$.

\smallskip\noindent\textbf{Minimum Happiness Ratio~\cite{Xie:2020}:}
Given a subset $S \subseteq \mathcal{D}$ of points and a utility vector $u$, the \emph{happiness ratio} (HR) of $S$ over $\mathcal{D}$ w.r.t.~$u$ is defined as
$hr(u, S, \mathcal{D}) = \frac{\max_{p \in S}f_u(p)}{\max_{p \in \mathcal{D}}f_u(p)}$.
Intuitively, the happiness ratio measures how satisfied a user is when she/he sees the tuple with the highest score from the subset $S$ instead of the tuple with the highest score from the database $\mathcal{D}$. We need to find a subset that achieves a good happiness ratio for all utility functions to satisfy all possible users. The \emph{minimum happiness ratio} (MHR) is defined to measure the minimum of the happiness ratios achieved by $S$ over $\mathcal{D}$ in the worst case, \ie $mhr(S, \mathcal{D}) = \min_{u \in \mathbb{S}^{d-1}_{+}} hr(u, S, \mathcal{D})$.
When the context is clear, we will drop $\mathcal{D}$ from the notations of HR and MHR. By definition, $hr(u, S)$ and $mhr(S)$ are in the range $[0,1]$ for any subset $S$ and vector $u$.

\smallskip\noindent\textbf{Fairness Model:}
We adopt a well-established model for group fairness in the existing literature~\cite{DBLP:conf/nips/HalabiMNTT20,Stoyanovich:2018,DBLP:conf/fat/MehrotraC21,DBLP:conf/ijcai/CelisHV18}. In addition to the numeric attributes as discussed in the data model, the tuples in the database $\mathcal{D}$ are also associated with one or more categorical attributes. When the tuples denote individuals, each categorical attribute often corresponds to a protected feature such as \emph{gender} or \emph{race}. Suppose that $\mathcal{A}$ is the domain of an attribute $A$ with $C = |\mathcal{A}|$. The database $\mathcal{D}$ is partitioned into $C$ disjoint groups $\{\mathcal{D}_1, \mathcal{D}_2, \ldots, \mathcal{D}_C\}$ by attribute $A$ with $\mathcal{D} = \bigcup_{c=1}^{C} \mathcal{D}_c$. The partitioning scheme can also be generalized to multiple attributes $A_1, \ldots, A_r$, where $\mathcal{D}$ is divided into $C = \prod_{j=1}^{r} C_j$ groups, each of which corresponds to a unique combination of values for all the $r$ attributes. For example, the tuples in Table~\ref{tab:example} can be partitioned into $2$, $4$, and $8$ groups by \emph{gender}, \emph{race}, and both of them, respectively.

Given a database $\mathcal{D}$ with $C$ groups $\{\mathcal{D}_1, \mathcal{D}_2, \ldots, \mathcal{D}_C\}$, the group fairness constraint is defined by restricting that the number of tuples selected from each group $\mathcal{D}_c$ is between a lower bound $l_c$ and an upper bound $h_c$ in $\mathbb{Z}_{+}$ with $l_c \leq h_c$. Formally, a subset $S \subseteq \mathcal{D}$ is considered fair if and only if $l_c \leq |S \cap \mathcal{D}_c| \leq h_c$ for every $c \in [C]$. Furthermore, we limit the total number of tuples in the result set $S$ to a positive integer $k \in \mathbb{Z}_{+}$. In practice, the values of $l_c$ and $h_c$ for each $c \in [C]$ can be specified according to different concepts of \emph{fairness}. Two typical examples~\cite{DBLP:conf/nips/HalabiMNTT20} are to set $l_c= \lfloor (1 - \alpha) k \cdot \frac{|\mathcal{D}_c|}{|\mathcal{D}|} \rfloor$ and $h_c= \lceil (1 + \alpha) k \cdot \frac{|\mathcal{D}_c|}{|\mathcal{D}|} \rceil$ for \emph{proportional representation} and to set $l_c = \lfloor \frac{(1 - \alpha) k}{C} \rfloor$ and $h_c = \lceil \frac{(1 + \alpha) k}{C} \rceil$ for \emph{balanced representation}, given a parameter $\alpha \in (0, 1)$.

An important notion closely related to the above-defined fairness constraint is \emph{matroid}~\cite{Korte2012,DBLP:conf/nips/HalabiMNTT20}. Specifically, a matroid $\mathcal{M}$ is defined on a finite ground set $\mathcal{D}$ and a family $\mathcal{I}$ of subsets of $\mathcal{D}$ called \emph{independent sets} with the following properties: \emph{(i)} $\emptyset \in \mathcal{I}$; \emph{(ii)} If $S_1 \subset S_2 \subseteq \mathcal{D}$ and $S_2 \in \mathcal{I}$, then $S_1 \in \mathcal{I}$; \emph{(iii)} If $S_1, S_2 \in \mathcal{I}$ and $|S_2| > |S_1|$, then $\exists p \in S_2 \backslash S_1$ such that $S_1 \cup \{p\} \in \mathcal{I}$. Following the analysis in~\cite{DBLP:conf/nips/HalabiMNTT20}, we define a matroid $\mathcal{M} = (\mathcal{D}, \mathcal{I})$ based on the group fairness constraint, where the family of independent sets is
\begin{displaymath}
\mathcal{I} = \Big\{ S \subseteq \mathcal{D}: \sum_{c=1}^{C} \max\{|S \cap \mathcal{D}_c|, l_c\} \leq k \wedge |S \cap \mathcal{D}_c| \leq h_c, \forall c \in [C] \Big\}.
\end{displaymath}
As shown in~\cite{DBLP:conf/nips/HalabiMNTT20}, all feasible size-$k$ subsets of $\mathcal{D}$ under the group fairness constraint are in $\mathcal{I}$; and for any set $S \in \mathcal{I}$ with $|S| < k$, there exists at least one superset of $S$ satisfying the group fairness constraint. According to the above results, the group fairness constraint will be treated as a special case of matroid constraints in our proposed algorithms.

\smallskip\noindent\textbf{Problem Formulation:}
Next, we formally define the fair happiness maximizing set (FairHMS) problem based on the above notions.
\begin{definition}[FairHMS]
\label{def:problem}
  Given a database $\mathcal{D}$ with $C$ disjoint groups $\{\mathcal{D}_1, \ldots, \mathcal{D}_C\}$, a positive integer $k$, and a set of $2C$ positive integers $l_1,\ldots,l_C$ and $h_1,\ldots,h_C$, the FairHMS problem asks for a subset $S \subseteq \mathcal{D}$ of size $k$ such that $mhr(S)$ is maximized and $l_c \leq |S \cap \mathcal{D}_c| \leq h_c$ for each group $\mathcal{D}_c$. Formally,
  \begin{displaymath}
    S^* = \argmax_{S \subseteq \mathcal{D} : |S| = k} mhr(S) \; \mathrm{s.t.} \; l_c \leq |S \cap \mathcal{D}_c| \leq h_c, \forall c \in [C]
  \end{displaymath}
\end{definition}
Here, we use $S^*$ and $\mathtt{OPT} = mhr(S^*)$ to denote the optimal solution for FairHMS and its MHR on $\mathcal{D}$, respectively.

\begin{example}
  \label{ex:def}
  Let us consider the database $\mathcal{D}$ in Table~\ref{tab:example}, where the class of utility functions is defined as the linear combinations of \emph{LSAT} and \emph{GPA} scores. An HMS with $k = 2$ returns $S_0 = \{a_4, a_5\}$ with $mhr(S_0) = 0.9846$. Nevertheless, a FairHMS, for which the fairness constraint is imposed on \emph{gender} with $l_c = h_c = 1$ for $c = 1,2$, does not return $S_0$ because $S_0$ is no longer a valid solution as both tuples are from the \emph{male} group. Alternatively, the optimal solution for this FairHMS problem is $S^* = \{a_5, a_8\}$ with $mhr(S^*) = 0.9834$.
\end{example}

\noindent\textbf{Properties of FairHMS:}
Since HMS is a special case of FairHMS when $C = 1$ and $l_1 = k = h_1$, and HMS has been shown to be NP-hard~\cite{Agarwal:2017,Chester:2014,Cao:2017} in $\mathbb{R}^{d}_{+}$ when $d \geq 3$, FairHMS is also NP-hard in $\mathbb{R}^{d}_{+}$ when $d \geq 3$. Note that the happiness ratio function $hr(\cdot, \cdot)$ is monotone\footnote{A set function $f(\cdot)$ is monotone iff $f(S_1) \leq f(S_2)$ for any $S_1 \subseteq S_2 \subseteq \mathcal{D}$.} and submodular\footnote{A set function $f(\cdot)$ is submodular iff $f(S_{1}\cup \{p\})-f(S_{1}) \geq f(S_{2}\cup \{p\})-f(S_{2})$ for any $S_1 \subseteq S_2 \subseteq \mathcal{D}$ and $p \in \mathcal{D} \setminus S_{2}$.}~\cite{Qiu:2018,Xie:2020,Luenam:2021} as indicated in Lemma~\ref{lem:sub}.
\begin{lemma}[\cite{Qiu:2018}]
  \label{lem:sub}
  The happiness ratio function $hr(u, \cdot) : 2^{\mathcal{D}} \rightarrow \mathbb{R}_{+}$ is monotone and submodular for any vector $u \in \mathbb{R}^{d}_{+}$.
\end{lemma}
Despite the submodularity of $hr(\cdot, \cdot)$, the minimum happiness ratio function $mhr(\cdot)$ is monotone but not submodular because the minimum of more than one submodular function is not submodular anymore~\cite{Krause:2014}. Therefore, in subsequent sections, we will propose an exact algorithm and bicriteria approximation algorithms for FairHMS in $\mathbb{R}^{2}_{+}$ and $\mathbb{R}^{d}_{+}$ ($d \geq 2$), respectively.

\section{Exact Algorithm in 2D}
\label{sec:2d}

In this section, we present our exact algorithm called \AlgTwoD for FairHMS on two-dimensional databases by transforming a FairHMS instance into several instances of \emph{interval cover}.

\subsection{Reduction from FairHMS to Interval Cover}
\label{subsec:reduction}

We first consider the decision version of FairHMS: \emph{given a threshold $\tau \in [0, 1]$, does there exist a feasible set $S$ of size $k$ such that $S$ satisfies the group fairness constraint and $mhr(S)$ is at least $\tau$?}
Recall that $mhr(S)$ is in the range $[0, 1]$.
By solving its decision version, FairHMS is transformed into the problem of finding the largest threshold $\tau$ for which the answer to the decision problem is yes. Next, we will show how this decision problem is formulated as an interval cover problem based on the geometric properties of HMS.

For FairHMS in $\mathbb{R}^2_{+}$, the space of all nonnegative linear utility functions is essentially one-dimensional, which can be denoted by the interval $[0,1]$. We consider that each utility vector $u = (u[1], u[2])$ is rescaled to have a unit $l_1$-norm, i.e., $u[1] + u[2] =1$, and denoted with a parameter $\lambda \in [0, 1]$ as $u = (\lambda, 1-\lambda)$ by setting $\lambda = u[1]$. Given a two-dimensional point $p \in \mathbb{R}^{2}_{+}$, its score for the utility function w.r.t.~$\lambda$ is $ f_u(p) = \lambda p[1] + (1 - \lambda) p[2] = p[2] + (p[1] - p[2]) \lambda $. That is, the scores of each point for all utility functions correspond to a line segment in the plane. For example, Figure~\ref{fig:1} presents six line segments w.r.t.~points $\{p_1, \ldots, p_6\}$. Then, the four points $\{p_1, \ldots, p_4\}$ achieve the highest scores for some utility functions (or geometrically, they are vertices on the convex hull), and their line segments form the upper envelope with minimum happiness ratio $1$ (\ie no regret). Then, we define the $\tau$-fraction of the upper envelope as the $\tau$-envelope ($\tau$-$env$ for short). By definition, any point whose corresponding line segment is on or above the $\tau$-envelope at position $\lambda \in [0, 1]$ achieves a happiness ratio of at least $\tau \in (0,1)$ for utility vector $u = (\lambda, 1-\lambda)$. For a line segment $L(p)$ of a point $p$, we use $I_{\tau}(p) = \{\lambda \, | \, L_{\lambda}(p) \geq \tau\text{-}env \wedge \lambda \in [0,1]\}$ to denote the sub-interval of $[0,1]$ where $L(p)$ is on or above the $\tau$-envelope and thus point $p$ achieves happiness ratios of at least $\tau$ for the corresponding utility functions. The sub-interval $I_{\tau}(p)$ for any point $p$ and threshold $\tau$ can be computed from the intersections of $L(p)$ and the line segments in the $\tau$-envelope. Finally, the decision version of FairHMS is reduced to an interval cover problem, which aims to find a feasible subset of intervals whose corresponding points satisfy the fairness constraint to fully cover $[0,1]$. Suppose the interval cover problem can be solved optimally. In that case, the optimal solution for FairHMS can be found by performing a binary search on a sorted array of all possible values of minimum happiness ratios, which can be computed based on the results for two-dimensional RMS in~\cite{Asudeh:2017, Cao:2017}.

\begin{figure}[t]
  \centering
  \begin{subfigure}[b]{0.23\textwidth}
    \centering
    \includegraphics[height=1.25in]{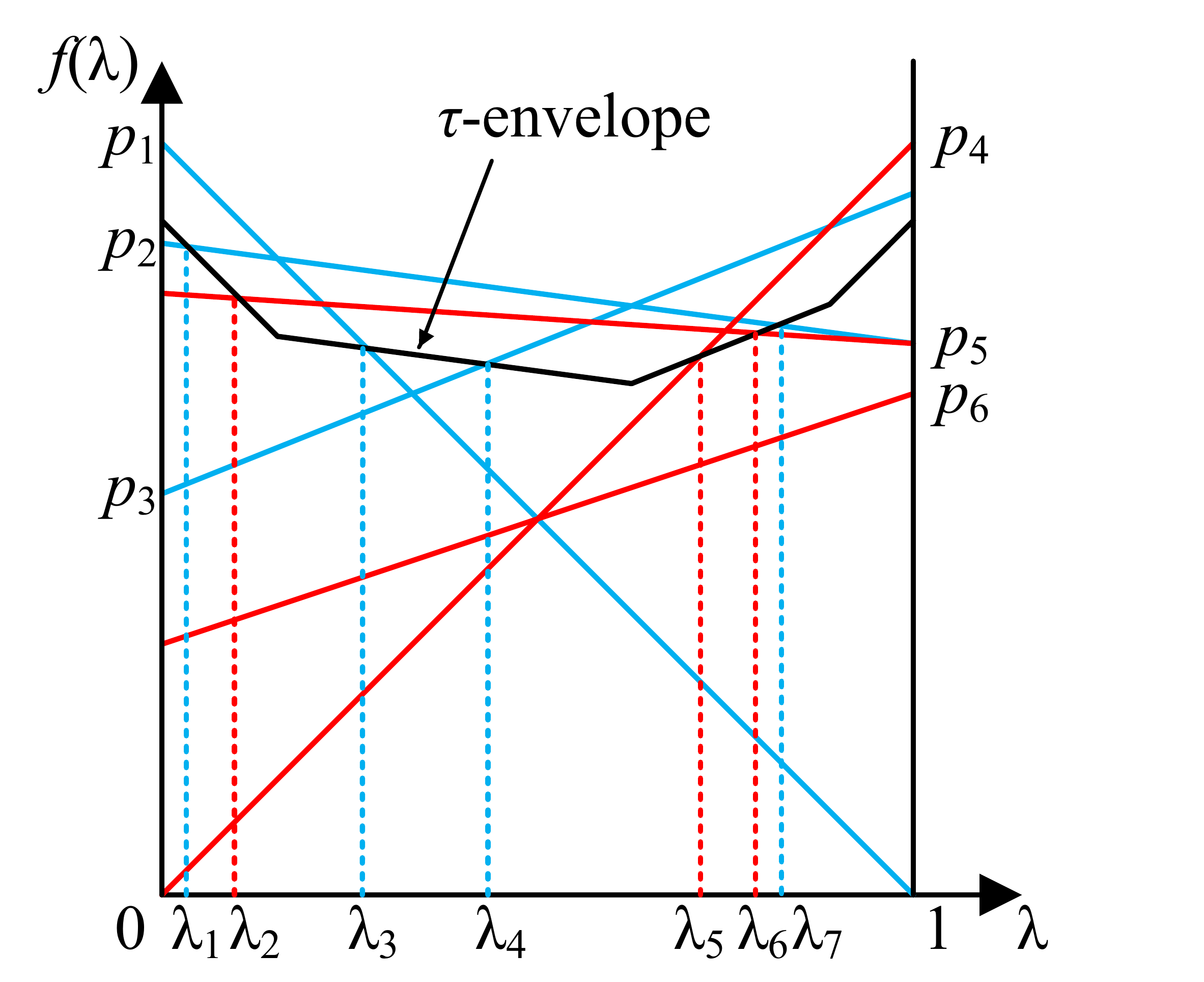}
    \caption{Transform points to intervals}
    \label{fig:2D}
  \end{subfigure}
  \hfill
  \begin{subfigure}[b]{0.24\textwidth}
    \centering
    \includegraphics[height=1.25in]{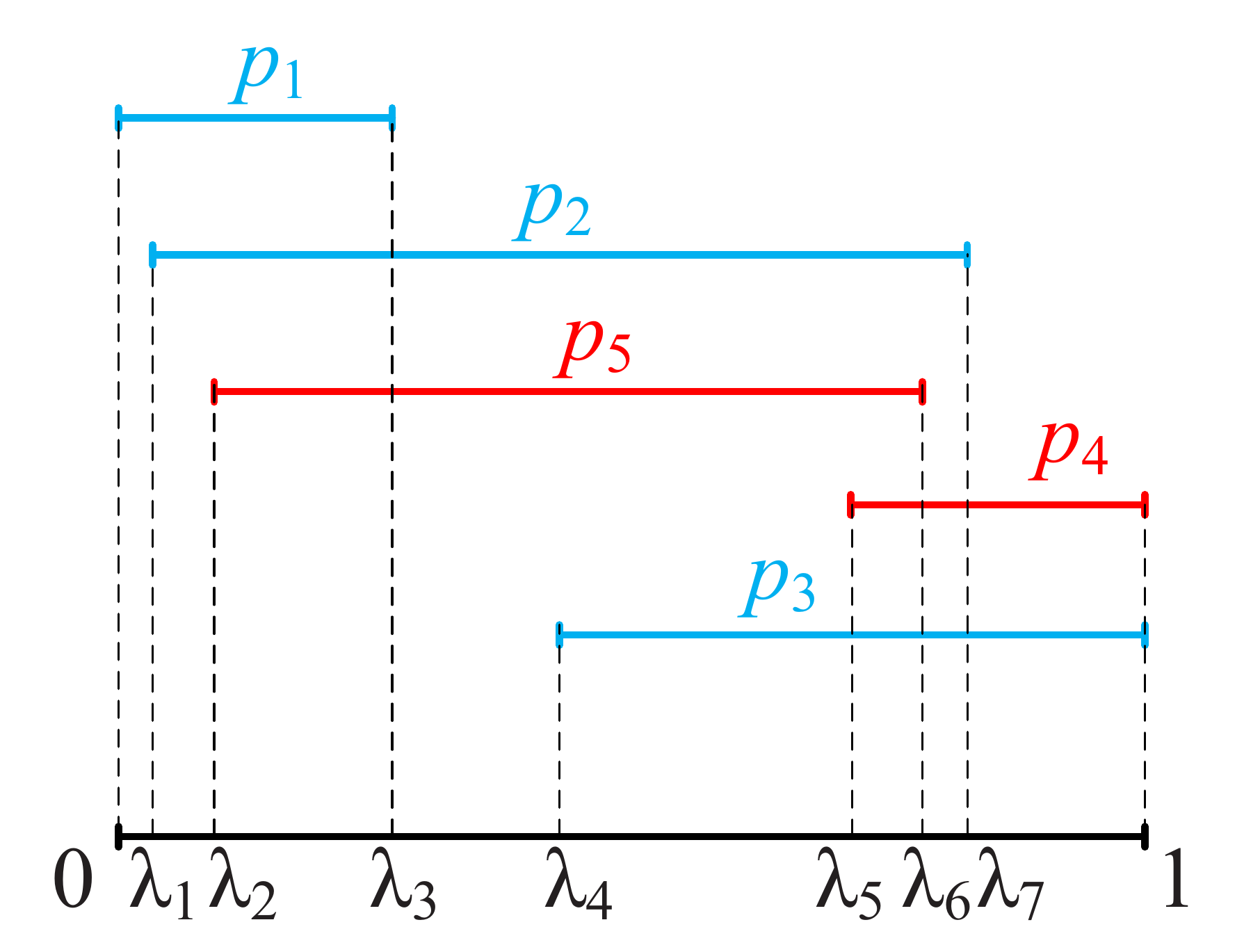}
    \caption{Reduce FairHMS to \textsc{IntCov}}
    \label{fig:2DCov}
  \end{subfigure}
  \caption{Illustration of the \AlgTwoD algorithm in $\mathbb{R}^{2}$. We show the transformation from points to intervals and draw a $\tau$-envelope w.r.t.~a minimum happiness ratio $\tau = 0.9$ in Figure~(a). We then convert the decision version of FairHMS when $\tau = 0.9$ into an interval cover problem in Figure~(b): for FairHMS with $k = 3$ and $l_c = 1$, $h_c = 2$, $S = \{p_1, p_4, p_5\}$ covers the interval $[0,1]$ and thus the answer is \emph{yes} when $\tau = 0.9$.}
  \Description{Illustration of the 2D algorithm.}
  \label{fig:1}
\end{figure}

\subsection{Exact Two-Dimensional Algorithm}
\label{subsec:2d:alg}

In Section~\ref{subsec:reduction}, we have considered the reduction from FairHMS to \emph{interval cover} and the computation of an interval for a given point $p$ and a threshold $\tau$. There are two problems remaining for an exact two-dimensional FairHMS algorithm. First, we should identify all possible values of MHRs from which the optimal one can always be found. Second, we need to determine whether a feasible interval set exists under the fairness constraint to cover the interval $[0,1]$ for a given value of MHR. Next, we present our solutions to both problems and describe our exact algorithm called \AlgTwoD for FairHMS in $\mathbb{R}^2_{+}$ in Algorithm~\ref{alg:2d}.

\begin{algorithm}[t]
  \small
  \caption{\AlgTwoD}
  \label{alg:2d}
  \KwIn{Dataset $\mathcal{D} \subseteq \mathbb{R}^{2}_{+}$ with $\mathcal{D} = \bigcup_{c=1}^{C} \mathcal{D}_c$; solution size $k \in \mathbb{Z}^{+}$; bounds $l_1, \ldots, l_C \in \mathbb{Z}^{+}$ and $h_1, \ldots, h_C \in \mathbb{Z}^{+}$}
  \KwOut{A feasible set $S^* \subseteq \mathcal{D}$ for FairHMS}
  Initialize $\mathcal{H} \gets \emptyset$\;\label{ln:2d:init}
  \ForEach{$p \in \mathcal{D}$}
  {
    Add $p[1]$ and $p[2]$ to $\mathcal{H}$\;\label{ln:2d:cand:1}
  }
  \ForEach{$p_i, p_j \in \mathcal{D}$ s.t.~$p_i \neq p_j$\label{ln:2d:cand:2:s}}
  {
    Let $u'$ be the vector $u \in \mathbb{S}^1$ s.t.~$f_u(p_i) = f_u(p_j)$\;
    \If{$u' \in \mathbb{S}^1_+$}
    {
      Add $hr(u', \{p_i, p_j\}) = \frac{f_{u'}(p_i)}{\max_{p \in \mathcal{D}} f_{u'}(p)}$ to $\mathcal{H}$\;\label{ln:2d:cand:2:t}
    }
  }
  Sort $\mathcal{H}$ ascendingly and set $l=1, h=|\mathcal{H}|$\;
  \While{$\mathcal{H}[h] > \mathcal{H}[l]$\label{ln:2d:bb}}
  {
    $cur = (l + h) / 2$ and $\tau = \mathcal{H}[cur]$\;
    Compute $I_{\tau}(p)$ for each $p \in \mathcal{D}$\;
    $S \gets$ \textsc{DynProg}$(\mathcal{I}_{\tau} = \{I_{\tau}(p) : p \in \mathcal{D}\})$\;
    \uIf{$S = \emptyset$}
    {
      $h \gets cur - 1$\;
    }
    \Else
    {
      $S^* \gets S$ and $l \gets cur + 1$\;\label{ln:2d:be}
    }
  }
  \Return{$S^*$}\;\label{ln:2d:return}
\end{algorithm}
\begin{algorithm}[t]
  \small
  \caption{\textsc{DynProg}$(\mathcal{I}_{\tau})$}
  \label{alg:dp}
  Create a stack $\mathcal{ST}$ and add a state $\mathtt{IC}[h_1, \ldots, h_C]$ to $\mathcal{ST}$\;\label{ln:dp:init}
  Initialize a state $\mathtt{IC}[0, \ldots, 0] = 0$ and set it as \emph{visited}\;\label{ln:dp:empty}
  \While{$\mathcal{ST} \neq \emptyset$}
  {
    $\mathtt{IC}[k_1, \ldots, k_C] \gets \mathcal{ST}$.top()\;
    \uIf{$\mathtt{IC}[k_1,\ldots, k_C]$ is not visited\label{ln:dp:new:s}}
    {
      Set $\mathtt{IC}[k_1, \ldots, k_C]$ as \emph{visited}\;
      $\mathcal{ST}$.push$(\mathtt{IC}[k_1,\ldots , k_c - 1, \ldots, k_C])$, $\forall c \in [C]$;\label{ln:dp:new:t}
    }
    \Else
    {
      $\mathtt{IC}[k_1, \ldots, k_C] \gets \mathcal{ST}$.pop()\;
      \If{$\mathtt{IC}[k_1, \ldots, k_{C}]$ is infeasible\label{ln:dp:inf:s}}
      {
        \textbf{continue};\label{ln:dp:inf:t}
      }
      Update $\mathtt{IC}[k_1, \ldots, k_{C}]$ according to Equation~\ref{eq:dp}\;\label{ln:dp:update}
      \If{$\mathtt{IC}[k_1, \ldots, k_{C}]=1$\label{ln:dp:yes:s}}
      {
        \Return{the set $S$ w.r.t.~$\mathtt{IC}[k_1, \ldots, k_{C}]$}\;\label{ln:dp:yes:t}
      }
    }
  }
  \Return{$\emptyset$}\;\label{ln:dp:no}
\end{algorithm}

To resolve the first problem, \ie finding all possible values of the MHR in $[0, 1]$, we are based on an important observation from~\cite[Theorem 2]{Asudeh:2017} that, for any subset $S \subseteq \mathcal{D}$, $mhr(S)$ is always equal to $hr(u, S)$ for either $u = (1, 0)$, or $u = (0, 1)$, or the vector $u$ where $f_u(p_i) = f_u(p_j)$ for each pair of points $p_i, p_j \in S$.
Therefore, we only consider the happiness ratios of each point or pair of points in all the above cases.
Specifically, we initialize an array of candidates for MHR as $\mathcal{H} = \emptyset$ (Line~\ref{ln:2d:init}).
For each $p \in \mathcal{D}$, we add $p[1]$ and $p[2]$ to $\mathcal{H}$ (Line~\ref{ln:2d:cand:1}). For each pair $\langle p_i, p_j \rangle$ of points in $\mathcal{D}$, we compute the vector $u$ such that $f_u(p_i) = f_u(p_j)$. If $u \in \mathbb{S}^{1}_{+}$, we compute the value of $hr(u, \{p_i, p_j\})$ and add it to $\mathcal{H}$ (Lines~\ref{ln:2d:cand:2:s}--\ref{ln:2d:cand:2:t}).
Then, we sort the array $\mathcal{H}$ of size $O(n^2)$ for binary search.
For each candidate $\tau$, the interval for each point is computed according to Section~\ref{subsec:reduction}, and the decision problem is solved by dynamic programming, as will be presented subsequently. Based on the answer to the decision problem, we narrow the range of the binary search by half and perform the above procedure again (Lines~\ref{ln:2d:bb}--\ref{ln:2d:be}). After the binary search is finished, the optimal solution to FairHMS is found (Line~\ref{ln:2d:return}).

To resolve the second problem (\ie \emph{fair interval cover}), we propose a dynamic programming algorithm in Algorithm~\ref{alg:dp}.
As shown in~\cite{Kleinberg:2006}, the interval cover problem can be solved optimally by a simple greedy algorithm, which starts from an interval beginning at $0$ with the rightmost ending and then greedily adds the interval whose start point is within the covered interval and end point is the rightmost until the interval $[0,1]$ is fully covered. However, the greedy algorithm is not directly applicable to our problem because it cannot guarantee the fulfillment of the fairness constraint. Thus, we propose a dynamic program based on the greedy algorithm to solve the fair interval cover problem. We define a state $\mathtt{IC}[k_1, \ldots, k_C]$ with $C$ parameters in the dynamic program, where $k_c$ denotes the number of tuples from $\mathcal{D}_c$ in the current solution, and its value is the end point of the covered interval.
We use a stack $\mathcal{ST}$ in the recursive procedure to maintain and visit the states. Initially, the state $\mathtt{IC}[h_1, \ldots, h_C]$ with the upper bounds of the fairness constraint is pushed to $\mathcal{ST}$ (Line~\ref{ln:dp:init}). First, we set the start point to $0$, \ie $\mathtt{IC}[0, \ldots, 0] = 0$ for an empty set $\emptyset$ with covered interval $[0, 0]$ (Line~\ref{ln:dp:empty}). Then, at each step, the top state is popped from $\mathcal{ST}$. A state $\mathtt{IC}[k_1, \ldots, k_C]$ is possibly transited from at most $C$ states $\mathtt{IC}[k_1, \ldots, k_c - 1, \ldots, k_C]$ for each $c \in [C]$. If $\mathtt{IC}[k_1, \ldots, k_C]$ is neither the initial state nor visited, its $C$ predecessors are pushed to $\mathcal{ST}$ (Lines~\ref{ln:dp:new:s}--\ref{ln:dp:new:t}).
Otherwise, we compute the value of $\mathtt{IC}[k_1, \ldots, k_C]$ according to its predecessors by the greedy strategy (Line~\ref{ln:dp:update}).
If there is an interval $I_{\tau}(p)$ for some $c \in [C]$ and $p \in \mathcal{D}_c$ whose start point is at most $\mathtt{IC}[k_1, \ldots, k_c - 1, \ldots, k_C]$ and end point is the rightmost, then $\mathtt{IC}[k_1, \ldots, k_c, \ldots, k_C]$ will be set to the end point of $I_{\tau}(p)$, \ie $p$ is added to the solution.
Formally,
\begin{multline}\label{eq:dp}
  \mathtt{IC}[k_1, \ldots, k_c, \ldots, k_C] = \\ \max_{c \in [C]} \Big\{ \max_{p \in \mathcal{D}_c \; : \; I^{-}_{\tau}(p) \leq \mathtt{IC}[k_1, \ldots, k_c - 1, \ldots, k_C]} I^{+}_{\tau}(p)\Big\}
\end{multline}
where $I^{-}_{\tau}(p)$ and $I^{+}_{\tau}(p)$ denote the left and right bounds of $I_{\tau}(p)$, respectively.
Moreover, a state $\mathtt{IC}[k_1, \ldots, k_C]$ is infeasible if $\sum_{c=1}^{C}$ $\max(l_c, k_c) > k$. Such a state is skipped directly because it cannot provide any feasible solution (Lines~\ref{ln:dp:inf:s}--\ref{ln:dp:inf:t}). The recursion procedure will terminate if either $\mathtt{IC}[k_1, \ldots, k_C] = 1$, \ie a feasible solution that covers $[0,1]$ is found and the decision problem is answered by ``\emph{yes}'' (Lines~\ref{ln:dp:yes:s}--\ref{ln:dp:yes:t}) or $\mathcal{ST}$ is empty and the decision problem is answered by ``\emph{no}'' (Line~\ref{ln:dp:no}).

Next, we will explain why our \AlgTwoD algorithm is optimal for the FairHMS problem.
\begin{theorem}
  \textnormal{\AlgTwoD} provides an optimal solution to FairHMS.
\end{theorem}  
\begin{proof}
  First, \textsc{DynProg} is optimal for the fair interval cover problem because (1) the recursion procedure in the dynamic program has visited all feasible group permutations to form a fair solution, and (2) the greedy strategy starting from $0$ ensures the optimality of the solution found for every group permutation.
  Then, the reduction from the decision version of FairHMS to the fair interval cover problem guarantees that they must have the same answer (either ``\emph{yes}'' or ``\emph{no}'') for the same value of $\tau \in [0, 1]$.
  Finally, following the analysis of~\cite{Asudeh:2017, Cao:2017}, the array $\mathcal{H}$ has included all possible values of the optimal MHR.
  Given all the above results, we conclude that the output $S^*$ of \AlgTwoD is optimal for FairHMS.
\end{proof}

\noindent\textbf{Complexity Analysis:} In Algorithm~\ref{alg:2d}, finding all candidates of MHR takes $O(n^2 \log{n})$ time. Then, the binary search performs $O(\log{n})$ iterations to find the largest $\tau$. For each candidate $\tau$, computing the intervals for points in $\mathcal{D}$ takes $O(n \log{n})$ time. Algorithm~\ref{alg:dp} has at most $\prod_{c = 1}^{C}{(1 + h_c)}$ states and the time to compute the value of each state is $O(n)$. Since it must hold that $h_c \leq k$, Algorithm~\ref{alg:dp} runs in $O(n (k^C + \log{n}))$ time. To sum up, the time complexity of Algorithm~\ref{alg:2d} is $O\big(n \log{n} (n + k^C)\big)$.

\section{Algorithms in MD}
\label{sec:algorithm}

The FairHMS problem becomes more challenging in three and higher dimensions for its NP-hardness. Thus, we focus on efficient approximation algorithms for FairHMS on multi-dimensional datasets. Subsequently, we first introduce the background on $\delta$-net~\cite{Agarwal:2017,Kumar:2018,Wang:2021b} and the transformation from FairHMS to multi-objective submodular maximization (MOSM)~\cite{Krause:2008,Anari:2019,DBLP:conf/nips/Udwani18} under a matroid constraint in Section~\ref{subsec:background}. We then present our basic algorithm \AlgBG for FairHMS based on existing methods for MOSM in Section~\ref{subsec:bigreedy}. We further propose an adaptive sampling strategy to improve the practical efficiency of \AlgBG in Section~\ref{subsec:imp}.
Note that the proofs of all lemmas are deferred to Appendix~\ref{app:proofs}.

\subsection{Background: \texorpdfstring{$\delta$}{delta}-Net and Multi-Objective Submodular Maximization}
\label{subsec:background}

The transformation from RMS to a hitting-set or set-cover problem based on the notion of $\delta$-net in~\cite{Agarwal:2017,Kumar:2018,Wang:2021b} can be similarly adapted to HMS. As shown in Section~\ref{sec:problem}, the space of all utility vectors after the $l_2$-normalization can be denoted as the $(d-1)$-dimensional unit sphere $\mathbb{S}^{d-1}_{+} = \{ u \in \mathbb{R}_{+}^{d} : ||u|| = 1\}$.
Given a parameter $\delta \in (0,1)$, a set $\mathcal{N} \subset \mathbb{S}^{d-1}_{+}$ is called a $\delta$-net~\cite{Agarwal:2017} of $\mathbb{S}^{d-1}_{+}$ iff there exists a vector $v \in \mathcal{N}$ with $\langle u, v \rangle \geq \cos\delta$ for any vector $u \in \mathbb{S}^{d-1}_{+}$, where $\langle \cdot, \cdot \rangle$ is the dot product function.
Intuitively, the idea of $\delta$-net is to approximate an infinite number of utility vectors in continuous space with a finite number of representative vectors such that the errors in terms of \emph{angular distances} (quantified by $\cos \delta$) are bounded.
A $\delta$-net of size $m = O(\frac{1}{\delta^{d-1}})$ can be computed by drawing a ``uniform'' grid on $\mathbb{S}^{d-1}_{+}$. A more widely used method~\cite{Saff:1997} to compute a $\delta$-net is to sample a set of $m = O(\frac{1}{\delta^{d-1}}\log{\frac{1}{\delta}})$ vectors uniformly at random on $\mathbb{S}^{d-1}_{+}$, which will be a $\delta$-net with probability at least $1/2$. Note that the success probability of $\delta$-net computation can be arbitrarily high with repeated trials.
Figure~\ref{fig:deltaNet} illustrates a $\frac{\pi}{16}$-net in 2D by drawing $5$ vectors uniformly on $\mathbb{S}^{1}_{+}$.
We adopt the random sampling method for $\delta$-net computation in our implementation. Let us define the minimum happiness ratio (MHR) of a subset $S$ over $\mathcal{D}$ on a $\delta$-net $\mathcal{N}$ of vectors as $mhr(S | \mathcal{N}) = \min_{u \in \mathcal{N}} hr(u, S)$. We have the following lemma to indicate that $mhr(S | \mathcal{N})$ provides a lower bound of $mhr(S)$ with an error of at most $\frac{2 \delta d}{1 + \delta d}$.

\begin{figure}[t]
  \centering
  \includegraphics[width=0.24\textwidth]{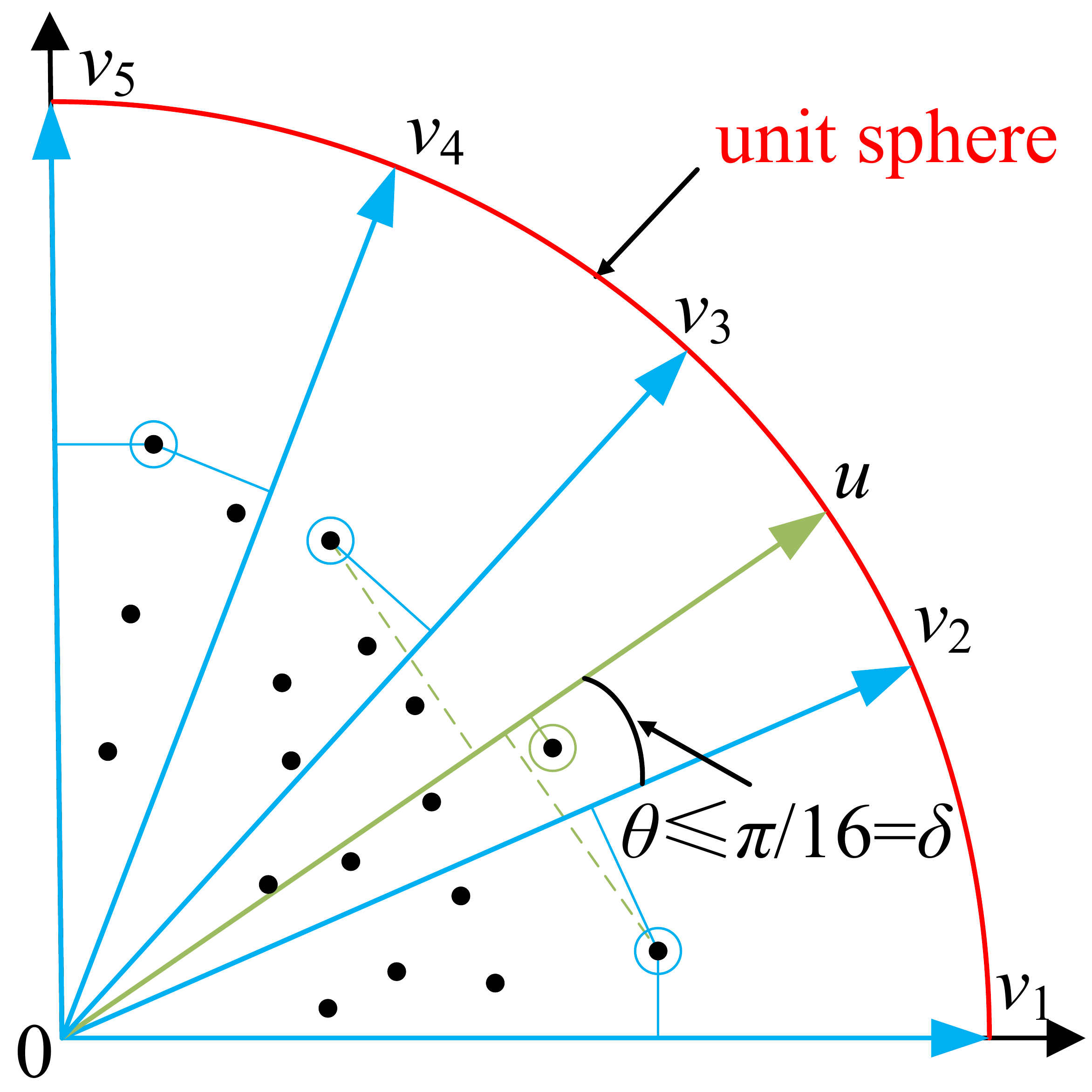}
  \caption{Illustration of a $\frac{\pi}{16}$-net $\mathcal{N}$ comprised of $5$ uniform vectors $v_1,\ldots,v_5$ in 2D. For any vector $u \in \mathbb{S}^{1}_{+}$, there always exists a vector (\eg $v_2$) in $\mathcal{N}$ with $\langle u,v_2 \rangle \geq \cos\frac{\pi}{16}$.}
  \label{fig:deltaNet}
  \Description{Example for delta-net}
\end{figure}

\begin{lemma}
\label{lm:delta:net}
  Given a $\delta$-net $\mathcal{N} \subset \mathbb{S}^{d-1}_{+}$, a database $\mathcal{D}$, and a subset $S \subseteq \mathcal{D}$, it holds that $mhr(S) \leq mhr(S | \mathcal{N}) \leq mhr(S) + \frac{2 \delta d}{1 + \delta d}$.
\end{lemma}
According to Lemma~\ref{lm:delta:net}, when $mhr(S)$ is estimated on a $\frac{\delta}{d(2 - \delta)}$-net $\mathcal{N} \subset \mathbb{S}^{d-1}_{+}$, it holds that $mhr(S) \leq mhr(S | \mathcal{N}) \leq mhr(S) + \delta$. In this way, we find an approximate solution to FairHMS by reducing it to the problem of finding a subset to maximize the MHR on a $\frac{\delta}{d(2 - \delta)}$-net $\mathcal{N}$. However, although the reduced FairHMS problem is defined on a finite number of utility functions instead of an infinite utility space, it is still an NP-hard problem even without fairness constraints, as indicated in Lemma~\ref{lm:inapprox}.
\begin{lemma}
\label{lm:inapprox}
There does not exist any polynomial-time algorithm to approximate the reduced FairHMS problem defined on a set $\mathcal{N}$ of vectors in $\mathbb{S}^{d-1}_{+}$ with a factor of $(1-\varepsilon) \cdot \log{m}$, where $m = |\mathcal{N}|$, for any parameter $\varepsilon > 0$ unless P=NP.
\end{lemma}
Due to the inapproximability result of reduced FairHMS, we will focus on proposing approximation algorithms within a nearly best possible factor. In Section~\ref{sec:problem}, we have shown that the happiness ratio function $hr(u, \cdot)$ is monotone and submodular for any vector $u$ and the fairness constraint is a special case of \emph{matroid constraints}. Therefore, the reduced FairHMS problem is an instance of maximizing the minimum of $m$ monotone submodular functions, where $m = |\mathcal{N}|$, known as the \emph{multi-objective submodular maximization} (MOSM) in the literature~\cite{Krause:2008,Anari:2019,DBLP:conf/nips/Udwani18,Torrico:2021}, under a matroid constraint. Although MOSM is generally inapproximable~\cite{Krause:2008}, several bicriteria approximation algorithms have been proposed for MOSM under a cardinality or matroid constraint~\cite{Krause:2008, Anari:2019, Torrico:2021}. Here, an algorithm is called $(\alpha, \beta)$-bicriteria approximate for $\alpha,\beta > 0$ if it provides a solution $S$ of size $\alpha k$ such that $f(S) \geq \beta f(S^*)$, where $f$ is the minimum of $m$ monotone submodular functions and $S^*$ is the optimal size-$k$ solution for maximizing $f$. Next, we will generalize an existing bicriteria approximation algorithm for MOSM under a matroid constraint to FairHMS and analyze it theoretically.

The basic idea of existing algorithms is to transform the non-submodular objective function into a monotone submodular function by introducing a capped value $\tau > 0$~\cite{Fujito:2000, Krause:2008, Anari:2019}. Specifically, in our FairHMS problem the function $mhr(\cdot | \mathcal{N})$ is not submodular. Nevertheless, it can be transformed into a monotone submodular function as follows: First, we define a truncated happiness ratio function $hr_{\tau}(u,S) = \min\{hr(u, S), \tau\}$ for some capped value $\tau \in (0, 1)$. As proven by \citet{Krause:2008}, for any monotone submodular function, its truncated function is monotone and submodular. Second, we define a truncated MHR function as:
\begin{equation}
\label{eq:mhr}
  mhr_{\tau}(S|\mathcal{N}) = \frac{1}{m} \sum_{u \in \mathcal{N}} hr_{\tau}(u,S)
\end{equation}
where $m = |\mathcal{N}|$.
Lemma~\ref{lem:mhrsub} proves the soundness of truncation.
\begin{lemma}
\label{lem:mhrsub}
  $mhr_{\tau}(S|\mathcal{N})$ is a monotone submodular function for any capped value $\tau \in (0, 1)$.
\end{lemma}
The proof of Lemma~\ref{lem:mhrsub} is trivial since $mhr_{\tau}(\cdot|\mathcal{N})$ is a nonnegative linear combination of $m$ monotone submodular functions~\cite{Krause:2014}.
Then, the reduced FairHMS problem can be regarded as a submodular maximization problem under a matroid constraint when the truncated MHR function is considered the objective function. Nevertheless, we still need to show the relationship between the truncated and original MHR functions in Lemma~\ref{lem:capped}.
\begin{lemma}
  \label{lem:capped}
  $mhr(S|\mathcal{N}) \geq \tau$ if and only if $mhr_{\tau}(S|\mathcal{N}) = \tau$.
\end{lemma}
Lemma~\ref{lem:capped} indicates that the truncated and original MHR functions are equivalent only when the capped value $\tau$ is achieved. And the remaining problem becomes how to find the largest $\tau$ for which there is a feasible solution $S^*$ under the fairness constraint with $mhr_{\tau}(S^*|\mathcal{N}) = \tau$. In our algorithm, we attempt different values in the range $[0, 1]$ so that at least one of them is near optimal with bounded error. Furthermore, to find a fair solution $S$, we adopt the greedy algorithm~\cite{Fisher:1978} for submodular maximization under a matroid constraint. However, the greedy algorithm is $\frac{1}{2}$-approximate, it only guarantees to find a feasible solution $S$ with $mhr_{\tau^*}(S|\mathcal{N}) \geq \frac{\tau^*}{2}$ for the optimal $\tau^*$ and cannot achieve any bound on $mhr(S|\mathcal{N})$. A solution to this problem is to run the greedy algorithm in $\gamma$ rounds and to take the union of the solutions $S_1, \ldots, S_{\gamma}$ in all $\gamma$ rounds as the final solution $S = \bigcup_{i=1}^{\gamma} S_i$. It will lead to a solution $S$ of size $\gamma k$ that satisfies the loosened fairness constraint where all lower and upper bounds are scaled by $\gamma$ following the result in Lemma~\ref{lem:greedy}.
\begin{lemma} 
\label{lem:greedy}
  Let $\tau^*$ be the largest $\tau$ for which a feasible solution $S^*$ with $mhr_{\tau}(S^*|\mathcal{N}) = \tau$ exists and $S_i$ be the solution returned by the greedy algorithm at the $i$-th round on $\mathcal{D} \setminus \bigcup_{j=1}^{i-1} S_j$ for function $mhr_{\tau^*}(\cdot|\mathcal{N})$. If $\gamma \geq \lceil \log_{2}\frac{m}{\varepsilon} \rceil$, then, for $S = \bigcup_{i=1}^{\gamma} S_i$, it holds that $mhr(S|\mathcal{N}) \geq (1-\varepsilon) \cdot \tau^*$.
\end{lemma}

\subsection{The BiGreedy Algorithm}
\label{subsec:bigreedy}

\begin{algorithm}[tb]
  \small
  \caption{\textsc{BiGreedy}}
  \label{alg:bicriteria}
  \KwIn{Dataset $\mathcal{D} \subseteq \mathbb{R}^{d}_{+}$ with $\mathcal{D} = \bigcup_{c=1}^{C} \mathcal{D}_c$; solution size $k \in \mathbb{Z}^{+}$; matroid $\mathcal{M} = (\mathcal{D}, \mathcal{I})$; parameters $\delta, \varepsilon \in (0, 1)$}
  \KwOut{A feasible set $S' \subseteq \mathcal{D}$ for FairHMS}
  Sample a $\frac{\delta}{d(2 - \delta)}$-net $\mathcal{N}$ from $\mathbb{S}^{d-1}_{+}$\;\label{ln:bg:delta}
  Set $\gamma \gets \lceil \log_2{\frac{2m}{\varepsilon}} \rceil$ where $m = |\mathcal{N}|$\;\label{ln:bg:gamma}
  Initialize $\tau \gets 1$ and $ \mathcal{S} \gets \emptyset $\;\label{ln:bg:tau:s}
  \While{$\tau \geq \frac{1}{m}$}
  {
	$S \gets$ \textsc{MRGreedy}$(\tau, \gamma, \mathcal{N}, \mathcal{D})$\;
	\If{$S \neq \emptyset$}
	{
	  $\mathcal{S} \gets \mathcal{S} \cup \{S\}$\;
	}
	$\tau \gets (1 - \frac{\varepsilon}{2}) \cdot \tau$\;\label{ln:bg:tau:t}
  }
  \Return{$S' \gets \argmax_{S \in \mathcal{S}} mhr(S)$}\;\label{ln:bg:sol}
  \tcc{Multi-Round Greedy Algorithm}
  \Fn{\textnormal{\textsc{MRGreedy}($\tau, \gamma, \mathcal{N}, \mathcal{D}$)}\label{ln:bg:mrg:s}}
  {
    Initialize $S \gets \emptyset$ and $\mathcal{D}_0 \gets \mathcal{D}$\;
    \For{$i \gets 1, \ldots, \gamma$}
    {
      Initialize $S_i \gets \emptyset$ and $\mathcal{D}_i \gets \mathcal{D}_{i - 1}$\;
      \While{there exists $p \in \mathcal{D}_i$ s.t.~$S_i \cup \{p\} \in \mathcal{I}$\label{ln:mrg:greedy:s}}
      {
        $p^* \gets \arg\max_{p \in \mathcal{D}_i \,:\, S_i \cup \{p\} \in \mathcal{I}} \Delta(p, mhr_{\tau}(S_i|\mathcal{N}))$\;
        $S_i \gets S_i \cup \{p^*\}$\;\label{ln:mrg:greedy:t}
      }
      $S \gets S \cup S_i$ and $\mathcal{D}_i \gets \mathcal{D}_i \setminus S_i$\;\label{ln:mrg:update:sol}
      \If{$mhr_{\tau}(S|\mathcal{N}) \geq (1 - \frac{\varepsilon}{2m}) \cdot \tau$}
      {
        \textbf{break}\;\label{ln:mrg:early}
      }
    }
    \If{$mhr_{\tau}(S|\mathcal{N}) < (1 - \frac{\varepsilon}{2m}) \cdot \tau$}
    {
      \Return{$\emptyset$}\;\label{ln:mrg:empty}
    }
    \Return{$S$}\;\label{ln:bg:mrg:t}
  }
\end{algorithm}

In Lemmas~\ref{lm:delta:net}--\ref{lem:greedy}, we show that the FairHMS problem can be solved as a submodular maximization problem under a matroid constraint with a bicriteria approximation guarantee, for which a multi-round greedy procedure can be used for solution computation.
Accordingly, based on the above theoretical results, we propose the \textsc{BiGreedy} algorithm for FairHMS in Algorithm~\ref{alg:bicriteria} algorithm.
The \textsc{BiGreedy} algorithm uses two parameters $\delta, \varepsilon \in (0, 1)$ as input to control the errors led by sampling utility vectors and searching appropriate capped values, respectively.
Initially, we sample a $\frac{\delta}{d(2 - \delta)}$-net $\mathcal{N}$ of $m$ utility vectors from $\mathbb{S}^{d-1}_{+}$ so that the error in the MHR estimation is bounded by $\delta$ (Line~\ref{ln:bg:delta}), as indicated by Lemma~\ref{lm:delta:net}.
Then, we determine that the maximum number $\gamma$ of rounds is $\lceil \log_2{\frac{2m}{\varepsilon}} \rceil$ to ensure that the error in the capped value is at most $\frac{\varepsilon}{2}$ (Line~\ref{ln:bg:gamma}), as shown in Lemma~\ref{lem:greedy}.
Next, we attempt different capped values $\tau$ in the range $[\frac{1}{m}, 1]$ so that at least one of them is within $[(1-\frac{\varepsilon}{2})\cdot \tau^*, \tau^*]$ (Lines~\ref{ln:bg:tau:s}--\ref{ln:bg:tau:t}), where $\tau^*$ is the optimal capped value.
For each capped value $\tau$, we run the multi-round greedy (\textsc{MRGreedy}) algorithm in Lines~\ref{ln:bg:mrg:s}--\ref{ln:bg:mrg:t} to find a solution $S$ for FairHMS.
The \textsc{MRGreedy} algorithm has at most $\gamma$ rounds and uses the greedy algorithm for monotone submodular maximization under a matroid constraint~\cite{Fisher:1978} as a subroutine to compute a solution $S_i$ in the $i$-th round. At each iteration of the greedy algorithm, the point $p^*$ that achieves the largest marginal gain $\Delta(p, mhr_{\tau}(S_i|\mathcal{N})) = mhr_{\tau}(S_i \cup \{p\}|\mathcal{N}) - mhr_{\tau}(S_i|\mathcal{N}) $ w.r.t.~the truncated MHR function and guarantees the fulfillment of the matroid constraint for group fairness in Section~\ref{sec:problem}, \ie $S_i \cup \{p\} \in \mathcal{I}$, is added to $S_i$ in the $i$-th round (Lines~\ref{ln:mrg:greedy:s}--\ref{ln:mrg:greedy:t}). After round $i$, $S_i$ is added to the solution $S$ of \textsc{MRGreedy} (Line~\ref{ln:mrg:update:sol}). Then, we check whether $mhr_{\tau}(S|\mathcal{N}) \geq (1 - \frac{\varepsilon}{2m}) \cdot \tau$: if yes, we will return $S$ immediately without further rounds (Line~\ref{ln:mrg:early}); otherwise, we will continue to perform the next round. If the condition is still not satisfied after all $\gamma$ rounds, an empty set will be returned by \textsc{MRGreedy} (Line~\ref{ln:mrg:empty}). Finally, we keep all nonempty solutions returned by \textsc{MRGreedy} for different capped values in $\mathcal{S}$ and pick the best one with the highest $mhr(S)$ among them as the final solution $S'$ to FairHMS (Line~\ref{ln:bg:sol}).

Next, we analyze the approximation ratio and time complexity of the \textsc{BiGreedy} algorithm in Theorem~\ref{thm:bicriteria}.
\begin{theorem}
\label{thm:bicriteria}
  The \textnormal{\textsc{BiGreedy}} algorithm is an $\big(O(d \log{\frac{1}{\delta\varepsilon}}), 1 - \varepsilon - \frac{\delta}{\mathtt{OPT}}\big)$-bicriteria approximation algorithm for the FairHMS problem, where $\mathtt{OPT}$ is the minimum happiness ratio of the optimal solution to FairHMS, running in $O\big(n k \varepsilon^{-2} \delta^{-d} \log^2{\frac{1}{\delta}}\big)$ time.
\end{theorem}
\begin{proof}
  First of all, we have $|S| \leq \lceil \log_2{\frac{2m}{\varepsilon}} \rceil \cdot k$ and $\gamma' \cdot l_c \leq |S \cap \mathcal{D}_c| \leq \gamma' \cdot h_c$ for each $c \in [C]$ because $|S_i| = k$ and $l_c \leq |S_i \cap \mathcal{D}_c| \leq h_c$ for each $S_i$ in the $i$-th round and $\gamma = \lceil \log_2{\frac{2m}{\varepsilon}} \rceil$, where $\gamma' \leq \gamma$ is the practical number of rounds for the greedy algorithm. Then, let $S^*$ and $\mathtt{OPT}$ be the optimal solution to FairHMS and its minimum happiness ratio, respectively. According to Lemma~\ref{lm:delta:net}, since $\mathcal{N}$ is a $\frac{\delta}{d(2 - \delta)}$-net, we have $\mathtt{OPT} \leq mhr(S^* | \mathcal{N}) \leq \mathtt{OPT} + \delta$. No matter what is $mhr(S^* | \mathcal{N})$, there always exists a capped value $\tau$ in Algorithm~\ref{alg:bicriteria} such that $\tau \in [(1-\frac{\varepsilon}{2})\cdot mhr(S^* | \mathcal{N}), mhr(S^* | \mathcal{N})]$. For such a capped value $\tau$, its corresponding solution $S$ returned by \textsc{MRGreedy} satisfies that $mhr(S|\mathcal{N}) \geq (1-\frac{\varepsilon}{2}) \cdot \tau$ by Lemma~\ref{lem:greedy}. Therefore, we have:
  \begin{align*}
    mhr(S) & \geq mhr(S|\mathcal{N}) - \delta \geq (1-\varepsilon) \cdot mhr(S^* | \mathcal{N}) - \delta \\
           & \geq (1-\varepsilon) \cdot \mathtt{OPT} - \delta = (1 - \varepsilon - \frac{\delta}{\mathtt{OPT}}) \cdot \mathtt{OPT}
  \end{align*}
  Since $m = O(\delta^{-d})$, we prove that the \textsc{BiGreedy} algorithm is an $\big(O(d \log{\frac{1}{\delta\varepsilon}}), 1 - \varepsilon - \frac{\delta}{\mathtt{OPT}}\big)$-bicriteria approximation algorithm for the FairHMS problem.
  
  In terms of time complexity, there are at most $O(\frac{\log m}{\varepsilon})$ different values of $\tau$. For each value of $\tau$, \textsc{MRGreedy} runs in $O(n m k \log{\frac{m}{\varepsilon}})$ time because it takes $O(nm)$ time to find $p^*$ per iteration and has at most $k$ iterations in $O(\log{\frac{m}{\varepsilon}})$ rounds. Therefore, the overall time complexity of \textsc{BiGreedy} is $O\big(n k \varepsilon^{-2} \delta^{-d} \log^2{\frac{1}{\delta}}\big)$.
\end{proof}
Theorem~\ref{thm:bicriteria} indicates that the solution of \textsc{BiGreedy} is near-optimal for FairHMS when $\varepsilon$ and $\delta$ are arbitrarily small. Compared with the best possible approximation ratio, i.e., $\log{m} = O(d\log{\frac{1}{\delta}})$, the approximation ratio of \textsc{BiGreedy} is lowered by a factor of $O(\log{\frac{1}{\varepsilon}})$.
Finally, to ensure that $S'$ returned by \textsc{BiGreedy} is feasible for FairHMS (\ie $|S'| = k$ and $S' \in \mathcal{I}$), we replace $k$ with $k' = \frac{k}{\lceil \log_2{\frac{2m}{\varepsilon}} \rceil}$ in Algorithm~\ref{alg:bicriteria}. If the rounding errors are ignored, \textsc{BiGreedy} with input $k'$ returns a feasible solution to FairHMS whose MHR is close to $\mathtt{OPT}'$, where $\mathtt{OPT}'$ is the optimal MHR for solution size $k'$.

\begin{algorithm}[tb]
  \small
  \caption{\textsc{BiGreedy+}}
  \label{alg:impl:bicriteria}
  \KwIn{Dataset $\mathcal{D} \subseteq \mathbb{R}^{d}_{+}$ with $\mathcal{D} = \bigcup_{c=1}^{C} \mathcal{D}_c$; solution size $k \in \mathbb{Z}^{+}$; matroid $\mathcal{M} = (\mathcal{D}, \mathcal{I})$; parameters $\lambda, \varepsilon \in (0, 1)$}
  \KwOut{A feasible set $S' \subseteq \mathcal{D}$ for FairHMS}
  Draw a set $\mathcal{N}_0$ of $m_0$ vectors from $\mathbb{S}^{d-1}_{+}$\;\label{ln:bgp:init:1}
  Run \textsc{BiGreedy} on $\mathcal{N}_0$ to get $S'_0$ with capped value $\tau_0$\;\label{ln:bgp:init:2}
  \For{$i \gets 1, \ldots, \lceil \log_2{\frac{M}{m_0}} \rceil$}
  {
	Draw a set $\mathcal{N}_i$ of $m_i = 2 m_{i-1}$ vectors from $\mathbb{S}^{d-1}_{+}$\;\label{ln:bgp:iter:1}
	Run \textsc{BiGreedy} on $\mathcal{N}_i$ to get $S'_i$ with capped value $\tau_i$\;\label{ln:bgp:iter:2}
	\If{$\tau_{i - 1} - \tau_{i} < \lambda$\label{ln:bgp:cond:s}}
	{
	  \textbf{break}\;\label{ln:bgp:cond:t}
	}
  }
  \Return{$S' \gets \argmax_{i} mhr(S'_i)$}\;\label{ln:bgp:sol}
\end{algorithm}

\subsection{Adaptive Sampling}
\label{subsec:imp}

In this subsection, we consider an adaptive sampling strategy for $\delta$-net computation to improve the performance of the \textsc{BiGreedy} algorithm in practice and propose the \textsc{BiGreedy+} algorithm in Algorithm~\ref{alg:impl:bicriteria} accordingly.
The main reason why \textsc{BiGreedy} is not efficient in practice is the huge number $m = O(\delta^{-d})$ of vectors to form a $\frac{\delta}{d(2 - \delta)}$-net $\mathcal{N}$, particularly so when $d$ is high.
However, such a large sample size is often not necessary in practice. Thus, we propose an adaptive method to reduce the sample size. The basic idea is to initialize with a small size $m_0$ (Line~\ref{ln:bgp:init:1}) and run \textsc{BiGreedy} on a set $\mathcal{N}_0$ of vectors ($m_0 = |\mathcal{N}_0|$) to get a solution $S'_0$ with capped value $\tau_0$ (Line~\ref{ln:bgp:init:2}). Then, we double the sample size $m_i = 2m_{i-1}$ and run \textsc{BiGreedy} again to get $S'_i$ with capped value $\tau_i$ (Lines~\ref{ln:bgp:iter:1}--\ref{ln:bgp:iter:2}). If the capped values $\tau_i$ and $\tau_{i - 1}$ in two consecutive rounds are close to each other (\ie $\tau_{i - 1} - \tau_i < \lambda$), the sample size will be regarded as large enough; otherwise, we will further increase $m$ and run \textsc{BiGreedy} again until $m$ exceeds a predefined limit $M$ (Lines~\ref{ln:bgp:cond:s}--\ref{ln:bgp:cond:t}), \eg $M = O(\delta^{-d})$. Finally, the solution with the largest minimum happiness ratio among the ones computed in different rounds is returned as the final solution $S'$ (Line~\ref{ln:bgp:sol}). The worst-case running time of \textsc{BiGreedy+} is the same as \textsc{BiGreedy} when $M = O(\delta^{-d})$. But its practical efficiency is significantly higher than \textsc{BiGreedy} because it often terminates with much smaller $m$. At the same time, its solution quality is close to that of \textsc{BiGreedy} in most cases, as will be shown in the experiments.

\section{Experiments}
\label{sec:experiments}

In this section, we evaluate the performance of our proposed algorithms on synthetic and real-world datasets. We first introduce the experimental setup in Section~\ref{ssec:exp:seutp}. Then, the experimental results are reported in Section~\ref{ssec:exp:results}.

\begin{figure*}[t]
    \centering
    \includegraphics[height=0.12in]{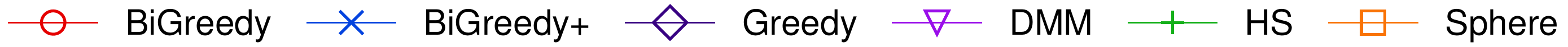}
    \\
    \begin{subfigure}[b]{0.19\textwidth}
        \centering
        \includegraphics[width=\textwidth]{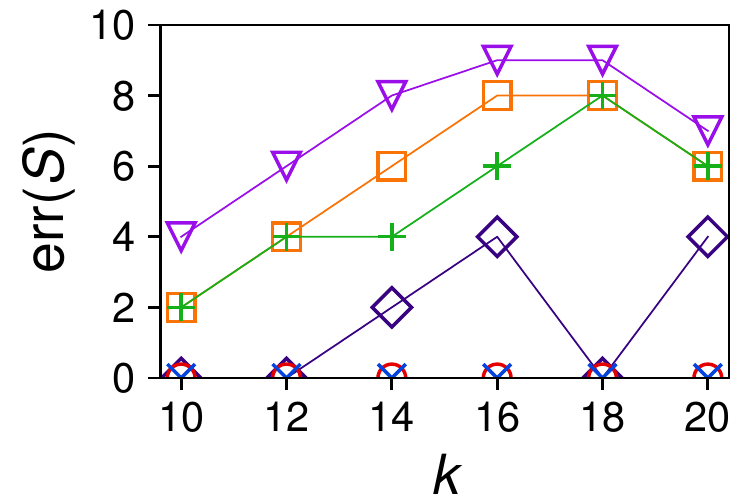}
        \caption{Adult (Gender)}
    \end{subfigure}
    \hfill
    \begin{subfigure}[b]{0.19\textwidth}
        \centering
        \includegraphics[width=\textwidth]{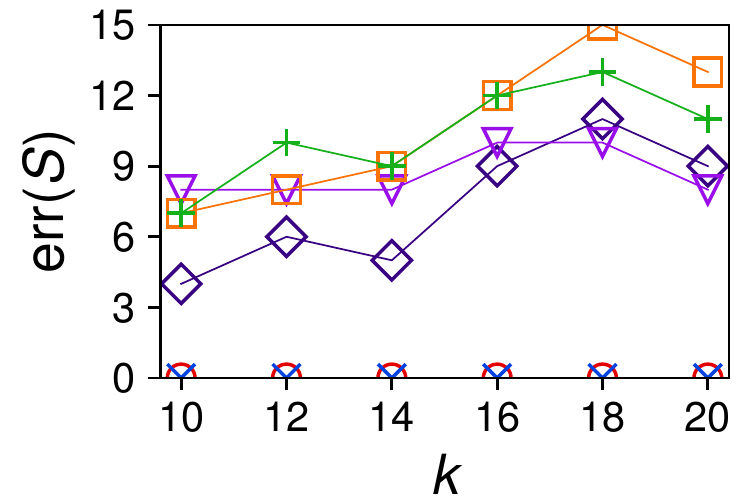}
        \caption{Adult (Race)}
    \end{subfigure}
    \hfill
    \begin{subfigure}[b]{0.19\textwidth}
        \centering
        \includegraphics[width=\textwidth]{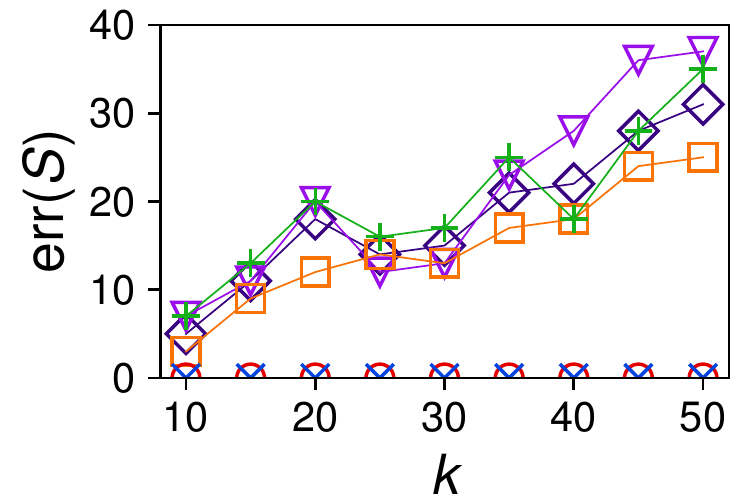}
        \caption{AntiCor\_6D}
    \end{subfigure}
    \hfill
    \begin{subfigure}[b]{0.19\textwidth}
        \centering
        \includegraphics[width=\textwidth]{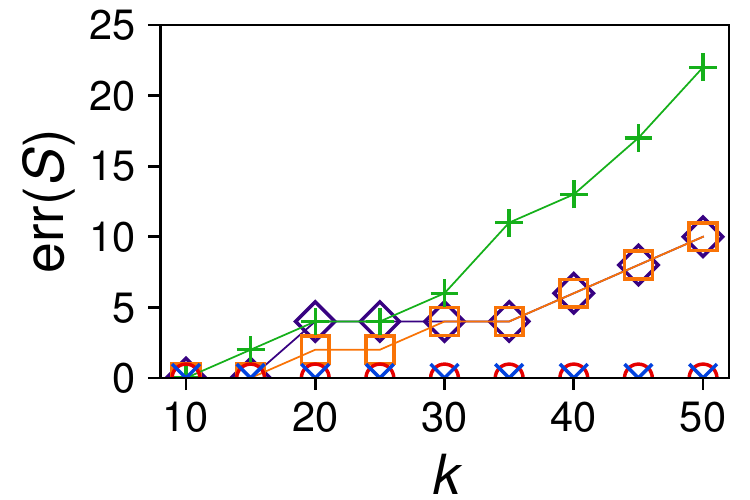}
        \caption{Compas (Gender)}
    \end{subfigure}
    \hfill
    \begin{subfigure}[b]{0.19\textwidth}
        \centering
        \includegraphics[width=\textwidth]{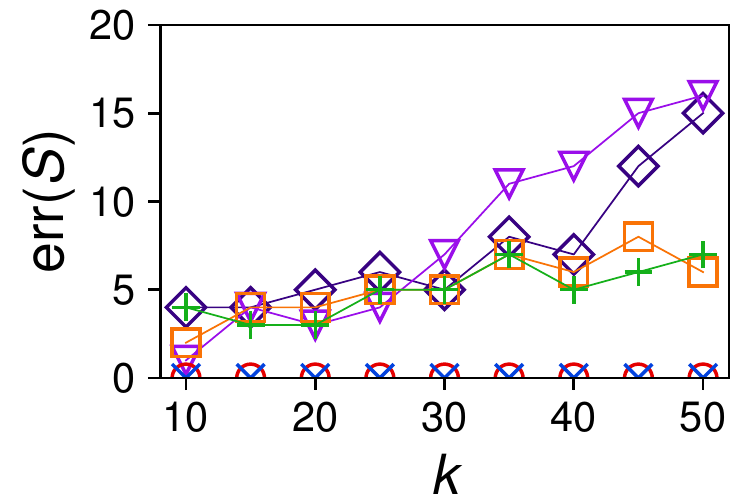}
        \caption{Credit (Job)}
    \end{subfigure}
    \caption{Numbers of fairness violations of different algorithms. Their original implementations without considering group fairness constraints are used for \AlgDMM, \AlgNWF, HS, and \AlgSPH.}
    \label{fig:fair}
    \Description{experimental results}
\end{figure*}

\subsection{Experimental Setup}
\label{ssec:exp:seutp}
All algorithms were implemented in C++, and all experiments were conducted on a PC running Ubuntu 18.04 LTS with a 3.00GHz processor and 32GB memory. We used the CPU time of each algorithm and the minimum happiness ratios (MHRs) of their solutions as the efficiency and effectiveness measures. Furthermore, we used the number of fairness violations defined in~\cite{DBLP:conf/nips/HalabiMNTT20}, \ie for a given $S$,
\begin{equation}\label{eq:fair}
err(S) = \sum_{c \in [C]} \max\{|S \cap \mathcal{D}_c| - h_c, l_c - |S \cap \mathcal{D}_c|, 0\}
\end{equation}
to measure how unfair a solution is compared to the one that satisfies the group fairness constraint.

\smallskip\noindent\textbf{Datasets:}
The experiments are conducted on one synthetic and four real-world datasets as follows.
\begin{itemize}
  \item \textbf{Anti-Correlated} are synthetic datasets with $d \in \{2,\ldots,8\}$ and $n \in \{10^2,\ldots,10^6\}$ created by the generator in~\cite{Borzsony:2001}. Each dataset is divided into $C \in \{2,\ldots,10\}$ groups as follows: we sort the points by the sums of their attributes and divide them into $C$ equal-sized groups accordingly. By default, we use the dataset with $d = 2$ or $6$, $n =$ 10,000 and $C = 3$.
  \item \textbf{Lawschs\footnote{\url{http://www.seaphe.org/databases.php}}} includes 65,494 students' information from 25 law schools from 2005 to 2007. We use two numerical attributes, namely \emph{LSAT} and \emph{GPA}, and two demographic attributes, namely \emph{gender} and \emph{race}, in our experiments.
  \item \textbf{Adult\footnote{\url{https://archive.ics.uci.edu/ml/datasets/adult}}} is a 5$d$ dataset with 32,561 tuples, each of which contains an individual's \emph{education years}, \emph{capital gain}, \emph{capital lose}, \emph{work hours per week} and \emph{overall weight}. The dataset is divided into groups by \emph{gender} and \emph{race}.
  \item \textbf{Compas\footnote{\url{https://github.com/propublica/compas-analysis}}} is a 9$d$ dataset of size 4,743. It consists of customer data from an insurance company. Each tuple denotes an applicant's \emph{insurance number of days}, \emph{count of priority}, \emph{comprehensive score}, and so on. The dataset is partitioned by \emph{gender} and \emph{whether the insured is recidivous}.
  \item \textbf{Credit\footnote{\url{https://archive.ics.uci.edu/ml/datasets/statlog+(german+credit+data)}}} is a 7$d$ dataset of 1,000 tuples about German credit information. The dataset is divided according to \emph{job}, \emph{housing condition}, and \emph{number of employment years}.
\end{itemize}
The statistics of all the above datasets are summarized in Table~\ref{tbl:datasets}.
All the tuples in each dataset are normalized, \ie each numerical attribute is scaled to $[0, 1]$.
For each dataset, the skylines are precomputed as input for each algorithm.
Except for the original groups defined by one categorical attribute, we also combine several of them to form new groups. For example, the combination of \emph{gender} and \emph{race} forms 10 new groups in the \emph{Adult} dataset.

\begin{table}[t]
    \centering
    \footnotesize
    \caption{Statistics of datasets in the experiments. Here, \#skylines is the sum of the numbers of skylines in all groups of each dataset extracted for solution computation.}
    \label{tbl:datasets}
    \begin{tabular}{cccccc}
    \toprule
    \textbf{Dataset} & \textbf{Group} & $d$ & $n$ & $C$ & \textbf{\#skylines} \\
    \midrule
    Anti-Correlated & - & $2$--$16$ & $10^2$--$10^6$ & $2$--$5$ & $0.9 n$--$n$   \\
    \midrule
    \multirow{2}{*}{Lawschs} & Gender & \multirow{2}{*}{2} & \multirow{2}{*}{65,494} & 2 & 19 \\
    & Race & & & 5 & 42 \\
    \midrule
    \multirow{3}{*}{Adult} & Gender & \multirow{3}{*}{5} & \multirow{3}{*}{32,561} & 2 & 130 \\
    & Race & & & 5 & 206 \\
    & G+R & & & 10 & 339 \\
    \midrule
    \multirow{3}{*}{Compas} & Gender & \multirow{3}{*}{9} & \multirow{3}{*}{4,743} & 2 & 195 \\
    & isRecid & & & 2 & 229 \\
    & G+iR & & & 4 & 296 \\
    \midrule
    \multirow{3}{*}{Credit} & Housing & \multirow{3}{*}{7} & \multirow{3}{*}{1,000} & 3 & 120 \\
    & Job & & & 4 & 126 \\
    & Working Years & & & 5 & 185 \\
    \bottomrule
    \end{tabular}
\end{table}

\begin{figure*}[t]
    \centering
    \includegraphics[height=0.12in]{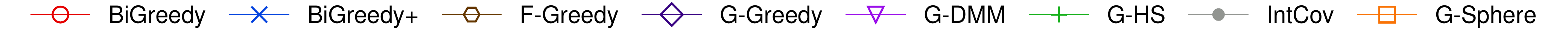}
    \\
    \begin{subfigure}[b]{0.19\textwidth}
        \centering
        \includegraphics[width=\textwidth]{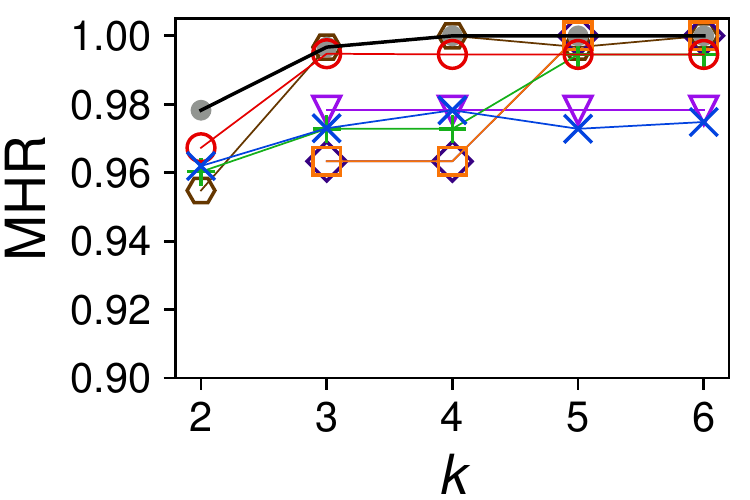}
        \caption{Lawschs (Gender)}
    \end{subfigure}
    \hfill
    \begin{subfigure}[b]{0.19\textwidth}
        \centering
        \includegraphics[width=\textwidth]{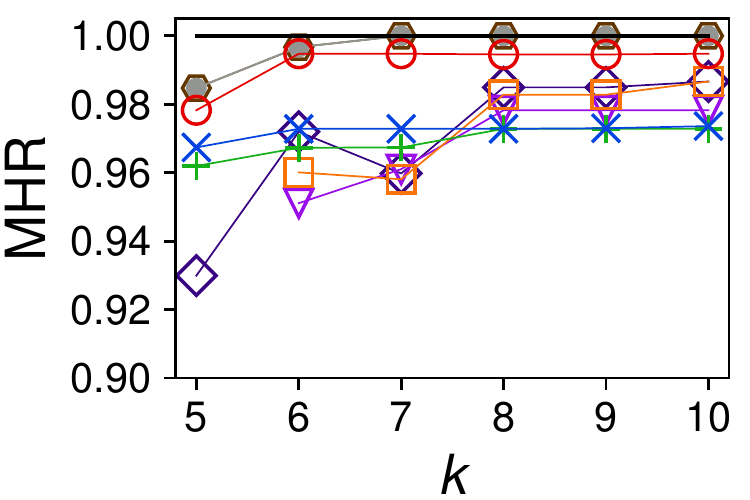}
        \caption{Lawschs (Race)}
    \end{subfigure}
    \hfill
    \begin{subfigure}[b]{0.19\textwidth}
        \centering
        \includegraphics[width=\textwidth]{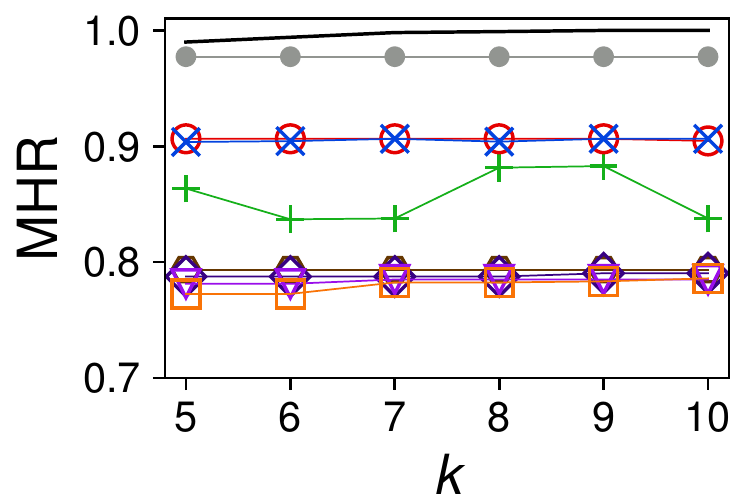}
        \caption{AntiCor\_2D}
    \end{subfigure}
    \hfill
    \begin{subfigure}[b]{0.19\textwidth}
        \centering
        \includegraphics[width=\textwidth]{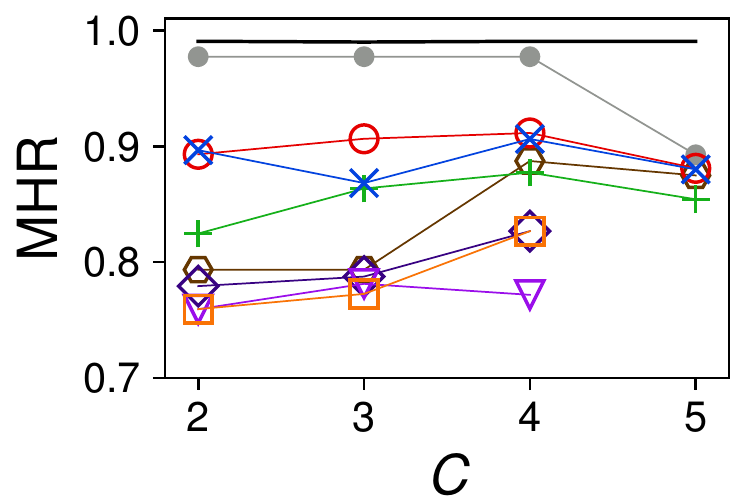}
        \caption{AntiCor\_2D ($k=5$)}
    \end{subfigure}
    \hfill
    \begin{subfigure}[b]{0.19\textwidth}
        \centering
        \includegraphics[width=\textwidth]{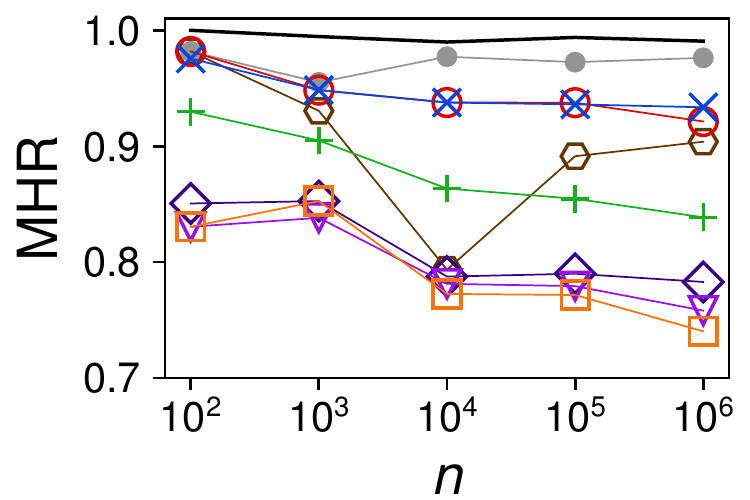}
        \caption{AntiCor\_2D ($k=5$)}
    \end{subfigure}
    \begin{subfigure}[b]{0.19\textwidth}
        \centering
        \includegraphics[width=\textwidth]{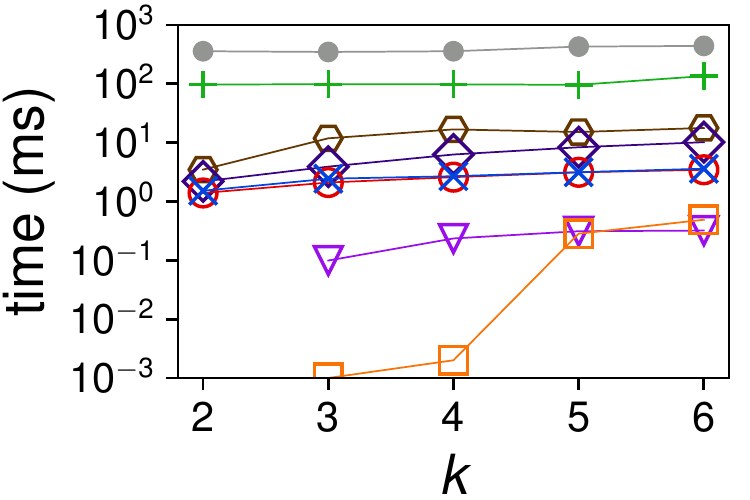}
        \caption{Lawschs (Gender)}
    \end{subfigure}
    \hfill
    \begin{subfigure}[b]{0.19\textwidth}
        \centering
        \includegraphics[width=\textwidth]{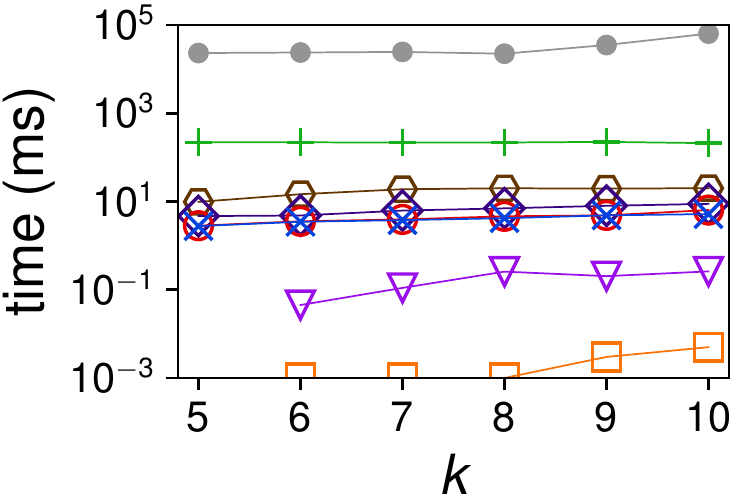}
        \caption{Lawschs (Race)}
    \end{subfigure}
    \hfill
    \begin{subfigure}[b]{0.19\textwidth}
        \centering
        \includegraphics[width=\textwidth]{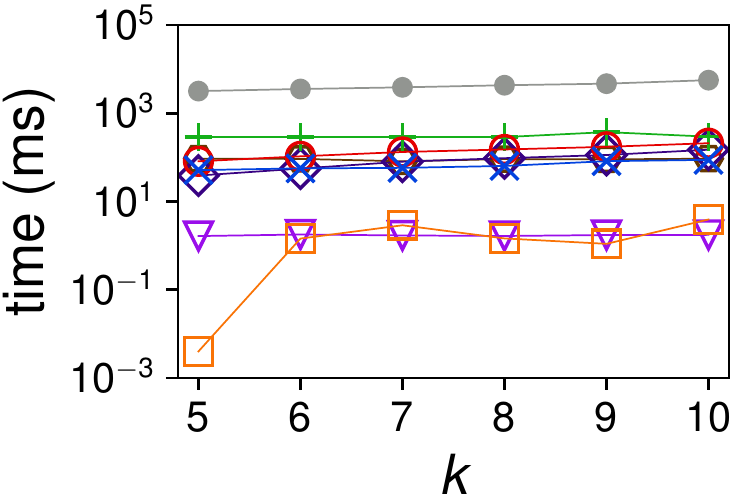}
        \caption{AntiCor\_2D}
    \end{subfigure}
    \hfill
    \begin{subfigure}[b]{0.19\textwidth}
        \centering
        \includegraphics[width=\textwidth]{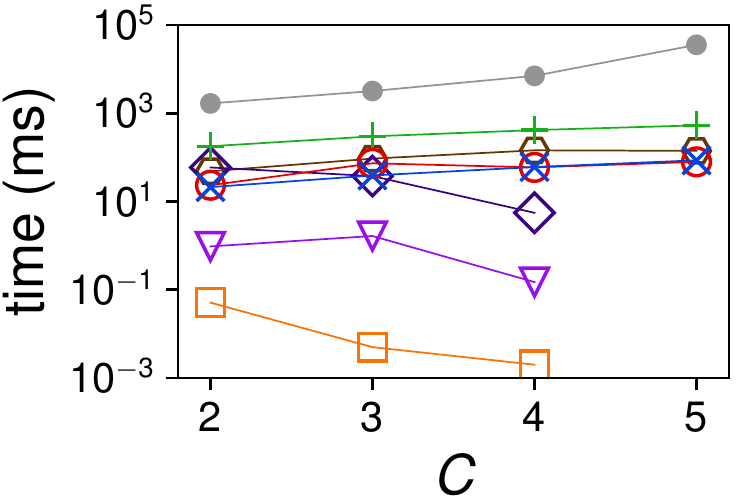}
        \caption{AntiCor\_2D ($k=5$)}
    \end{subfigure}
    \hfill
    \begin{subfigure}[b]{0.19\textwidth}
        \centering
        \includegraphics[width=\textwidth]{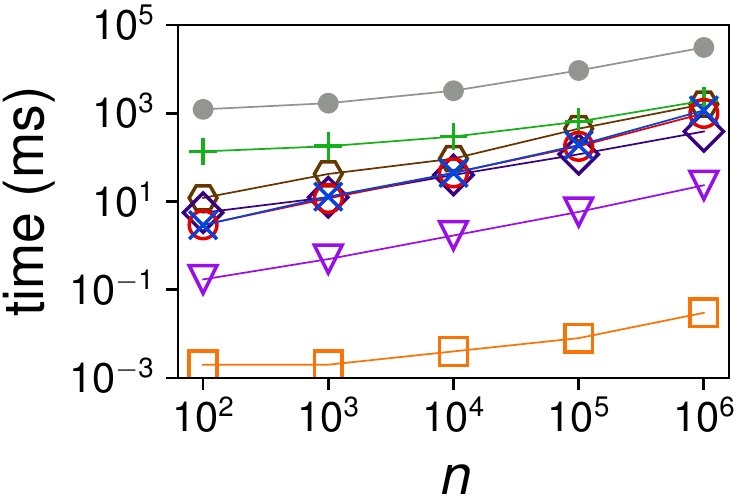}
        \caption{AntiCor\_2D ($k=5$)}
    \end{subfigure}
    \caption{Performance of different algorithms on two-dimensional datasets. In Figures (a)-(e), the MHRs of the optimal solutions without fairness constraints are plotted as black lines to illustrate the ``price of fairness''.}
    \Description{experimental results}
    \label{fig:2d}
\end{figure*}

\smallskip\noindent\textbf{Algorithms:}
We compare the following algorithms for regret-minimizing (RMS) and happiness maximizing set (HMS) problems.
\begin{itemize}
  \item \AlgTwoD is our exact algorithm for FairHMS on 2$d$ datasets proposed in Section~\ref{sec:2d}.
  \item \AlgBG is our bicriteria approximation algorithm for FairHMS in Section~\ref{subsec:bigreedy}.
  \item \AlgIBG is the improved version of \AlgBG by introducing adaptive sampling in Section~\ref{subsec:imp}.
  \item \AlgNWF is an RDP-Greedy algorithm for RMS in~\cite{Nanongkai:2010}.
  \item \AlgDMM is a set-cover-based RMS algorithm in~\cite{Asudeh:2017}.
  \item \AlgSPH is an $\varepsilon$-kernel-based algorithm for RMS in~\cite{Xie:2018}.
  \item HS is a hitting-set-based RMS algorithm in~\cite{Agarwal:2017, Kumar:2018}.
\end{itemize}
We follow the original papers to set the parameters in \AlgDMM, \AlgSPH, and HS.
For \AlgBG and \AlgIBG, we perform some preliminary experiments, which are included in Appendix~\ref{app:exp}, to evaluate the effects of the values of $\delta$, $\varepsilon$, and $\lambda$.
In practice, we use an increasing value of $\delta$ with $d$ to restrict the sample size $m = O(kd)$ so that \AlgBG fits in memory. The values of $m_0$ and $M$ in \AlgIBG are set to $0.05m$ and $m$ accordingly.
Moreover, we fix $\varepsilon = 0.02, \lambda=0.04$ since using smaller $\varepsilon, \lambda$ cannot further improve the quality of solutions.

Note that all the above algorithms, except \AlgTwoD, \AlgBG, and \AlgIBG, are not designed for FairHMS.
To obtain fair solutions, we run those algorithms on each group separately and take the union of the group-specific solutions as the final solution. Such an adapted version of each algorithm is referred to with a prefix `G-' before its name, \eg G-\AlgDMM for the adapted \AlgDMM. For \AlgNWF, an alternative adaptation scheme for fairness constraints is to use the \emph{matroid greedy} algorithm in~\cite{DBLP:conf/nips/HalabiMNTT20}, which adds an item that maximally increases the MHR w.r.t.~the set of selected items while ensuring the satisfaction of fairness constraint at each iteration until $k$ items are included. We refer to it as F-\AlgNWF in the experiments. Nevertheless, to our knowledge, other algorithms cannot be adapted for fairness constraints similarly to \AlgNWF. Also, note that none of the adapted algorithms can achieve any theoretical guarantee for FairHMS.

To verify whether these algorithms can produce fair solutions without explicitly imposing fairness constraints, we run them on the original dataset and compute the number $err(\cdot)$ of fairness violations. We use the proportional representation~\cite{DBLP:conf/nips/HalabiMNTT20}, where the proportion of each group in the solution is approximately equal to that of the original dataset, as the fairness constraint. Specifically, we set $l_c$ to $\lfloor (1 - \alpha) k \cdot \frac{|\mathcal{D}_c|}{|\mathcal{D}|} \rfloor$ or at least $1$ and $h_c$ to $\lceil (1 + \alpha) k \cdot \frac{|\mathcal{D}_c|}{|\mathcal{D}|} \rceil$ or at most $k-C+1$, where $\alpha = 0.1$ following~\cite{DBLP:conf/nips/HalabiMNTT20}.

\subsection{Experimental Results}
\label{ssec:exp:results}

\begin{figure*}[t]
    \centering
    \includegraphics[height=0.12in]{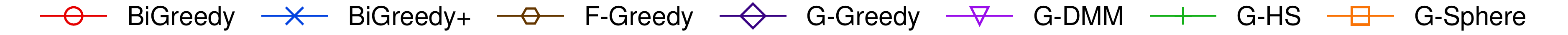}
    \\
    \begin{subfigure}[b]{0.19\textwidth}
        \centering
        \includegraphics[width=\textwidth]{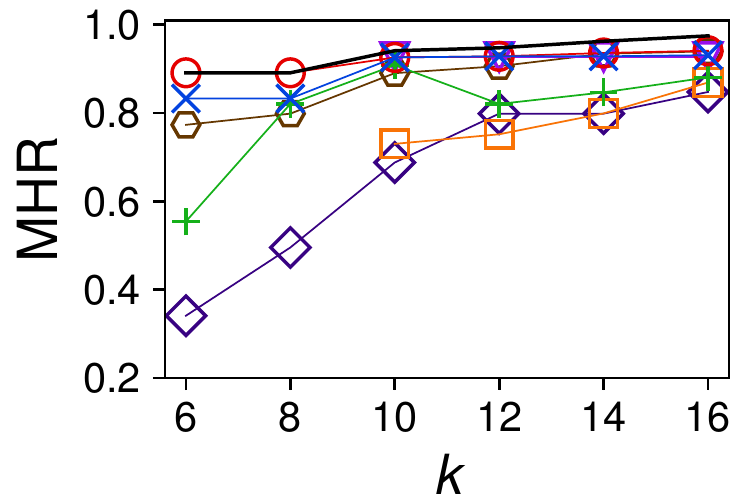}
        \caption{Adult (Gender)}
    \end{subfigure}
    \hfill
    \begin{subfigure}[b]{0.19\textwidth}
        \centering
        \includegraphics[width=\textwidth]{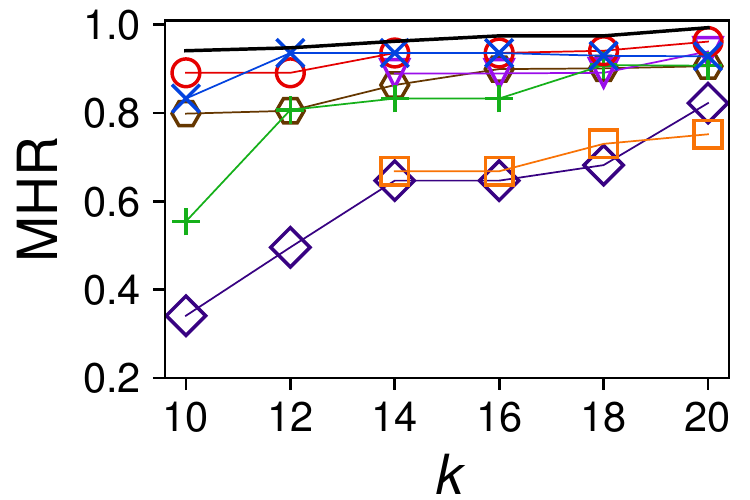}
        \caption{Adult (Race)}
    \end{subfigure}
    \hfill
    \begin{subfigure}[b]{0.19\textwidth}
        \centering
        \includegraphics[width=\textwidth]{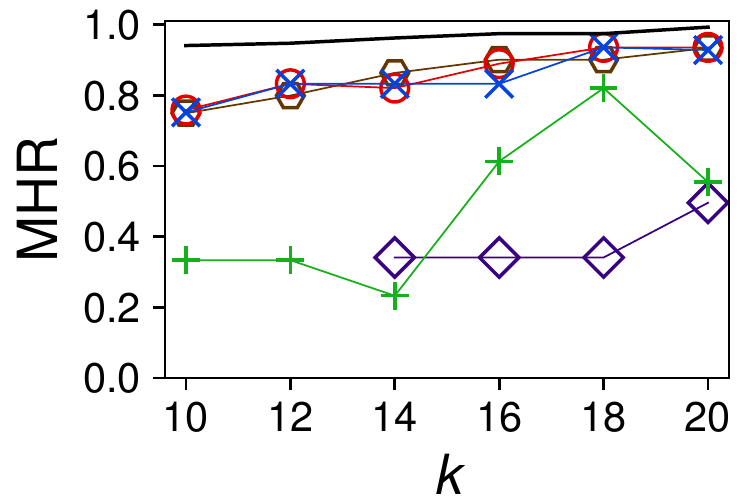}
        \caption{Adult (G+R)}
    \end{subfigure}
    \hfill
    \begin{subfigure}[b]{0.19\textwidth}
        \centering
        \includegraphics[width=\textwidth]{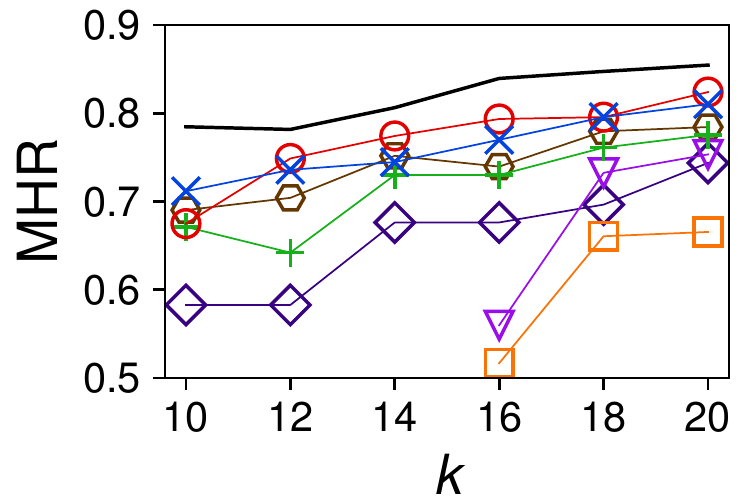}
        \caption{AntiCor\_6D}
    \end{subfigure}
    \hfill
    \begin{subfigure}[b]{0.19\textwidth}
        \centering
        \includegraphics[width=\textwidth]{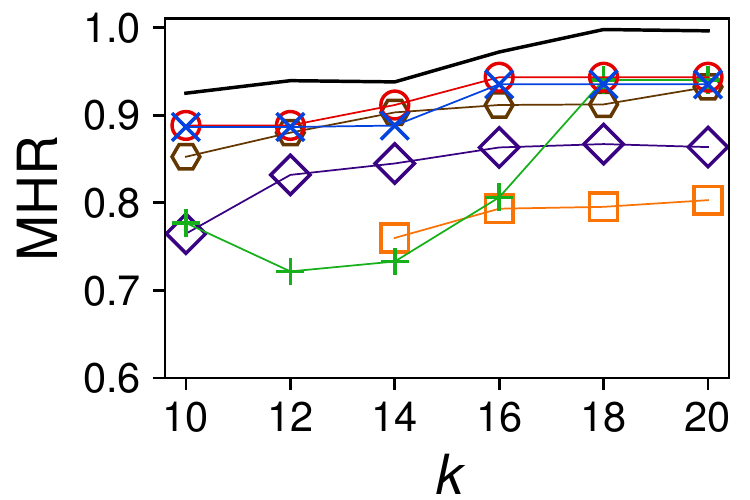}
        \caption{Compas (Gender)}
    \end{subfigure}
    \begin{subfigure}[b]{0.19\textwidth}
        \centering
        \includegraphics[width=\textwidth]{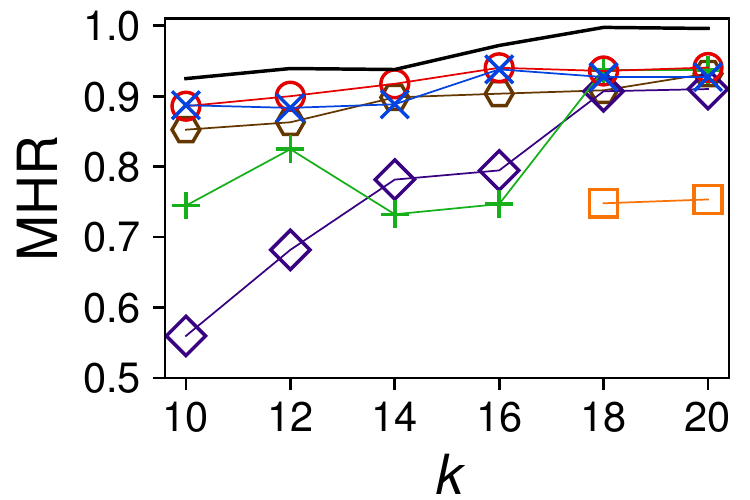}
        \caption{Compas (isRecid)}
    \end{subfigure}
    \hfill
    \begin{subfigure}[b]{0.19\textwidth}
        \centering
        \includegraphics[width=\textwidth]{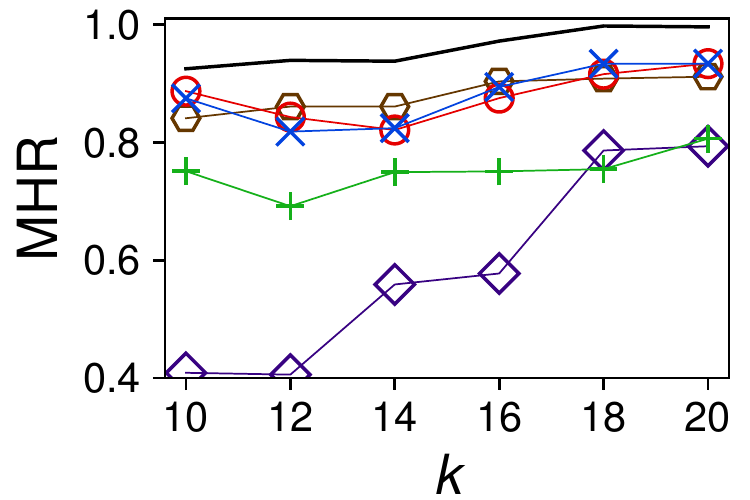}
        \caption{Compas (G+iR)}
    \end{subfigure}
    \hfill
    \begin{subfigure}[b]{0.19\textwidth}
        \centering
        \includegraphics[width=\textwidth]{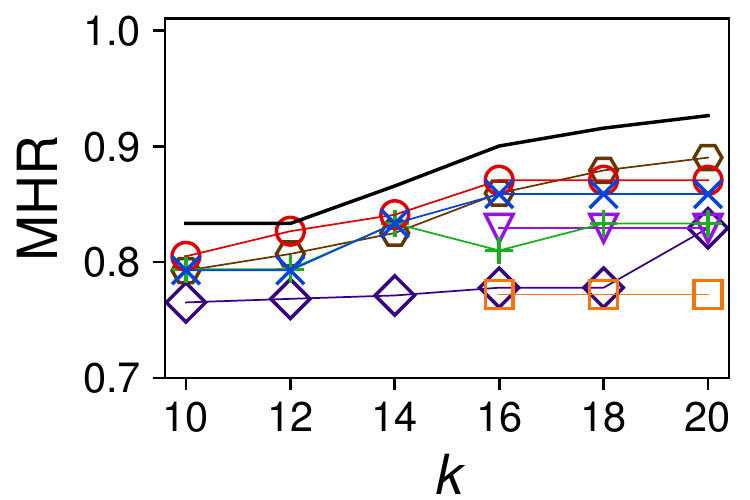}
        \caption{Credit (Job)}
    \end{subfigure}
    \hfill
    \begin{subfigure}[b]{0.19\textwidth}
        \centering
        \includegraphics[width=\textwidth]{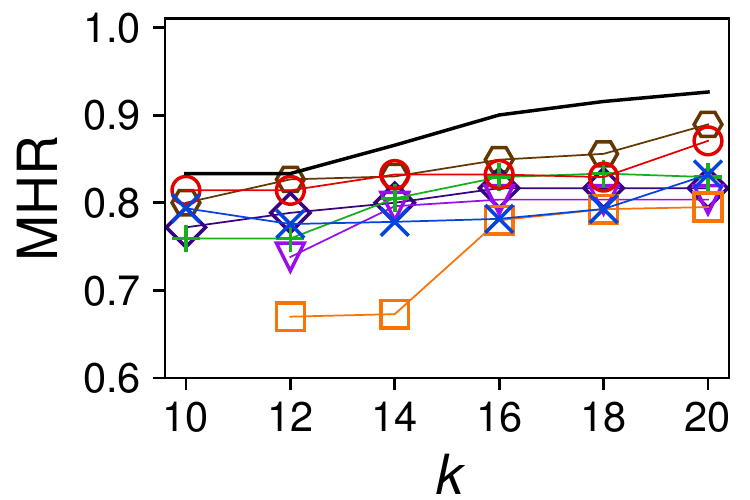}
        \caption{Credit (Housing)}
    \end{subfigure}
    \hfill
    \begin{subfigure}[b]{0.19\textwidth}
        \centering
        \includegraphics[width=\textwidth]{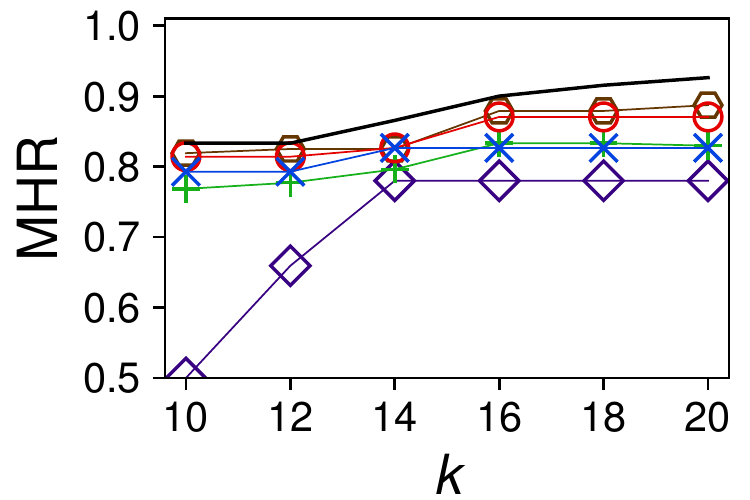}
        \caption{Credit (WY)}
    \end{subfigure}
    \caption{Results for the MHRs of different algorithms on multi-dimensional datasets by varying solution size $k$. In each figure, the MHRs of the best solutions without fairness constraints are plotted as black lines to illustrate the ``price of fairness''.}
    \Description{experimental results}
    \label{fig:md:mhr}
\end{figure*}

\begin{figure*}[t]
    \centering
    \includegraphics[height=0.12in]{exp/legend-md.pdf}
    \\
    \begin{subfigure}[b]{0.19\textwidth}
        \centering
        \includegraphics[width=\textwidth]{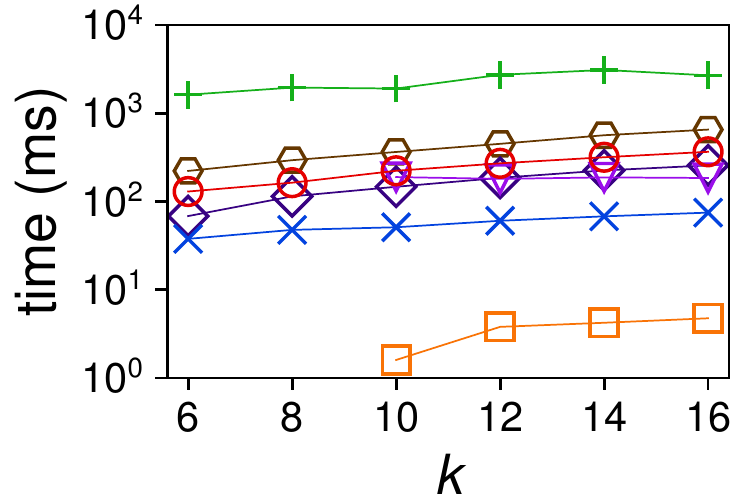}
        \caption{Adult (Gender)}
    \end{subfigure}
    \hfill
    \begin{subfigure}[b]{0.19\textwidth}
        \centering
        \includegraphics[width=\textwidth]{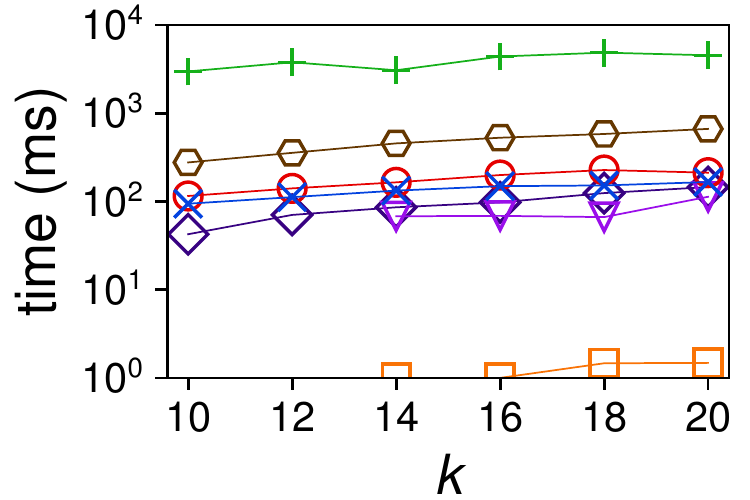}
        \caption{Adult (Race)}
    \end{subfigure}
    \hfill
    \begin{subfigure}[b]{0.19\textwidth}
        \centering
        \includegraphics[width=\textwidth]{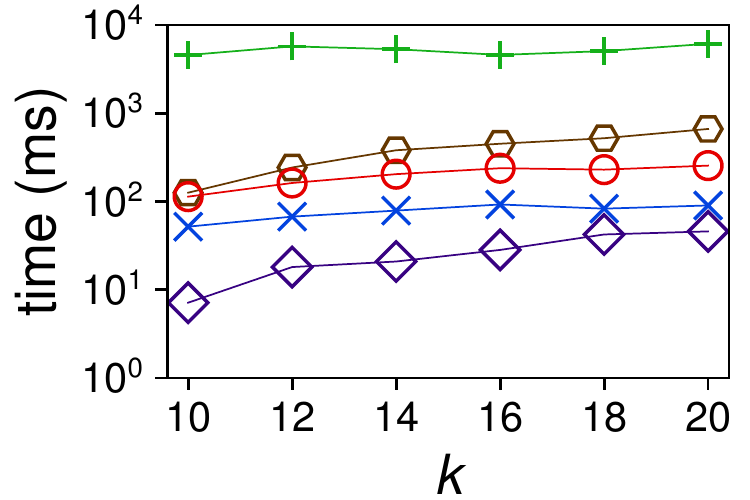}
        \caption{Adult (G+R)}
    \end{subfigure}
    \hfill
    \begin{subfigure}[b]{0.19\textwidth}
        \centering
        \includegraphics[width=\textwidth]{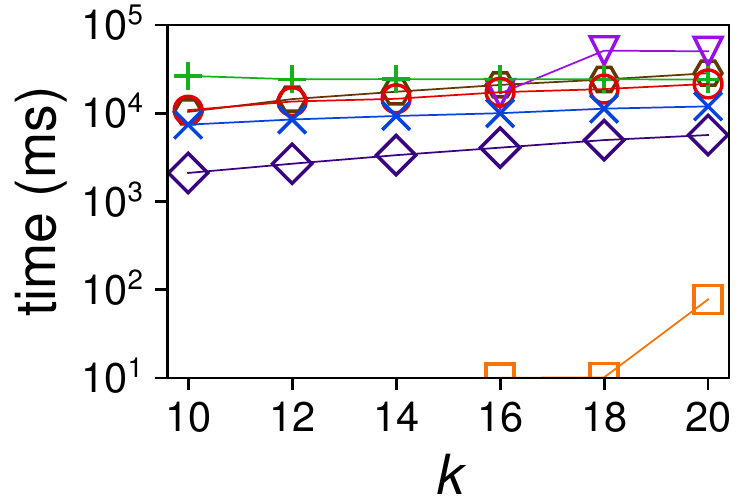}
        \caption{AntiCor\_6D}
    \end{subfigure}
    \hfill
    \begin{subfigure}[b]{0.19\textwidth}
        \centering
        \includegraphics[width=\textwidth]{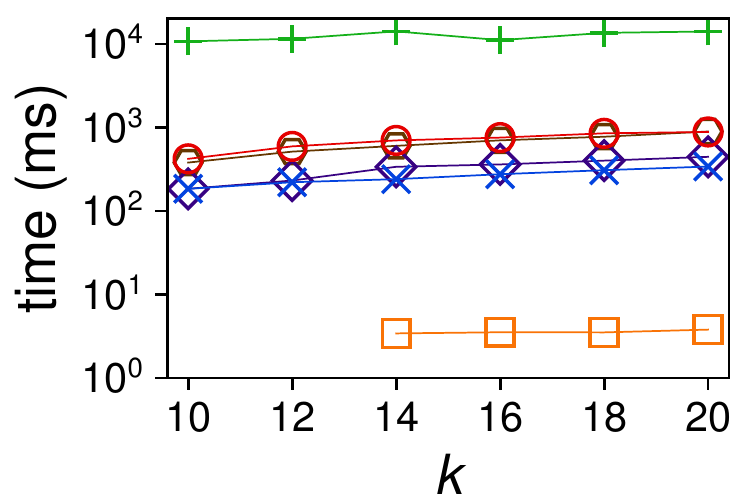}
        \caption{Compas (Gender)}
    \end{subfigure}
    \begin{subfigure}[b]{0.19\textwidth}
        \centering
        \includegraphics[width=\textwidth]{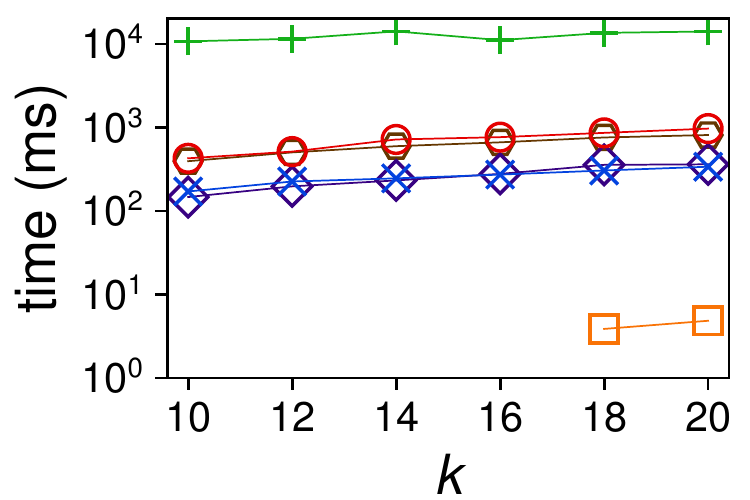}
        \caption{Compas (isRecid)}
    \end{subfigure}
    \hfill
    \begin{subfigure}[b]{0.19\textwidth}
        \centering
        \includegraphics[width=\textwidth]{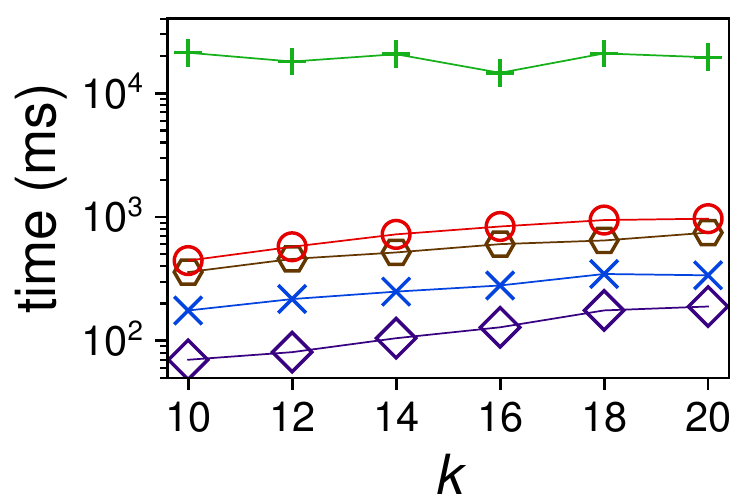}
        \caption{Compas (G+iR)}
    \end{subfigure}
    \hfill
    \begin{subfigure}[b]{0.19\textwidth}
        \centering
        \includegraphics[width=\textwidth]{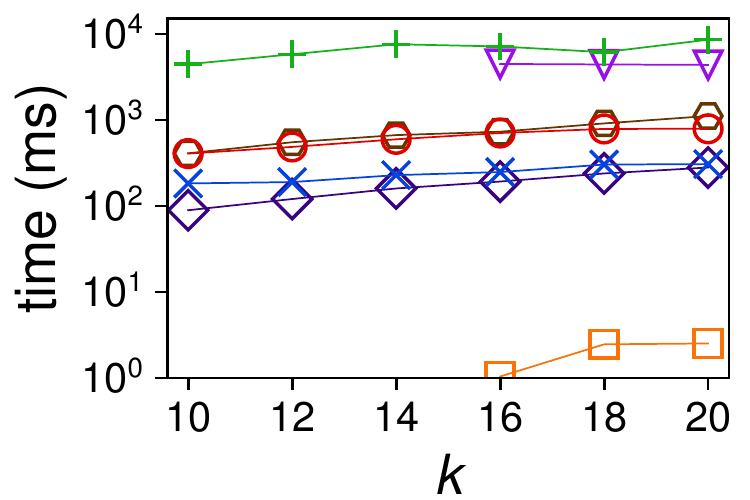}
        \caption{Credit (Job)}
    \end{subfigure}
    \hfill
    \begin{subfigure}[b]{0.19\textwidth}
        \centering
        \includegraphics[width=\textwidth]{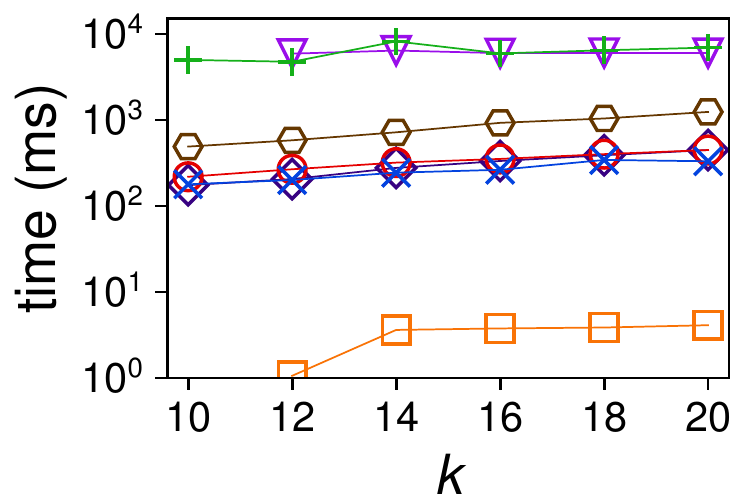}
        \caption{Credit (Housing)}
    \end{subfigure}
    \hfill
    \begin{subfigure}[b]{0.19\textwidth}
        \centering
        \includegraphics[width=\textwidth]{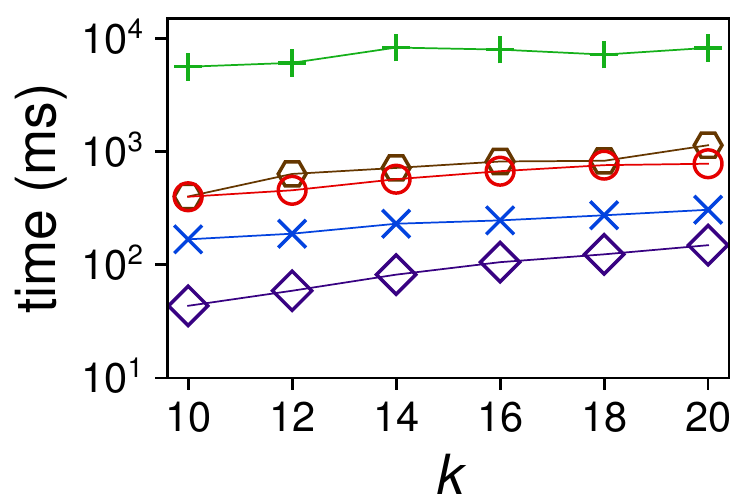}
        \caption{Credit (WY)}
    \end{subfigure}
    \caption{Results for the running time of different algorithms on multi-dimensional datasets by varying solution size $k$.}
    \Description{experimental results}
    \label{fig:md:time}
\end{figure*}

\noindent\textbf{Fairness Violations:}
In Figure~\ref{fig:fair}, we present the number of fairness violations computed by $err(S)$ in Equation~\ref{eq:fair} for the solution $S$ returned by each algorithm with varying the solution size $k$. Here, $err(S) = 0$ means that $S$ satisfies the group fairness constraint, while the larger $err(S)$ is, the further away $S$ is from a fair solution.
All the baseline algorithms violate the group fairness constraints in almost all cases.
This is expected because their original implementations do not consider fairness; thus, their solutions show a bias towards the advantaged groups while under-representing the disadvantaged ones.
Our proposed algorithms always obtain a solution with $err(S)=0$ because they strictly follow the group fairness constraint.
Therefore, we will use the fair adaptations described in Section~\ref{ssec:exp:seutp} instead of the original (unfair) implementations of the baseline algorithms in the remaining experiments.

\smallskip\noindent\textbf{Results on Two-Dimensional Datasets:}
In Figure~\ref{fig:2d}, we compare the performance of our exact \AlgTwoD algorithm and approximation \AlgBG and \AlgIBG algorithms with the (adapted) baseline algorithms on two-dimensional datasets, \ie \emph{Lawschs} and \emph{AntiCor\_2D}.
Note that the results of G-\AlgDMM and G-\AlgSPH are ignored when $k$ is small or $C$ is large because they require that $k \geq d$ for computation and cannot provide any solution if there exists any $c \in [C]$ with $h_c < d$.
In the figures for MHRs, we plot the MHRs of the optimal solutions for unconstrained HMS as black lines, from which we observe that the ``price of fairness'', \ie the decrease in MHR led by fairness constraints, is low in most cases, as the differences in MHRs between the unconstrained and fair solutions are mostly within $0.02$.
Among the algorithms we compare, \AlgTwoD always obtains the highest MHRs due to its optimality. Meanwhile, it is also the slowest because of the time-consuming dynamic programming procedure.
All the remaining algorithms return near-optimal solutions (MHRs $>0.9$) on \emph{Lawschs} because the sizes of skylines (see Table~\ref{tbl:datasets}) are very small.
On \emph{AntiCor\_2D}, where the size of skylines is larger, \AlgBG and \AlgIBG are better than the baseline algorithms in terms of the quality of solutions.
In terms of time efficiency, all of the algorithms except \AlgTwoD finish within one second on all two-dimensional datasets.

In summary, \AlgTwoD provides exact solutions for FairHMS in reasonable time while \AlgBG and \AlgIBG return solutions of better quality than the baselines with high efficiency for FairHMS on two-dimensional datasets.

\smallskip\noindent\textbf{Results on Multi-Dimensional Datasets:}
In Figures~\ref{fig:md:mhr} and~\ref{fig:md:time}, we present the MHRs and running time of different algorithms on multi-dimensional datasets, \ie \emph{Adult}, \emph{AntiCor\_6D}, \emph{Compas}, and \emph{Credit}, with different group partitions by varying the solution size $k$.
We note that some results are omitted since the corresponding algorithms cannot obtain any solution. For example, \AlgDMM cannot finish when $d > 7$ due to the huge memory consumption. Thus, the results of G-\AlgDMM are ignored on \emph{Compas}. Moreover, as is the case for two-dimensional datasets, G-\AlgDMM and G-\AlgSPH still cannot provide any solution if there exists any $c \in [C]$ with $h_c < d$.

Regarding the quality of solutions, the MHRs of all algorithms generally increase with $k$.
In Figure~\ref{fig:md:mhr}, we plot the highest MHRs among the solutions returned by all the baseline algorithms for unconstrained HMS as black lines to illustrate the ``price of fairness''.
Compared with the results in 2$d$, the gaps in MHRs between unconstrained and fair solutions become larger in some cases, which may be due to two reasons. First, unconstrained solutions are merely picked from a few groups because of data skewness and thus substantially different from fair solutions. Second, unlike \AlgTwoD, the solutions of \AlgBG and \AlgIBG are suboptimal.
Furthermore, the MHRs of \AlgBG are the same as or slightly higher than those of \AlgIBG, both of which are greater than those of adapted baseline algorithms in most cases.
These results confirm the effectiveness of \AlgBG and \AlgIBG for FairHMS.
We also note that in a few cases, such as Figure~\ref{fig:md:mhr} (h)--(j), the MHRs of \AlgBG and \AlgIBG are slightly lower than those of F-\AlgNWF when $k$ is large. This is because the estimations of MHRs based on $\delta$-nets are inaccurate, as $\delta$ is too large for high dimensionality and the distributions of values for different attributes are highly skewed. In such cases, the MHR estimation using linear programs in F-\AlgNWF is more accurate.

In terms of time efficiency, we observe that \AlgIBG runs up to 5 times faster than \AlgBG as excepted because of smaller sample sizes for $\delta$-nets.
Although several baseline algorithms, e.g., G-\AlgSPH, G-\AlgNWF, and G-\AlgDMM (when $d \leq 6$), run faster than \AlgBG and \AlgIBG, their solution quality is inferior to \AlgBG and \AlgIBG in almost all cases. Since the adaptation of these algorithms for ensuring fairness constraints is merely to run one instance of the algorithm on the skylines of each group $\mathcal{D}_c$ with a smaller solution size $k_c \in [l_c, h_c]$, these algorithms have little losses in efficiency due to adaptation. But their solutions may include highly redundant tuples from different groups because the selection procedures for different groups are independent of each other. Hence, they suffer from significant decreases in MHRs.
Since the first step of \AlgSPH is to select the ``extreme'' points with the largest attribute values in each dimension, its solution mainly consists of these extreme points when $k$ is close to $d$. For this reason, G-\AlgSPH runs the fastest but fails to provide good solutions in almost all cases.
F-\AlgNWF runs slower than \AlgBG and \AlgIBG in most cases. In F-\AlgNWF, a costly linear program is executed on each item in the skyline at every iteration to find the items to add. Therefore, its running time is almost decided by the number of linear programs to execute, which is close to the original \AlgNWF in~\cite{Nanongkai:2010} and much longer than that of G-\AlgNWF.

To sum up, \AlgBG and \AlgIBG outperform the baseline algorithms in the quality of solutions on multi-dimensional datasets.
Compared with \AlgBG, \AlgIBG further improves the time efficiency at the expense of a bit of loss in effectiveness.
Furthermore, they still only take one or several seconds for solution computation on multi-dimensional datasets.

\begin{figure}[t]
    \centering
    \includegraphics[height=0.24in]{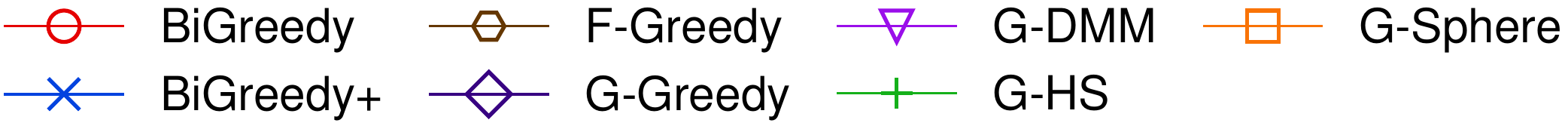}
    \\
    \begin{subfigure}[b]{0.4\textwidth}
        \centering
        \includegraphics[width=0.48\textwidth]{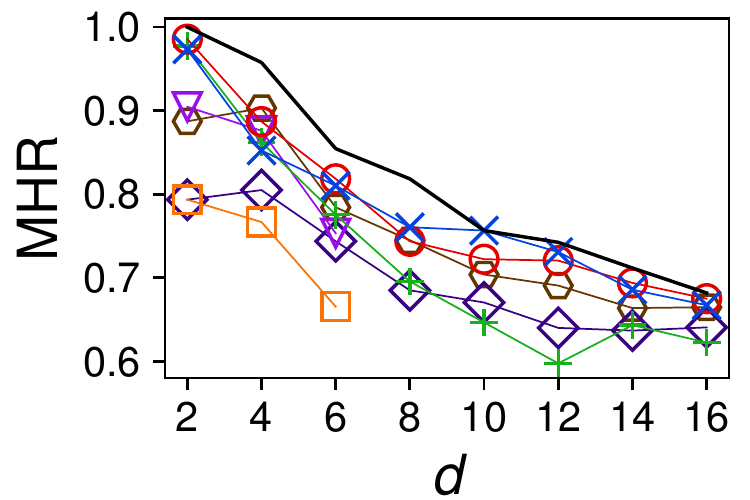}
        \hfill
        \includegraphics[width=0.48\textwidth]{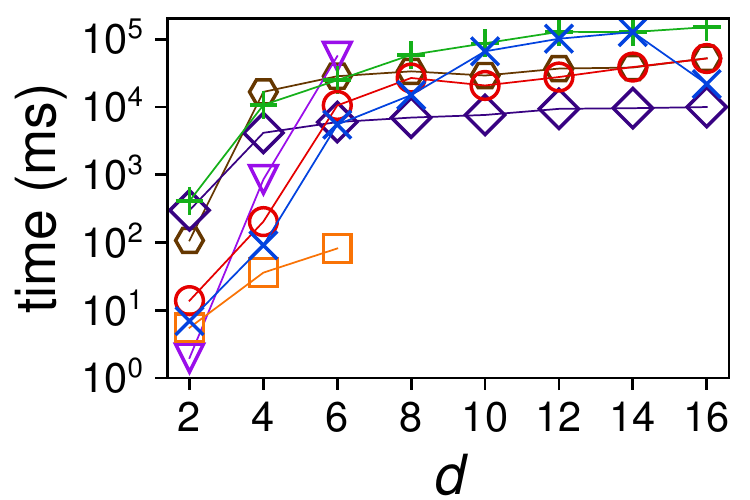}
        \caption{AntiCor (Varying $d$)}
    \end{subfigure}
    \\
    \begin{subfigure}[b]{0.4\textwidth}
        \centering
        \includegraphics[width=0.48\textwidth]{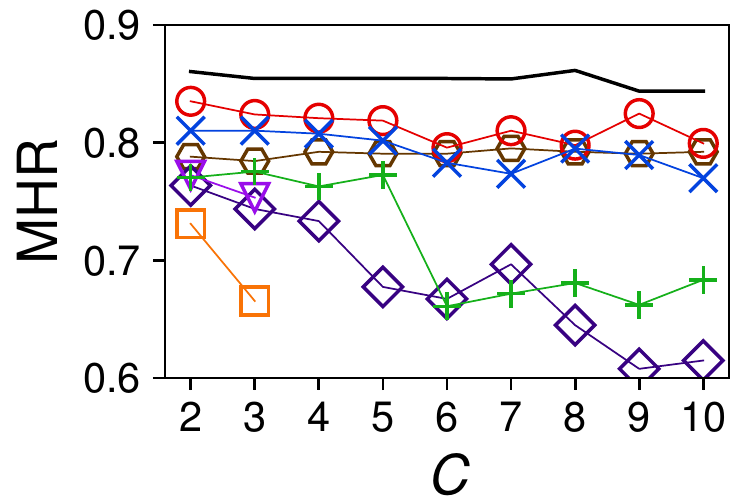}
        \hfill
        \includegraphics[width=0.48\textwidth]{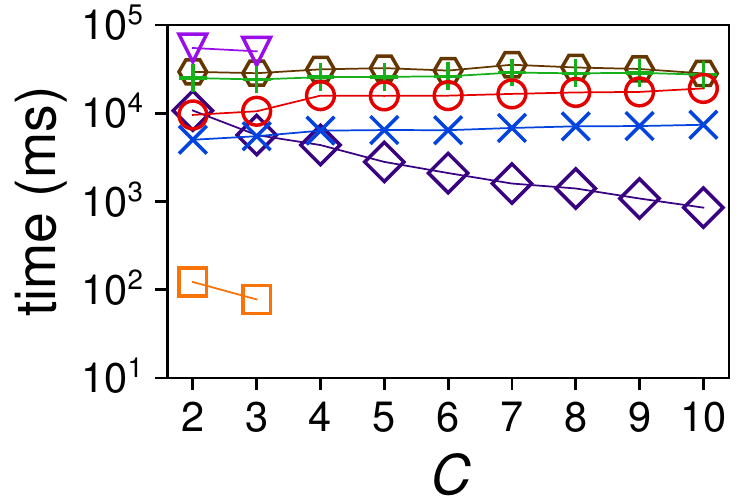}
        \caption{AntiCor\_6D (Varying $C$)}
    \end{subfigure}
    \\
    \begin{subfigure}[b]{0.4\textwidth}
        \centering
        \includegraphics[width=0.48\textwidth]{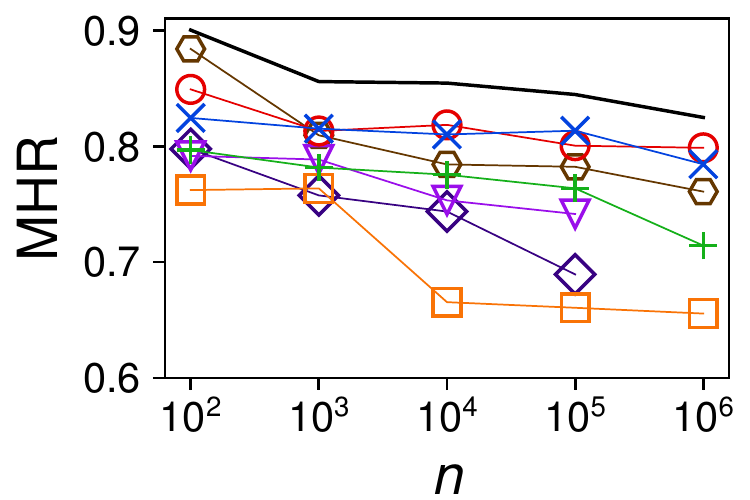}
        \hfill
        \includegraphics[width=0.48\textwidth]{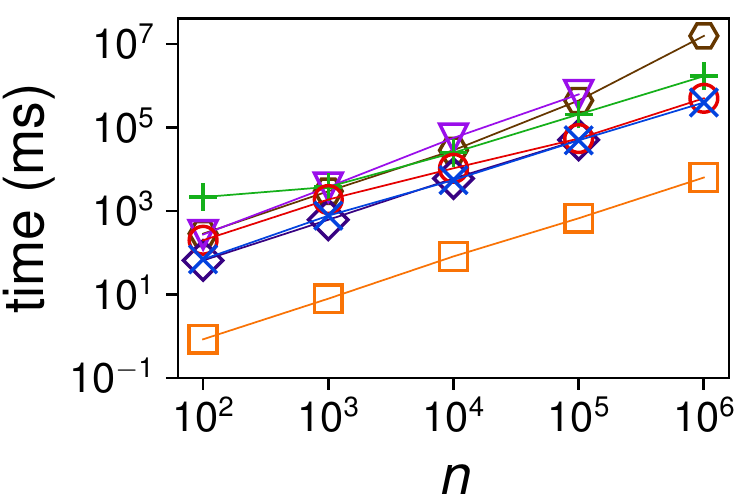}
        \caption{AntiCor\_6D (Varying $n$)}
    \end{subfigure}
    \caption{Performance of different algorithms on the anti-correlated datasets for solution size $k=20$ by varying dimensionality $d$, number of groups $C$, and dataset size $n$.}
    \Description{experimental results}
    \label{fig:scalability}
\end{figure}

\smallskip\noindent\textbf{Scalability:}
To evaluate the scalability of different algorithms with respect to \emph{dimensionality} $d$, \emph{number of groups} $C$ and \emph{dataset size} $n$, we perform extensive experiments on the generated anti-correlated datasets with $d$ ranging from $2$ to $8$ (with $n = 10^4$ and $C=3$), $C$ ranging from $2$ to $10$ (with $d=2$ or $6$ and $n=10^4$), and $n$ ranging from $10^2$ to $10^6$ (with $d=2$ or $6$ and $C=3$).
We fix the solution size $k$ to $5$ in the experiments on two-dimensional datasets only and $20$ in all remaining
experiments.

The results by varying $C$ and $n$ on 2$d$ anti-correlated datasets are presented in Figure~\ref{fig:2d} (d)--(e) and (i)--(j).
Since the time complexity of \AlgTwoD is exponential with respect to $C$, it cannot terminate within the time limit (\ie 1,000 seconds) when $C \geq 6$. Therefore, in the experiments, we vary $C$ from $2$ to $5$.
We observe that the MHRs of different algorithms drop when $C$ increases because the fairness constraint becomes more restricted, \eg it requires precisely one point from each group when $k=5$ and $C=5$. Nevertheless, the advantages of \AlgTwoD, \AlgBG, and \AlgIBG remain for different values of $C$.
In terms of time efficiency, unlike \AlgTwoD, \AlgBG and \AlgIBG only run slightly slower when $C$ is larger.
The running time of other algorithms drops with $C$ due to a smaller solution size for each group.
Furthermore, the MHRs of all the algorithms decreases with an increasing $n$. Meanwhile, their running time increases nearly linearly with $n$.
The above trends can be explained by the fact that the sizes of skylines increase proportionally with $n$ on the anti-correlated datasets.
Nevertheless, \AlgTwoD, \AlgBG, and \AlgIBG exhibit significant advantages in solution quality for different values of $n$ while the running time is still within 100 seconds for \AlgTwoD or 1 second for \AlgBG and \AlgIBG.

The results by varying $d$, $C$, and $n$ on the anti-correlated datasets of higher dimensionalities are shown in Figure~\ref{fig:scalability}.
When $d$ increases, the MHR decreases while each algorithm's running time grows drastically. For \AlgBG and \AlgIBG, the sampling sizes of $\delta$-nets grow exponentially with $d$. To achieve practical efficiency, we use larger values of $\delta$ with an increasing $d$, which inevitably leads to losses in solution quality. Other algorithms also face scalability problems due to ``the curse of dimensionality''.
Furthermore, when $C$ increases, the MHR decreases while the running time increases for all algorithms.
The reasons for such trends are similar to those for two-dimensional data.
A new observation is that the advantages of \AlgBG and \AlgIBG over the baselines become more apparent when $C$ is larger, mainly because the solutions of baselines computed from more groups are more highly redundant.
Moreover, we also observe similar trends for MHRs and running time with respect to $n$.
The advantages of \AlgBG and \AlgIBG over the baselines also become larger with an increasing $n$.
Finally, the running time of all algorithms increases nearly linearly with $n$.

To sum up, \AlgBG and \AlgIBG show better scalability than the baselines with regard to $C$ and $n$: they achieve bigger advantages in solution quality while the time efficiencies are still comparable. Moreover, they can provide solutions better than baselines in a reasonable time by adjusting the input parameters for high dimensionality $d$.

\section{Related Work}
\label{sec:literature}

\noindent\textbf{RMS, HMS \& Their Variants:}
There have been many studies on \emph{regret minimizing set} (RMS)~\cite{Asudeh:2017,Nanongkai:2010,Peng:2014,Shetiya:2020,Xie:2018}, \emph{happiness maximizing set} (HMS)~\cite{Qiu:2018,Xie:2020}, and different variants of them~\cite{Agarwal:2017,Asudeh:2019,Cao:2017,Chester:2014,Dong:2019,Faulkner:2015,Kumar:2018,Luenam:2021,Nanongkai:2012,Qi:2018,Soma:2017,Storandt:2019,Wang:2021b,Xiao:2021,Xie:2019,Zeighami:2019} (see~\cite{Xie:2020VLDBJ} for an extensive survey). The RMS problem was first proposed by \citet{Nanongkai:2010} to alleviate the deficiencies of top-$k$ and skyline queries. Due to the NP-hardness of RMS~\cite{Agarwal:2017,Cao:2017,Chester:2014} in three and higher dimensions, most of the literature focuses on approximation or heuristic algorithms for RMS. \citet{Nanongkai:2010} first proposed a greedy heuristic for RMS. \citet{Peng:2014} devised a geometric method to improve the efficiency of the greedy heuristic. \citet{Asudeh:2017} proposed an approximation algorithm for RMS by transforming it into a set cover problem. \citet{Xie:2018} proposed an asymptotically optimal approximation algorithm called \textsc{Sphere} for RMS based on the notion of $\varepsilon$-kernels~\cite{DBLP:journals/jacm/AgarwalHV04}. \citet{Shetiya:2020} proposed a \textsc{k-medoid} based heuristic algorithm for RMS. In addition, exact algorithms for RMS in $\mathbb{R}^2$ by transforming it into a shortest-path/cycle problem in a graph were proposed in~\cite{Asudeh:2017,DBLP:conf/pods/WangM0T21}. Since minimizing the \emph{regret ratio} and maximizing the \emph{happiness ratio} are essentially equivalent by definition, the HMS problem was considered recently~\cite{Qiu:2018,Xie:2020} in place of RMS for ease of theoretical analysis. \citet{Qiu:2018} extended and improved the greedy heuristic for HMS. \citet{Xie:2020} further proposed a \textsc{Cone-Greedy} algorithm based on the geometric interpretation of HMS.

Different variants of RMS/HMS were also extensively studied in the literature. \citet{Chester:2014} extended RMS to $k$RMS, where the notion of \emph{regret ratio} was relaxed to \emph{$k$-regret ratio} measuring the relative loss of the largest utility in the subset w.r.t.~the $k$-th largest utility in the whole database. There have been many algorithms proposed for $k$RMS, including greedy algorithms~\cite{Chester:2014,Dong:2019}, $\varepsilon$-kernel based algorithms~\cite{Agarwal:2017,Cao:2017}, hitting-set based algorithms~\cite{Agarwal:2017,Kumar:2018}, and fully-dynamic algorithms~\cite{Wang:2021b,9756312}. Data reduction-based algorithms for the happiness maximization version of $k$RMS ($k$HMS) were proposed in~\cite{Luenam:2021}. The average RMS problem that minimizes the average instead of the maximum of regret ratios was considered in~\cite{Zeighami:2019,Storandt:2019,Shetiya:2020,Shetiya:2020}. The rank-regret representative (RRR) problem, where the regret was defined by \emph{ranking} other than \emph{utilities}, was studied in~\cite{Asudeh:2019,Xiao:2021}. The RMS problems defined on non-linear utility functions were considered in~\cite{Faulkner:2015,Qi:2018}. Interactive RMS problems~\cite{Nanongkai:2012,Xie:2019,DBLP:conf/apweb/0001C20} were proposed to enhance RMS with user interactions.

However, despite the rich literature on RMS/HMS, none of the above studies have taken \emph{fairness} into consideration. They could provide solutions that under-represent some protected groups, as already shown in our experiments, and might further lead to bias and discrimination in algorithmic decisions.

\smallskip\noindent\textbf{Fairness in Subset Selection and Ranking:}
Another line of work related to ours is \emph{subset selection and ranking under fairness constraints}.
The relationship between this work and existing fair subset selection problems is that the notions of \emph{fairness} are the same or similar.
Specifically, the group fairness constraint in this work and its special case when the upper and lower bounds are equal to each other have been used in many real-world problems, including DPP-based data summarization~\cite{Celis:2018}, $k$-center clustering~\cite{DBLP:conf/icml/KleindessnerAM19}, top-$k$ selection~\cite{Stoyanovich:2018,DBLP:conf/fat/MehrotraC21}, submodular maximization~\cite{Wang:2021a,DBLP:conf/nips/HalabiMNTT20,DBLP:conf/ijcai/CelisHV18}, diversity maximization~\cite{Moumoulidou:2021}.
However, since the objective function of HMS is not submodular in its original form, existing algorithms for fair submodular maximization cannot be directly used for FairHMS. Moreover, as they only consider a single utility function, the fair top-$k$ selection algorithms are also not applicable to FairHMS. Finally, the inner product-based objective of HMS essentially differs from the distance-based objectives in summarization, clustering, and diversification, and thus the design of algorithms for the fair variants of these problems are very different from that for FairHMS.

Ranking problems under fairness constraints are also attracting much attention recently (see~\cite{Pitoura:2021,Zehlike:2021} for extensive surveys). \citet{DBLP:conf/cikm/ZehlikeB0HMB17} and \citet{DBLP:conf/icalp/CelisSV18} generalized the fairness constraints for top-$k$ selection to top-$k$ ranking by further restricting that every prefix of the ranking should also be proportionally fair. \citet{DBLP:conf/kdd/SinghJ18} defined fairness constraints on rankings in terms of exposure allocation based on the notion of \emph{discounted cumulative gain} (DCG). \citet{DBLP:conf/kdd/Garcia-SorianoB21} designed a maximin-fair ranking framework to achieve individual fairness under group-fairness constraints.
\citet{DBLP:conf/sigmod/AsudehJS019} proposed a scheme to modify a ranking function towards satisfying the desired fairness constraints.
\citet{DBLP:journals/pvldb/KuhlmanR20} considered how to aggregate multiple rankings under fairness constraints.
Unlike subset selection, the concepts of \emph{fairness} in ranking problems consider not only whether a tuple is picked but also its position in the result.
Since HMS is not ordered, the fairness constraints designed for ranking problems are not applicable for FairHMS.

\section{Conclusion}
\label{sec:conclusion}

In this paper, we studied the problem of happiness maximizing sets under group fairness constraints (FairHMS). We first formally defined the FairHMS problem and showed its NP-hardness in three and higher dimensions. Then, we proposed an exact algorithm for FairHMS on two-dimensional data. Moreover, we proposed a bicriteria approximation algorithm for FairHMS in any constant dimension based on the notion of $\delta$-nets and multi-objective submodular maximization. We further introduced an adaptive sampling strategy to improve its practical efficiency. Finally, extensive experiments on real-world and synthetic datasets confirmed the efficacy, efficiency, and scalability of our proposed algorithms.

\bibliographystyle{ACM-Reference-Format}
\bibliography{ref}


\begin{thebibliography}{62}


\ifx \showCODEN    \undefined \def \showCODEN     #1{\unskip}     \fi
\ifx \showDOI      \undefined \def \showDOI       #1{#1}\fi
\ifx \showISBNx    \undefined \def \showISBNx     #1{\unskip}     \fi
\ifx \showISBNxiii \undefined \def \showISBNxiii  #1{\unskip}     \fi
\ifx \showISSN     \undefined \def \showISSN      #1{\unskip}     \fi
\ifx \showLCCN     \undefined \def \showLCCN      #1{\unskip}     \fi
\ifx \shownote     \undefined \def \shownote      #1{#1}          \fi
\ifx \showarticletitle \undefined \def \showarticletitle #1{#1}   \fi
\ifx \showURL      \undefined \def \showURL       {\relax}        \fi
\providecommand\bibfield[2]{#2}
\providecommand\bibinfo[2]{#2}
\providecommand\natexlab[1]{#1}
\providecommand\showeprint[2][]{arXiv:#2}

\bibitem[\protect\citeauthoryear{Agarwal, Har{-}Peled, and Varadarajan}{Agarwal
  et~al\mbox{.}}{2004}]%
        {DBLP:journals/jacm/AgarwalHV04}
\bibfield{author}{\bibinfo{person}{Pankaj~K. Agarwal}, \bibinfo{person}{Sariel
  Har{-}Peled}, {and} \bibinfo{person}{Kasturi~R. Varadarajan}.}
  \bibinfo{year}{2004}\natexlab{}.
\newblock \showarticletitle{Approximating extent measures of points}.
\newblock \bibinfo{journal}{\emph{J. {ACM}}} \bibinfo{volume}{51},
  \bibinfo{number}{4} (\bibinfo{year}{2004}), \bibinfo{pages}{606--635}.
\newblock


\bibitem[\protect\citeauthoryear{Agarwal, Kumar, Sintos, and Suri}{Agarwal
  et~al\mbox{.}}{2017}]%
        {Agarwal:2017}
\bibfield{author}{\bibinfo{person}{Pankaj~K. Agarwal}, \bibinfo{person}{Nirman
  Kumar}, \bibinfo{person}{Stavros Sintos}, {and} \bibinfo{person}{Subhash
  Suri}.} \bibinfo{year}{2017}\natexlab{}.
\newblock \showarticletitle{Efficient Algorithms for k-Regret Minimizing Sets}.
  In \bibinfo{booktitle}{\emph{{SEA}}}. \bibinfo{pages}{7:1--7:23}.
\newblock


\bibitem[\protect\citeauthoryear{Anari, Haghtalab, Naor, Pokutta, Singh, and
  Torrico}{Anari et~al\mbox{.}}{2019}]%
        {Anari:2019}
\bibfield{author}{\bibinfo{person}{Nima Anari}, \bibinfo{person}{Nika
  Haghtalab}, \bibinfo{person}{Seffi Naor}, \bibinfo{person}{Sebastian
  Pokutta}, \bibinfo{person}{Mohit Singh}, {and} \bibinfo{person}{Alfredo
  Torrico}.} \bibinfo{year}{2019}\natexlab{}.
\newblock \showarticletitle{Structured Robust Submodular Maximization: Offline
  and Online Algorithms}. In \bibinfo{booktitle}{\emph{{AISTATS}}}.
  \bibinfo{pages}{3128--3137}.
\newblock


\bibitem[\protect\citeauthoryear{Asudeh, Jagadish, Stoyanovich, and Das}{Asudeh
  et~al\mbox{.}}{2019a}]%
        {DBLP:conf/sigmod/AsudehJS019}
\bibfield{author}{\bibinfo{person}{Abolfazl Asudeh}, \bibinfo{person}{H.~V.
  Jagadish}, \bibinfo{person}{Julia Stoyanovich}, {and} \bibinfo{person}{Gautam
  Das}.} \bibinfo{year}{2019}\natexlab{a}.
\newblock \showarticletitle{Designing Fair Ranking Schemes}. In
  \bibinfo{booktitle}{\emph{{SIGMOD}}}. \bibinfo{pages}{1259--1276}.
\newblock


\bibitem[\protect\citeauthoryear{Asudeh, Nazi, Zhang, and Das}{Asudeh
  et~al\mbox{.}}{2017}]%
        {Asudeh:2017}
\bibfield{author}{\bibinfo{person}{Abolfazl Asudeh}, \bibinfo{person}{Azade
  Nazi}, \bibinfo{person}{Nan Zhang}, {and} \bibinfo{person}{Gautam Das}.}
  \bibinfo{year}{2017}\natexlab{}.
\newblock \showarticletitle{Efficient Computation of Regret-ratio Minimizing
  Set: {A} Compact Maxima Representative}. In
  \bibinfo{booktitle}{\emph{{SIGMOD}}}. \bibinfo{pages}{821--834}.
\newblock


\bibitem[\protect\citeauthoryear{Asudeh, Nazi, Zhang, Das, and Jagadish}{Asudeh
  et~al\mbox{.}}{2019b}]%
        {Asudeh:2019}
\bibfield{author}{\bibinfo{person}{Abolfazl Asudeh}, \bibinfo{person}{Azade
  Nazi}, \bibinfo{person}{Nan Zhang}, \bibinfo{person}{Gautam Das}, {and}
  \bibinfo{person}{H.~V. Jagadish}.} \bibinfo{year}{2019}\natexlab{b}.
\newblock \showarticletitle{{RRR:} Rank-Regret Representative}. In
  \bibinfo{booktitle}{\emph{{SIGMOD}}}. \bibinfo{pages}{263--280}.
\newblock


\bibitem[\protect\citeauthoryear{B{\"{o}}rzs{\"{o}}nyi, Kossmann, and
  Stocker}{B{\"{o}}rzs{\"{o}}nyi et~al\mbox{.}}{2001}]%
        {Borzsony:2001}
\bibfield{author}{\bibinfo{person}{Stephan B{\"{o}}rzs{\"{o}}nyi},
  \bibinfo{person}{Donald Kossmann}, {and} \bibinfo{person}{Konrad Stocker}.}
  \bibinfo{year}{2001}\natexlab{}.
\newblock \showarticletitle{The Skyline Operator}. In
  \bibinfo{booktitle}{\emph{{ICDE}}}. \bibinfo{pages}{421--430}.
\newblock


\bibitem[\protect\citeauthoryear{Cao, Li, Wang, Wang, Wang, Wong, and Zhan}{Cao
  et~al\mbox{.}}{2017}]%
        {Cao:2017}
\bibfield{author}{\bibinfo{person}{Wei Cao}, \bibinfo{person}{Jian Li},
  \bibinfo{person}{Haitao Wang}, \bibinfo{person}{Kangning Wang},
  \bibinfo{person}{Ruosong Wang}, \bibinfo{person}{Raymond~Chi{-}Wing Wong},
  {and} \bibinfo{person}{Wei Zhan}.} \bibinfo{year}{2017}\natexlab{}.
\newblock \showarticletitle{k-Regret Minimizing Set: Efficient Algorithms and
  Hardness}. In \bibinfo{booktitle}{\emph{{ICDT}}}.
  \bibinfo{pages}{11:1--11:19}.
\newblock


\bibitem[\protect\citeauthoryear{Celis, Huang, and Vishnoi}{Celis
  et~al\mbox{.}}{2018a}]%
        {DBLP:conf/ijcai/CelisHV18}
\bibfield{author}{\bibinfo{person}{L.~Elisa Celis}, \bibinfo{person}{Lingxiao
  Huang}, {and} \bibinfo{person}{Nisheeth~K. Vishnoi}.}
  \bibinfo{year}{2018}\natexlab{a}.
\newblock \showarticletitle{Multiwinner Voting with Fairness Constraints}. In
  \bibinfo{booktitle}{\emph{{IJCAI}}}. \bibinfo{pages}{144--151}.
\newblock


\bibitem[\protect\citeauthoryear{Celis, Keswani, Straszak, Deshpande, Kathuria,
  and Vishnoi}{Celis et~al\mbox{.}}{2018b}]%
        {Celis:2018}
\bibfield{author}{\bibinfo{person}{L.~Elisa Celis}, \bibinfo{person}{Vijay
  Keswani}, \bibinfo{person}{Damian Straszak}, \bibinfo{person}{Amit
  Deshpande}, \bibinfo{person}{Tarun Kathuria}, {and}
  \bibinfo{person}{Nisheeth~K. Vishnoi}.} \bibinfo{year}{2018}\natexlab{b}.
\newblock \showarticletitle{Fair and Diverse DPP-Based Data Summarization}. In
  \bibinfo{booktitle}{\emph{{ICML}}}. \bibinfo{pages}{715--724}.
\newblock


\bibitem[\protect\citeauthoryear{Celis, Straszak, and Vishnoi}{Celis
  et~al\mbox{.}}{2018c}]%
        {DBLP:conf/icalp/CelisSV18}
\bibfield{author}{\bibinfo{person}{L.~Elisa Celis}, \bibinfo{person}{Damian
  Straszak}, {and} \bibinfo{person}{Nisheeth~K. Vishnoi}.}
  \bibinfo{year}{2018}\natexlab{c}.
\newblock \showarticletitle{Ranking with Fairness Constraints}. In
  \bibinfo{booktitle}{\emph{{ICALP}}}. \bibinfo{pages}{28:1--28:15}.
\newblock


\bibitem[\protect\citeauthoryear{Chester, Thomo, Venkatesh, and
  Whitesides}{Chester et~al\mbox{.}}{2014}]%
        {Chester:2014}
\bibfield{author}{\bibinfo{person}{Sean Chester}, \bibinfo{person}{Alex Thomo},
  \bibinfo{person}{S. Venkatesh}, {and} \bibinfo{person}{Sue Whitesides}.}
  \bibinfo{year}{2014}\natexlab{}.
\newblock \showarticletitle{Computing k-Regret Minimizing Sets}.
\newblock \bibinfo{journal}{\emph{Proc. {VLDB} Endow.}} \bibinfo{volume}{7},
  \bibinfo{number}{5} (\bibinfo{year}{2014}), \bibinfo{pages}{389--400}.
\newblock


\bibitem[\protect\citeauthoryear{Chouldechova and Roth}{Chouldechova and
  Roth}{2020}]%
        {Chouldechova:2020}
\bibfield{author}{\bibinfo{person}{Alexandra Chouldechova} {and}
  \bibinfo{person}{Aaron Roth}.} \bibinfo{year}{2020}\natexlab{}.
\newblock \showarticletitle{A snapshot of the frontiers of fairness in machine
  learning}.
\newblock \bibinfo{journal}{\emph{Commun. {ACM}}} \bibinfo{volume}{63},
  \bibinfo{number}{5} (\bibinfo{year}{2020}), \bibinfo{pages}{82--89}.
\newblock


\bibitem[\protect\citeauthoryear{Dong and Zheng}{Dong and Zheng}{2019}]%
        {Dong:2019}
\bibfield{author}{\bibinfo{person}{Qi Dong} {and} \bibinfo{person}{Jiping
  Zheng}.} \bibinfo{year}{2019}\natexlab{}.
\newblock \showarticletitle{Faster Algorithms for \emph{k}-Regret Minimizing
  Sets via Monotonicity and Sampling}. In \bibinfo{booktitle}{\emph{{CIKM}}}.
  \bibinfo{pages}{2213--2216}.
\newblock


\bibitem[\protect\citeauthoryear{Faulkner, Brackenbury, and Lall}{Faulkner
  et~al\mbox{.}}{2015}]%
        {Faulkner:2015}
\bibfield{author}{\bibinfo{person}{Taylor~Kessler Faulkner},
  \bibinfo{person}{Will Brackenbury}, {and} \bibinfo{person}{Ashwin Lall}.}
  \bibinfo{year}{2015}\natexlab{}.
\newblock \showarticletitle{k-Regret Queries with Nonlinear Utilities}.
\newblock \bibinfo{journal}{\emph{Proc. {VLDB} Endow.}} \bibinfo{volume}{8},
  \bibinfo{number}{13} (\bibinfo{year}{2015}), \bibinfo{pages}{2098--2109}.
\newblock


\bibitem[\protect\citeauthoryear{Feige}{Feige}{1998}]%
        {Feige:1998}
\bibfield{author}{\bibinfo{person}{Uriel Feige}.}
  \bibinfo{year}{1998}\natexlab{}.
\newblock \showarticletitle{A Threshold of ln \emph{n} for Approximating Set
  Cover}.
\newblock \bibinfo{journal}{\emph{J. {ACM}}} \bibinfo{volume}{45},
  \bibinfo{number}{4} (\bibinfo{year}{1998}), \bibinfo{pages}{634--652}.
\newblock


\bibitem[\protect\citeauthoryear{Fisher, Nemhauser, and Wolsey}{Fisher
  et~al\mbox{.}}{1978}]%
        {Fisher:1978}
\bibfield{author}{\bibinfo{person}{M.~L. Fisher}, \bibinfo{person}{G.~L.
  Nemhauser}, {and} \bibinfo{person}{L.~A. Wolsey}.}
  \bibinfo{year}{1978}\natexlab{}.
\newblock \showarticletitle{An analysis of approximations for maximizing
  submodular set functions---II}. In \bibinfo{booktitle}{\emph{Polyhedral
  Combinatorics}}. \bibinfo{pages}{73--87}.
\newblock


\bibitem[\protect\citeauthoryear{Friedman and Nissenbaum}{Friedman and
  Nissenbaum}{1996}]%
        {Friedman:1996}
\bibfield{author}{\bibinfo{person}{Batya Friedman} {and} \bibinfo{person}{Helen
  Nissenbaum}.} \bibinfo{year}{1996}\natexlab{}.
\newblock \showarticletitle{Bias in computer systems}.
\newblock \bibinfo{journal}{\emph{ACM Trans. Inf. Sys.}} \bibinfo{volume}{14},
  \bibinfo{number}{3} (\bibinfo{year}{1996}), \bibinfo{pages}{330--347}.
\newblock


\bibitem[\protect\citeauthoryear{Fujito}{Fujito}{2000}]%
        {Fujito:2000}
\bibfield{author}{\bibinfo{person}{Toshihiro Fujito}.}
  \bibinfo{year}{2000}\natexlab{}.
\newblock \showarticletitle{Approximation algorithms for submodular set cover
  with applications}.
\newblock \bibinfo{journal}{\emph{{IEICE} Trans. Inf. Syst.}}
  \bibinfo{volume}{83} (\bibinfo{year}{2000}), \bibinfo{pages}{480--487}.
\newblock


\bibitem[\protect\citeauthoryear{Garc{\'{\i}}a{-}Soriano and
  Bonchi}{Garc{\'{\i}}a{-}Soriano and Bonchi}{2021}]%
        {DBLP:conf/kdd/Garcia-SorianoB21}
\bibfield{author}{\bibinfo{person}{David Garc{\'{\i}}a{-}Soriano} {and}
  \bibinfo{person}{Francesco Bonchi}.} \bibinfo{year}{2021}\natexlab{}.
\newblock \showarticletitle{Maxmin-Fair Ranking: Individual Fairness under
  Group-Fairness Constraints}. In \bibinfo{booktitle}{\emph{{KDD}}}.
  \bibinfo{pages}{436--446}.
\newblock


\bibitem[\protect\citeauthoryear{Halabi, Mitrovic, Norouzi{-}Fard, Tardos, and
  Tarnawski}{Halabi et~al\mbox{.}}{2020}]%
        {DBLP:conf/nips/HalabiMNTT20}
\bibfield{author}{\bibinfo{person}{Marwa~El Halabi}, \bibinfo{person}{Slobodan
  Mitrovic}, \bibinfo{person}{Ashkan Norouzi{-}Fard}, \bibinfo{person}{Jakab
  Tardos}, {and} \bibinfo{person}{Jakub Tarnawski}.}
  \bibinfo{year}{2020}\natexlab{}.
\newblock \showarticletitle{Fairness in Streaming Submodular Maximization:
  Algorithms and Hardness}. In \bibinfo{booktitle}{\emph{NeurIPS}}.
  \bibinfo{pages}{13609--13622}.
\newblock


\bibitem[\protect\citeauthoryear{Ilyas, Beskales, and Soliman}{Ilyas
  et~al\mbox{.}}{2008}]%
        {Ilyas:2008}
\bibfield{author}{\bibinfo{person}{Ihab~F. Ilyas}, \bibinfo{person}{George
  Beskales}, {and} \bibinfo{person}{Mohamed~A. Soliman}.}
  \bibinfo{year}{2008}\natexlab{}.
\newblock \showarticletitle{A survey of top-\emph{k} query processing
  techniques in relational database systems}.
\newblock \bibinfo{journal}{\emph{{ACM} Comput. Surv.}} \bibinfo{volume}{40},
  \bibinfo{number}{4} (\bibinfo{year}{2008}), \bibinfo{pages}{11:1--11:58}.
\newblock


\bibitem[\protect\citeauthoryear{Kleinberg and Tardos}{Kleinberg and
  Tardos}{2006}]%
        {Kleinberg:2006}
\bibfield{author}{\bibinfo{person}{Jon Kleinberg} {and} \bibinfo{person}{\'Eva
  Tardos}.} \bibinfo{year}{2006}\natexlab{}.
\newblock \bibinfo{booktitle}{\emph{Algorithm Design}}.
\newblock \bibinfo{publisher}{Addison Wesley}.
\newblock


\bibitem[\protect\citeauthoryear{Kleindessner, Awasthi, and
  Morgenstern}{Kleindessner et~al\mbox{.}}{2019}]%
        {DBLP:conf/icml/KleindessnerAM19}
\bibfield{author}{\bibinfo{person}{Matth{\"{a}}us Kleindessner},
  \bibinfo{person}{Pranjal Awasthi}, {and} \bibinfo{person}{Jamie
  Morgenstern}.} \bibinfo{year}{2019}\natexlab{}.
\newblock \showarticletitle{Fair k-Center Clustering for Data Summarization}.
  In \bibinfo{booktitle}{\emph{{ICML}}}. \bibinfo{pages}{3448--3457}.
\newblock


\bibitem[\protect\citeauthoryear{Korte and Vygen}{Korte and Vygen}{2012}]%
        {Korte2012}
\bibfield{author}{\bibinfo{person}{Bernhard Korte} {and} \bibinfo{person}{Jens
  Vygen}.} \bibinfo{year}{2012}\natexlab{}.
\newblock \bibinfo{booktitle}{\emph{Combinatorial Optimization: Theory and
  Algorithms}}.
\newblock \bibinfo{publisher}{Springer Berlin Heidelberg},
  \bibinfo{address}{Berlin, Heidelberg}.
\newblock


\bibitem[\protect\citeauthoryear{Krause and Golovin}{Krause and
  Golovin}{2014}]%
        {Krause:2014}
\bibfield{author}{\bibinfo{person}{Andreas Krause} {and}
  \bibinfo{person}{Daniel Golovin}.} \bibinfo{year}{2014}\natexlab{}.
\newblock \showarticletitle{Submodular Function Maximization}.
\newblock In \bibinfo{booktitle}{\emph{Tractability: Practical Approaches to
  Hard Problems}}. \bibinfo{publisher}{Cambridge University Press},
  \bibinfo{pages}{71--104}.
\newblock


\bibitem[\protect\citeauthoryear{Krause, McMahan, Guestrin, and Gupta}{Krause
  et~al\mbox{.}}{2008}]%
        {Krause:2008}
\bibfield{author}{\bibinfo{person}{Andreas Krause}, \bibinfo{person}{H~Brendan
  McMahan}, \bibinfo{person}{Carlos Guestrin}, {and} \bibinfo{person}{Anupam
  Gupta}.} \bibinfo{year}{2008}\natexlab{}.
\newblock \showarticletitle{Robust Submodular Observation Selection}.
\newblock \bibinfo{journal}{\emph{J. Mach. Learn. Res.}}  \bibinfo{volume}{9}
  (\bibinfo{year}{2008}), \bibinfo{pages}{2761--2801}.
\newblock


\bibitem[\protect\citeauthoryear{Kuhlman and Rundensteiner}{Kuhlman and
  Rundensteiner}{2020}]%
        {DBLP:journals/pvldb/KuhlmanR20}
\bibfield{author}{\bibinfo{person}{Caitlin Kuhlman} {and}
  \bibinfo{person}{Elke~A. Rundensteiner}.} \bibinfo{year}{2020}\natexlab{}.
\newblock \showarticletitle{Rank Aggregation Algorithms for Fair Consensus}.
\newblock \bibinfo{journal}{\emph{Proc. {VLDB} Endow.}} \bibinfo{volume}{13},
  \bibinfo{number}{11} (\bibinfo{year}{2020}), \bibinfo{pages}{2706--2719}.
\newblock


\bibitem[\protect\citeauthoryear{Kumar and Sintos}{Kumar and Sintos}{2018}]%
        {Kumar:2018}
\bibfield{author}{\bibinfo{person}{Nirman Kumar} {and} \bibinfo{person}{Stavros
  Sintos}.} \bibinfo{year}{2018}\natexlab{}.
\newblock \showarticletitle{Faster Approximation Algorithm for the
  \emph{k}-Regret Minimizing Set and Related Problems}. In
  \bibinfo{booktitle}{\emph{ALENEX}}. \bibinfo{pages}{62--74}.
\newblock


\bibitem[\protect\citeauthoryear{Luenam, Chen, and Wong}{Luenam
  et~al\mbox{.}}{2021}]%
        {Luenam:2021}
\bibfield{author}{\bibinfo{person}{Phoomraphee Luenam},
  \bibinfo{person}{Yau~Pun Chen}, {and} \bibinfo{person}{Raymond Chi-Wing
  Wong}.} \bibinfo{year}{2021}\natexlab{}.
\newblock \bibinfo{title}{Approximating Happiness Maximizing Set Problems}.
\newblock
\newblock
\showeprint[arxiv]{2102.03578}~[cs.DB]


\bibitem[\protect\citeauthoryear{Ma, Guan, Toomey, and Wu}{Ma
  et~al\mbox{.}}{2022}]%
        {DBLP:conf/wsdm/MaGTW22}
\bibfield{author}{\bibinfo{person}{Hanchao Ma}, \bibinfo{person}{Sheng Guan},
  \bibinfo{person}{Christopher Toomey}, {and} \bibinfo{person}{Yinghui Wu}.}
  \bibinfo{year}{2022}\natexlab{}.
\newblock \showarticletitle{Diversified Subgraph Query Generation with Group
  Fairness}. In \bibinfo{booktitle}{\emph{{WSDM}}}. \bibinfo{pages}{686--694}.
\newblock


\bibitem[\protect\citeauthoryear{Mehrotra and Celis}{Mehrotra and
  Celis}{2021}]%
        {DBLP:conf/fat/MehrotraC21}
\bibfield{author}{\bibinfo{person}{Anay Mehrotra} {and}
  \bibinfo{person}{L.~Elisa Celis}.} \bibinfo{year}{2021}\natexlab{}.
\newblock \showarticletitle{Mitigating Bias in Set Selection with Noisy
  Protected Attributes}. In \bibinfo{booktitle}{\emph{FAccT}}.
  \bibinfo{pages}{237--248}.
\newblock


\bibitem[\protect\citeauthoryear{Moumoulidou, McGregor, and Meliou}{Moumoulidou
  et~al\mbox{.}}{2021}]%
        {Moumoulidou:2021}
\bibfield{author}{\bibinfo{person}{Zafeiria Moumoulidou},
  \bibinfo{person}{Andrew McGregor}, {and} \bibinfo{person}{Alexandra Meliou}.}
  \bibinfo{year}{2021}\natexlab{}.
\newblock \showarticletitle{Diverse Data Selection under Fairness Constraints}.
  In \bibinfo{booktitle}{\emph{{ICDT}}}. \bibinfo{pages}{13:1--13:25}.
\newblock


\bibitem[\protect\citeauthoryear{Nanongkai, Lall, Sarma, and Makino}{Nanongkai
  et~al\mbox{.}}{2012}]%
        {Nanongkai:2012}
\bibfield{author}{\bibinfo{person}{Danupon Nanongkai}, \bibinfo{person}{Ashwin
  Lall}, \bibinfo{person}{Atish~Das Sarma}, {and} \bibinfo{person}{Kazuhisa
  Makino}.} \bibinfo{year}{2012}\natexlab{}.
\newblock \showarticletitle{Interactive regret minimization}. In
  \bibinfo{booktitle}{\emph{{SIGMOD}}}. \bibinfo{pages}{109--120}.
\newblock


\bibitem[\protect\citeauthoryear{Nanongkai, Sarma, Lall, Lipton, and
  Xu}{Nanongkai et~al\mbox{.}}{2010}]%
        {Nanongkai:2010}
\bibfield{author}{\bibinfo{person}{Danupon Nanongkai},
  \bibinfo{person}{Atish~Das Sarma}, \bibinfo{person}{Ashwin Lall},
  \bibinfo{person}{Richard~J. Lipton}, {and} \bibinfo{person}{Jun~(Jim) Xu}.}
  \bibinfo{year}{2010}\natexlab{}.
\newblock \showarticletitle{Regret-Minimizing Representative Databases}.
\newblock \bibinfo{journal}{\emph{Proc. {VLDB} Endow.}} \bibinfo{volume}{3},
  \bibinfo{number}{1} (\bibinfo{year}{2010}), \bibinfo{pages}{1114--1124}.
\newblock


\bibitem[\protect\citeauthoryear{Peng and Wong}{Peng and Wong}{2014}]%
        {Peng:2014}
\bibfield{author}{\bibinfo{person}{Peng Peng} {and}
  \bibinfo{person}{Raymond~Chi{-}Wing Wong}.} \bibinfo{year}{2014}\natexlab{}.
\newblock \showarticletitle{Geometry approach for k-regret query}. In
  \bibinfo{booktitle}{\emph{{ICDE}}}. \bibinfo{pages}{772--783}.
\newblock


\bibitem[\protect\citeauthoryear{Pitoura, Stefanidis, and Koutrika}{Pitoura
  et~al\mbox{.}}{2022}]%
        {Pitoura:2021}
\bibfield{author}{\bibinfo{person}{Evaggelia Pitoura}, \bibinfo{person}{Kostas
  Stefanidis}, {and} \bibinfo{person}{Georgia Koutrika}.}
  \bibinfo{year}{2022}\natexlab{}.
\newblock \showarticletitle{Fairness in rankings and recommendations: an
  overview}.
\newblock \bibinfo{journal}{\emph{{VLDB} J.}} \bibinfo{volume}{31},
  \bibinfo{number}{3} (\bibinfo{year}{2022}), \bibinfo{pages}{431--458}.
\newblock


\bibitem[\protect\citeauthoryear{Qi, Zuo, Samet, and Yao}{Qi
  et~al\mbox{.}}{2018}]%
        {Qi:2018}
\bibfield{author}{\bibinfo{person}{Jianzhong Qi}, \bibinfo{person}{Fei Zuo},
  \bibinfo{person}{Hanan Samet}, {and} \bibinfo{person}{Jia~Cheng Yao}.}
  \bibinfo{year}{2018}\natexlab{}.
\newblock \showarticletitle{K-Regret Queries Using Multiplicative Utility
  Functions}.
\newblock \bibinfo{journal}{\emph{{ACM} Trans. Database Syst.}}
  \bibinfo{volume}{43}, \bibinfo{number}{2} (\bibinfo{year}{2018}),
  \bibinfo{pages}{10:1--10:41}.
\newblock


\bibitem[\protect\citeauthoryear{Qiu, Zheng, Dong, and Huang}{Qiu
  et~al\mbox{.}}{2018}]%
        {Qiu:2018}
\bibfield{author}{\bibinfo{person}{Xianhong Qiu}, \bibinfo{person}{Jiping
  Zheng}, \bibinfo{person}{Qi Dong}, {and} \bibinfo{person}{Xingnan Huang}.}
  \bibinfo{year}{2018}\natexlab{}.
\newblock \showarticletitle{Speed-Up Algorithms for Happiness-Maximizing
  Representative Databases}. In \bibinfo{booktitle}{\emph{APWeb/WAIM
  Workshops}}. \bibinfo{pages}{321--335}.
\newblock


\bibitem[\protect\citeauthoryear{Saff and Kuijlaars}{Saff and
  Kuijlaars}{1997}]%
        {Saff:1997}
\bibfield{author}{\bibinfo{person}{Edward~B. Saff} {and} \bibinfo{person}{Amo
  B.~J. Kuijlaars}.} \bibinfo{year}{1997}\natexlab{}.
\newblock \showarticletitle{Distributing many points on a sphere}.
\newblock \bibinfo{journal}{\emph{Math. Intell.}} \bibinfo{volume}{19},
  \bibinfo{number}{1} (\bibinfo{year}{1997}), \bibinfo{pages}{5--11}.
\newblock


\bibitem[\protect\citeauthoryear{Shetiya, Asudeh, Ahmed, and Das}{Shetiya
  et~al\mbox{.}}{2019}]%
        {Shetiya:2020}
\bibfield{author}{\bibinfo{person}{Suraj Shetiya}, \bibinfo{person}{Abolfazl
  Asudeh}, \bibinfo{person}{Sadia Ahmed}, {and} \bibinfo{person}{Gautam Das}.}
  \bibinfo{year}{2019}\natexlab{}.
\newblock \showarticletitle{A Unified Optimization Algorithm For Solving
  ``Regret-Minimizing Representative'' Problems}.
\newblock \bibinfo{journal}{\emph{Proc. {VLDB} Endow.}} \bibinfo{volume}{13},
  \bibinfo{number}{3} (\bibinfo{year}{2019}), \bibinfo{pages}{239--251}.
\newblock


\bibitem[\protect\citeauthoryear{Singh and Joachims}{Singh and
  Joachims}{2018}]%
        {DBLP:conf/kdd/SinghJ18}
\bibfield{author}{\bibinfo{person}{Ashudeep Singh} {and}
  \bibinfo{person}{Thorsten Joachims}.} \bibinfo{year}{2018}\natexlab{}.
\newblock \showarticletitle{Fairness of Exposure in Rankings}. In
  \bibinfo{booktitle}{\emph{{KDD}}}. \bibinfo{pages}{2219--2228}.
\newblock


\bibitem[\protect\citeauthoryear{Soma and Yoshida}{Soma and Yoshida}{2017}]%
        {Soma:2017}
\bibfield{author}{\bibinfo{person}{Tasuku Soma} {and} \bibinfo{person}{Yuichi
  Yoshida}.} \bibinfo{year}{2017}\natexlab{}.
\newblock \showarticletitle{Regret Ratio Minimization in Multi-Objective
  Submodular Function Maximization}. In \bibinfo{booktitle}{\emph{{AAAI}}}.
  \bibinfo{pages}{905--911}.
\newblock


\bibitem[\protect\citeauthoryear{Storandt and Funke}{Storandt and
  Funke}{2019}]%
        {Storandt:2019}
\bibfield{author}{\bibinfo{person}{Sabine Storandt} {and}
  \bibinfo{person}{Stefan Funke}.} \bibinfo{year}{2019}\natexlab{}.
\newblock \showarticletitle{Algorithms for Average Regret Minimization}. In
  \bibinfo{booktitle}{\emph{{AAAI}}}. \bibinfo{pages}{1600--1607}.
\newblock


\bibitem[\protect\citeauthoryear{Stoyanovich, Howe, and Jagadish}{Stoyanovich
  et~al\mbox{.}}{2020}]%
        {Stoyanovich:2020}
\bibfield{author}{\bibinfo{person}{Julia Stoyanovich}, \bibinfo{person}{Bill
  Howe}, {and} \bibinfo{person}{H.~V. Jagadish}.}
  \bibinfo{year}{2020}\natexlab{}.
\newblock \showarticletitle{Responsible Data Management}.
\newblock \bibinfo{journal}{\emph{Proc. {VLDB} Endow.}} \bibinfo{volume}{13},
  \bibinfo{number}{12} (\bibinfo{year}{2020}), \bibinfo{pages}{3474--3488}.
\newblock


\bibitem[\protect\citeauthoryear{Stoyanovich, Yang, and Jagadish}{Stoyanovich
  et~al\mbox{.}}{2018}]%
        {Stoyanovich:2018}
\bibfield{author}{\bibinfo{person}{Julia Stoyanovich}, \bibinfo{person}{Ke
  Yang}, {and} \bibinfo{person}{H.~V. Jagadish}.}
  \bibinfo{year}{2018}\natexlab{}.
\newblock \showarticletitle{Online Set Selection with Fairness and Diversity
  Constraints}. In \bibinfo{booktitle}{\emph{{EDBT}}}.
  \bibinfo{pages}{241--252}.
\newblock


\bibitem[\protect\citeauthoryear{Torrico, Singh, Pokutta, Haghtalab, Naor, and
  Anari}{Torrico et~al\mbox{.}}{2021}]%
        {Torrico:2021}
\bibfield{author}{\bibinfo{person}{Alfredo Torrico}, \bibinfo{person}{Mohit
  Singh}, \bibinfo{person}{Sebastian Pokutta}, \bibinfo{person}{Nika
  Haghtalab}, \bibinfo{person}{Joseph~(Seffi) Naor}, {and}
  \bibinfo{person}{Nima Anari}.} \bibinfo{year}{2021}\natexlab{}.
\newblock \showarticletitle{Structured Robust Submodular Maximization: Offline
  and Online Algorithms}.
\newblock \bibinfo{journal}{\emph{INFORMS J. Comput.}} \bibinfo{volume}{33},
  \bibinfo{number}{4} (\bibinfo{year}{2021}), \bibinfo{pages}{1590--1607}.
\newblock


\bibitem[\protect\citeauthoryear{Udwani}{Udwani}{2018}]%
        {DBLP:conf/nips/Udwani18}
\bibfield{author}{\bibinfo{person}{Rajan Udwani}.}
  \bibinfo{year}{2018}\natexlab{}.
\newblock \showarticletitle{Multi-objective Maximization of Monotone Submodular
  Functions with Cardinality Constraint}. In
  \bibinfo{booktitle}{\emph{NeurIPS}}. \bibinfo{pages}{9513--9524}.
\newblock


\bibitem[\protect\citeauthoryear{Wang, Fabbri, and Mathioudakis}{Wang
  et~al\mbox{.}}{2021a}]%
        {Wang:2021a}
\bibfield{author}{\bibinfo{person}{Yanhao Wang}, \bibinfo{person}{Francesco
  Fabbri}, {and} \bibinfo{person}{Michael Mathioudakis}.}
  \bibinfo{year}{2021}\natexlab{a}.
\newblock \showarticletitle{Fair and Representative Subset Selection from Data
  Streams}. In \bibinfo{booktitle}{\emph{{WWW}}}. \bibinfo{pages}{1340--1350}.
\newblock


\bibitem[\protect\citeauthoryear{Wang, Li, Wong, and Tan}{Wang
  et~al\mbox{.}}{2021b}]%
        {Wang:2021b}
\bibfield{author}{\bibinfo{person}{Yanhao Wang}, \bibinfo{person}{Yuchen Li},
  \bibinfo{person}{Raymond~Chi{-}Wing Wong}, {and} \bibinfo{person}{Kian{-}Lee
  Tan}.} \bibinfo{year}{2021}\natexlab{b}.
\newblock \showarticletitle{A Fully Dynamic Algorithm for k-Regret Minimizing
  Sets}. In \bibinfo{booktitle}{\emph{{ICDE}}}. \bibinfo{pages}{1631--1642}.
\newblock


\bibitem[\protect\citeauthoryear{Wang, Mathioudakis, Li, and Tan}{Wang
  et~al\mbox{.}}{2021c}]%
        {DBLP:conf/pods/WangM0T21}
\bibfield{author}{\bibinfo{person}{Yanhao Wang}, \bibinfo{person}{Michael
  Mathioudakis}, \bibinfo{person}{Yuchen Li}, {and} \bibinfo{person}{Kian{-}Lee
  Tan}.} \bibinfo{year}{2021}\natexlab{c}.
\newblock \showarticletitle{Minimum Coresets for Maxima Representation of
  Multidimensional Data}. In \bibinfo{booktitle}{\emph{{PODS}}}.
  \bibinfo{pages}{138--152}.
\newblock


\bibitem[\protect\citeauthoryear{Xiao and Li}{Xiao and Li}{2022}]%
        {Xiao:2021}
\bibfield{author}{\bibinfo{person}{Xingxing Xiao} {and}
  \bibinfo{person}{Jianzhong Li}.} \bibinfo{year}{2022}\natexlab{}.
\newblock \showarticletitle{Rank-Regret Minimization}. In
  \bibinfo{booktitle}{\emph{{ICDE}}}. \bibinfo{pages}{1848--1860}.
\newblock


\bibitem[\protect\citeauthoryear{Xie, Wong, and Lall}{Xie
  et~al\mbox{.}}{2019}]%
        {Xie:2019}
\bibfield{author}{\bibinfo{person}{Min Xie},
  \bibinfo{person}{Raymond~Chi{-}Wing Wong}, {and} \bibinfo{person}{Ashwin
  Lall}.} \bibinfo{year}{2019}\natexlab{}.
\newblock \showarticletitle{Strongly Truthful Interactive Regret Minimization}.
  In \bibinfo{booktitle}{\emph{{SIGMOD}}}. \bibinfo{pages}{281--298}.
\newblock


\bibitem[\protect\citeauthoryear{Xie, Wong, and Lall}{Xie
  et~al\mbox{.}}{2020a}]%
        {Xie:2020VLDBJ}
\bibfield{author}{\bibinfo{person}{Min Xie},
  \bibinfo{person}{Raymond~Chi{-}Wing Wong}, {and} \bibinfo{person}{Ashwin
  Lall}.} \bibinfo{year}{2020}\natexlab{a}.
\newblock \showarticletitle{An experimental survey of regret minimization query
  and variants: bridging the best worlds between top-k query and skyline
  query}.
\newblock \bibinfo{journal}{\emph{{VLDB} J.}} \bibinfo{volume}{29},
  \bibinfo{number}{1} (\bibinfo{year}{2020}), \bibinfo{pages}{147--175}.
\newblock


\bibitem[\protect\citeauthoryear{Xie, Wong, Li, Long, and Lall}{Xie
  et~al\mbox{.}}{2018}]%
        {Xie:2018}
\bibfield{author}{\bibinfo{person}{Min Xie},
  \bibinfo{person}{Raymond~Chi{-}Wing Wong}, \bibinfo{person}{Jian Li},
  \bibinfo{person}{Cheng Long}, {and} \bibinfo{person}{Ashwin Lall}.}
  \bibinfo{year}{2018}\natexlab{}.
\newblock \showarticletitle{Efficient k-Regret Query Algorithm with
  Restriction-free Bound for any Dimensionality}. In
  \bibinfo{booktitle}{\emph{{SIGMOD}}}. \bibinfo{pages}{959--974}.
\newblock


\bibitem[\protect\citeauthoryear{Xie, Wong, Peng, and Tsotras}{Xie
  et~al\mbox{.}}{2020b}]%
        {Xie:2020}
\bibfield{author}{\bibinfo{person}{Min Xie},
  \bibinfo{person}{Raymond~Chi{-}Wing Wong}, \bibinfo{person}{Peng Peng}, {and}
  \bibinfo{person}{Vassilis~J. Tsotras}.} \bibinfo{year}{2020}\natexlab{b}.
\newblock \showarticletitle{Being Happy with the Least: Achieving
  {\(\alpha\)}-happiness with Minimum Number of Tuples}. In
  \bibinfo{booktitle}{\emph{{ICDE}}}. \bibinfo{pages}{1009--1020}.
\newblock


\bibitem[\protect\citeauthoryear{Zehlike, Bonchi, Castillo, Hajian, Megahed,
  and Baeza{-}Yates}{Zehlike et~al\mbox{.}}{2017}]%
        {DBLP:conf/cikm/ZehlikeB0HMB17}
\bibfield{author}{\bibinfo{person}{Meike Zehlike}, \bibinfo{person}{Francesco
  Bonchi}, \bibinfo{person}{Carlos Castillo}, \bibinfo{person}{Sara Hajian},
  \bibinfo{person}{Mohamed Megahed}, {and} \bibinfo{person}{Ricardo
  Baeza{-}Yates}.} \bibinfo{year}{2017}\natexlab{}.
\newblock \showarticletitle{FA*IR: {A} Fair Top-k Ranking Algorithm}. In
  \bibinfo{booktitle}{\emph{{CIKM}}}. \bibinfo{pages}{1569--1578}.
\newblock


\bibitem[\protect\citeauthoryear{Zehlike, Yang, and Stoyanovich}{Zehlike
  et~al\mbox{.}}{2021}]%
        {Zehlike:2021}
\bibfield{author}{\bibinfo{person}{Meike Zehlike}, \bibinfo{person}{Ke Yang},
  {and} \bibinfo{person}{Julia Stoyanovich}.} \bibinfo{year}{2021}\natexlab{}.
\newblock \showarticletitle{Fairness in Ranking: A Survey}.
\newblock  (\bibinfo{year}{2021}).
\newblock
\showeprint[arxiv]{2103.14000}~[cs.IR]


\bibitem[\protect\citeauthoryear{Zeighami and Wong}{Zeighami and Wong}{2019}]%
        {Zeighami:2019}
\bibfield{author}{\bibinfo{person}{Sepanta Zeighami} {and}
  \bibinfo{person}{Raymond~Chi{-}Wing Wong}.} \bibinfo{year}{2019}\natexlab{}.
\newblock \showarticletitle{Finding Average Regret Ratio Minimizing Set in
  Database}. In \bibinfo{booktitle}{\emph{{ICDE}}}.
  \bibinfo{pages}{1722--1725}.
\newblock


\bibitem[\protect\citeauthoryear{Zheng and Chen}{Zheng and Chen}{2020}]%
        {DBLP:conf/apweb/0001C20}
\bibfield{author}{\bibinfo{person}{Jiping Zheng} {and} \bibinfo{person}{Chen
  Chen}.} \bibinfo{year}{2020}\natexlab{}.
\newblock \showarticletitle{Sorting-Based Interactive Regret Minimization}. In
  \bibinfo{booktitle}{\emph{APWeb/WAIM}}. \bibinfo{pages}{473--490}.
\newblock


\bibitem[\protect\citeauthoryear{Zheng, Wang, Wang, and Ma}{Zheng
  et~al\mbox{.}}{2022}]%
        {9756312}
\bibfield{author}{\bibinfo{person}{Jiping Zheng}, \bibinfo{person}{Yanhao
  Wang}, \bibinfo{person}{Xiaoyang Wang}, {and} \bibinfo{person}{Wei Ma}.}
  \bibinfo{year}{2022}\natexlab{}.
\newblock \showarticletitle{Continuous k-Regret Minimization Queries: A Dynamic
  Coreset Approach}.
\newblock \bibinfo{journal}{\emph{IEEE Trans. Knowl. Data Eng.}}
  (\bibinfo{year}{2022}).
\newblock
\urldef\tempurl%
\url{https://doi.org/10.1109/TKDE.2022.3166835}
\showDOI{\tempurl}


\bibitem[\protect\citeauthoryear{Zliobaite}{Zliobaite}{2017}]%
        {DBLP:journals/datamine/Zliobaite17}
\bibfield{author}{\bibinfo{person}{Indre Zliobaite}.}
  \bibinfo{year}{2017}\natexlab{}.
\newblock \showarticletitle{Measuring discrimination in algorithmic decision
  making}.
\newblock \bibinfo{journal}{\emph{Data Min. Knowl. Discov.}}
  \bibinfo{volume}{31}, \bibinfo{number}{4} (\bibinfo{year}{2017}),
  \bibinfo{pages}{1060--1089}.
\newblock


\end{thebibliography}

\clearpage
\appendix

\begin{figure*}[t]
    \centering
    \includegraphics[height=0.15in]{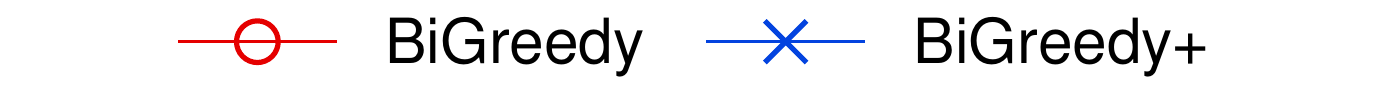}
    \\
    \begin{subfigure}[b]{0.195\textwidth}
        \centering
        \includegraphics[width=\textwidth]{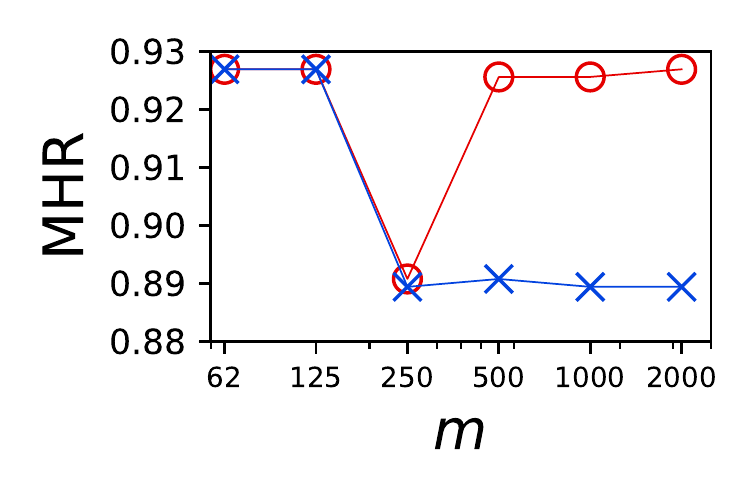}
        \caption{Adult (Gender)}
    \end{subfigure}
    \hfill
    \begin{subfigure}[b]{0.195\textwidth}
        \centering
        \includegraphics[width=\textwidth]{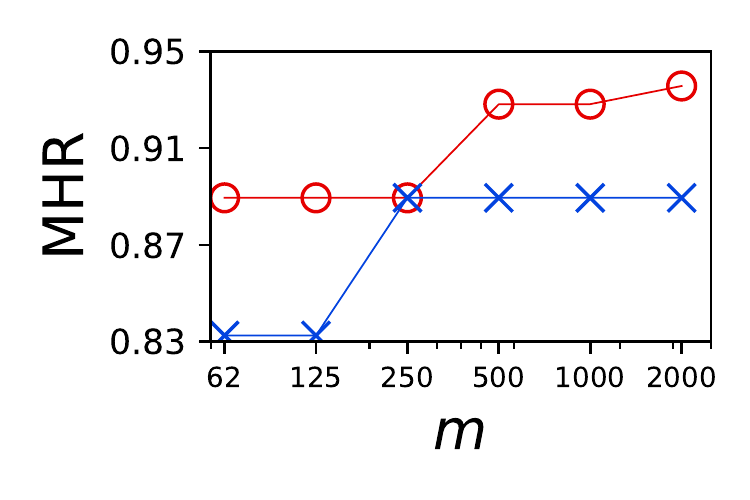}
        \caption{Adult (Race)}
    \end{subfigure}
    \hfill
    \begin{subfigure}[b]{0.195\textwidth}
        \centering
        \includegraphics[width=\textwidth]{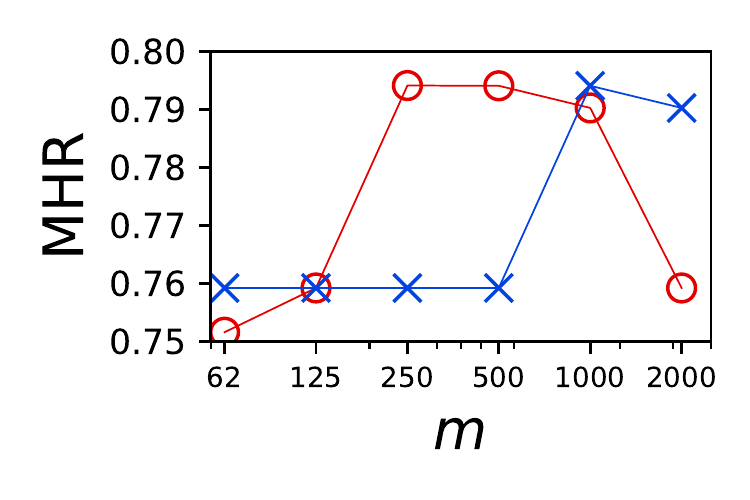}
        \caption{Adult (G+R)}
    \end{subfigure}
    \hfill
    \begin{subfigure}[b]{0.195\textwidth}
        \centering
        \includegraphics[width=\textwidth]{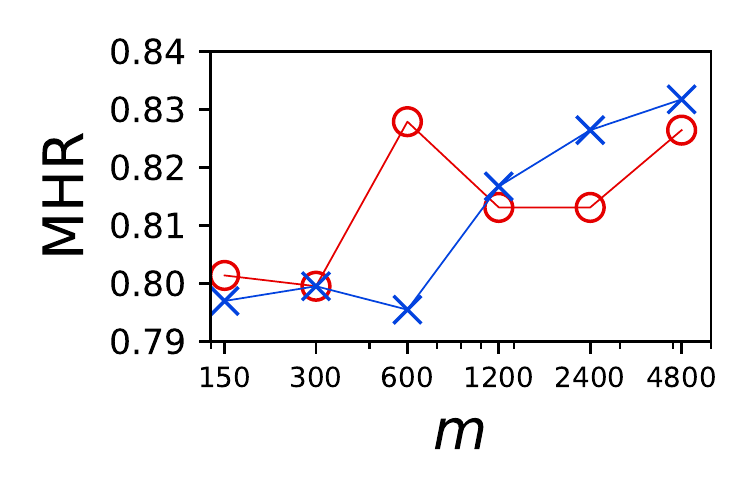}
        \caption{AntiCor\_6D}
    \end{subfigure}
    \hfill
    \begin{subfigure}[b]{0.195\textwidth}
        \centering
        \includegraphics[width=\textwidth]{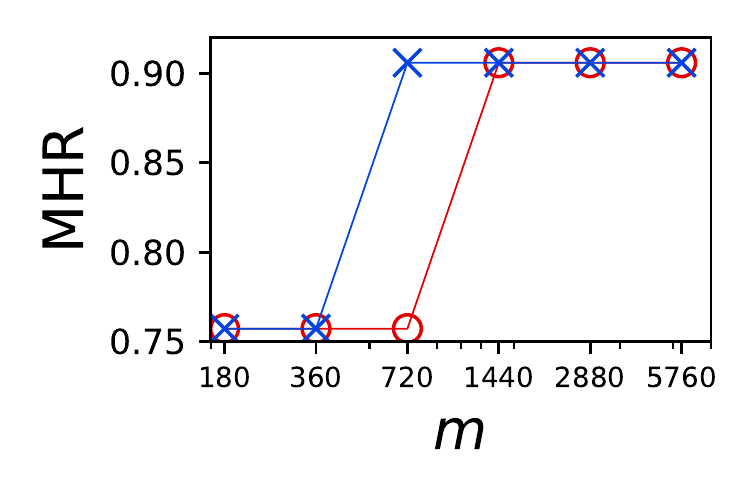}
        \caption{Compas (Gender)}
    \end{subfigure}
    \begin{subfigure}[b]{0.195\textwidth}
        \centering
        \includegraphics[width=\textwidth]{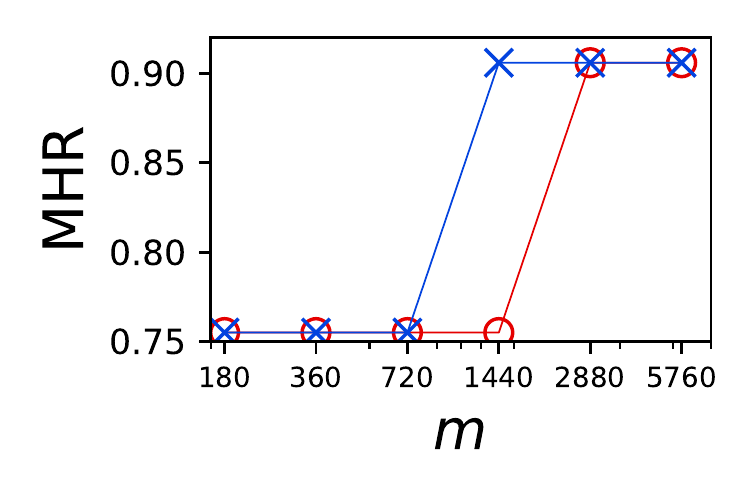}
        \caption{Compas (isRecid)}
    \end{subfigure}
    \hfill
    \begin{subfigure}[b]{0.195\textwidth}
        \centering
        \includegraphics[width=\textwidth]{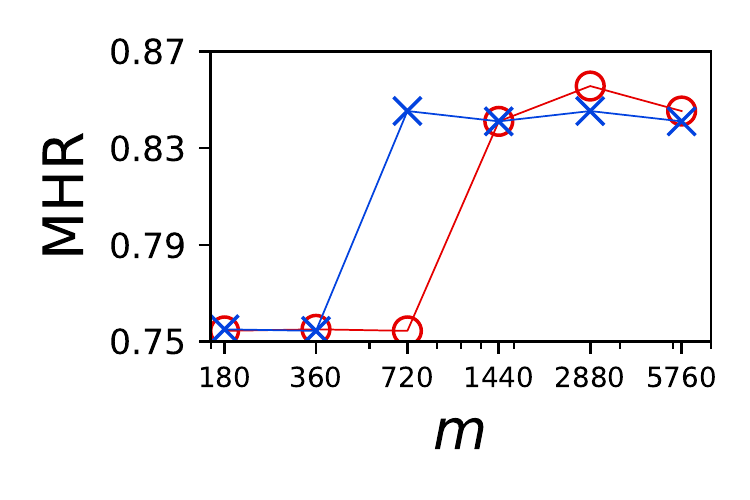}
        \caption{Compas (G+iR)}
    \end{subfigure}
    \hfill
    \begin{subfigure}[b]{0.195\textwidth}
        \centering
        \includegraphics[width=\textwidth]{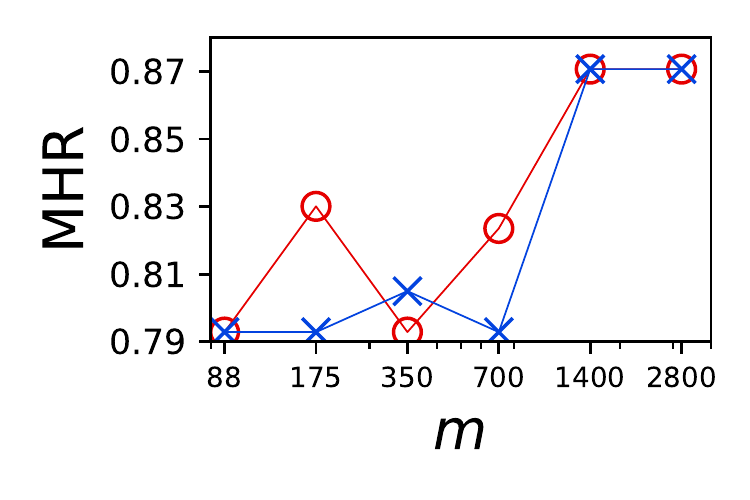}
        \caption{Credit (Job)}
    \end{subfigure}
    \hfill
    \begin{subfigure}[b]{0.195\textwidth}
        \centering
        \includegraphics[width=\textwidth]{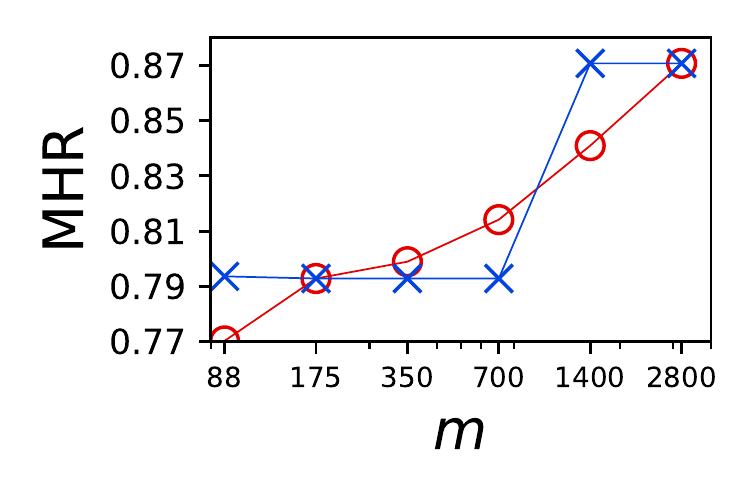}
        \caption{Credit (Housing)}
    \end{subfigure}
    \hfill
    \begin{subfigure}[b]{0.195\textwidth}
        \centering
        \includegraphics[width=\textwidth]{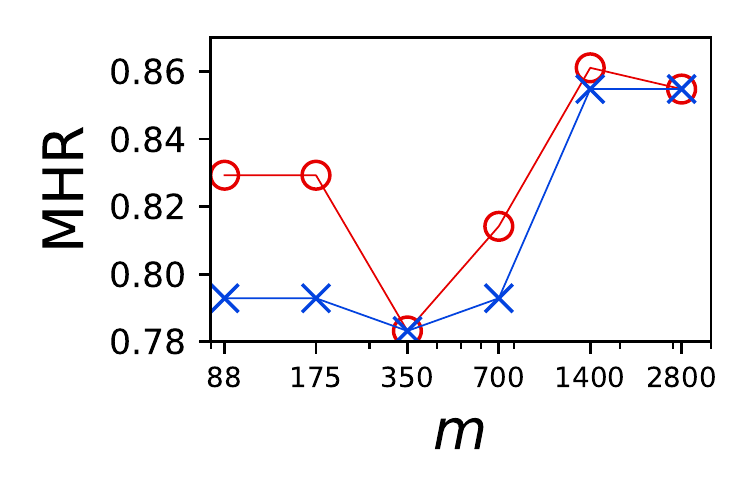}
        \caption{Credit (WY)}
    \end{subfigure}
    \caption{Results for the MHRs of \AlgBG and \AlgIBG by varying sample size $m$ or maximum sample size $M$.}
    \Description{experimental results}
    \label{fig:delta:mhr}
\end{figure*}

\begin{figure*}[t]
    \centering
    \includegraphics[height=0.15in]{figs_new/VaryDelta/legend-md.pdf}
    \\
    \begin{subfigure}[b]{0.195\textwidth}
        \centering
        \includegraphics[width=\textwidth]{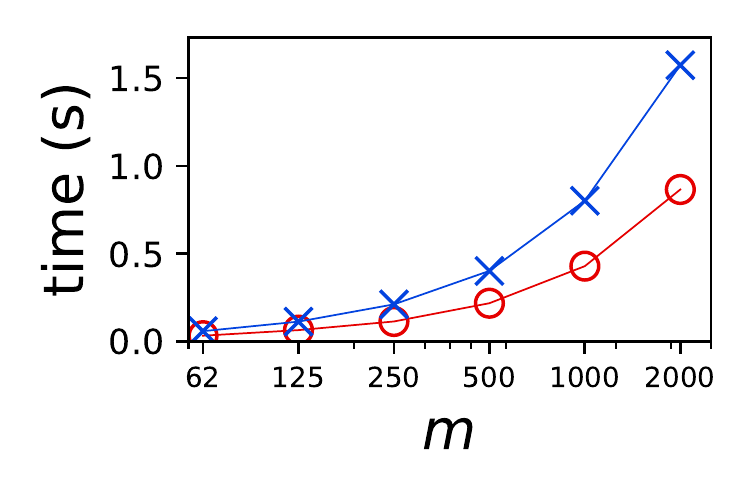}
        \caption{Adult (Gender)}
    \end{subfigure}
    \hfill
    \begin{subfigure}[b]{0.195\textwidth}
        \centering
        \includegraphics[width=\textwidth]{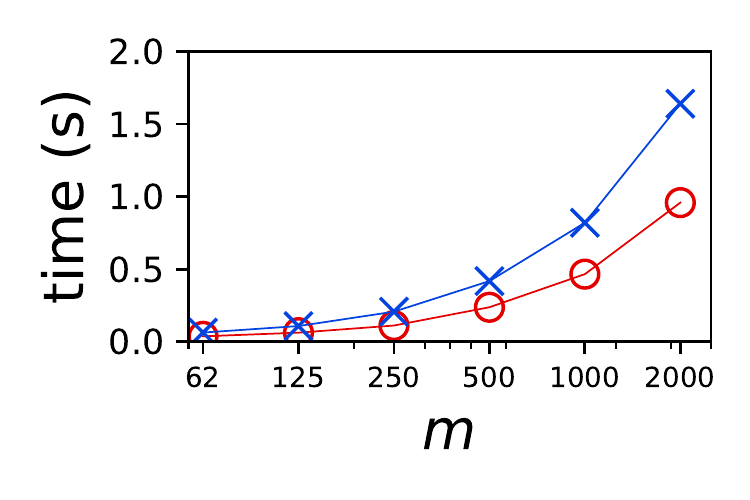}
        \caption{Adult (Race)}
    \end{subfigure}
    \hfill
    \begin{subfigure}[b]{0.195\textwidth}
        \centering
        \includegraphics[width=\textwidth]{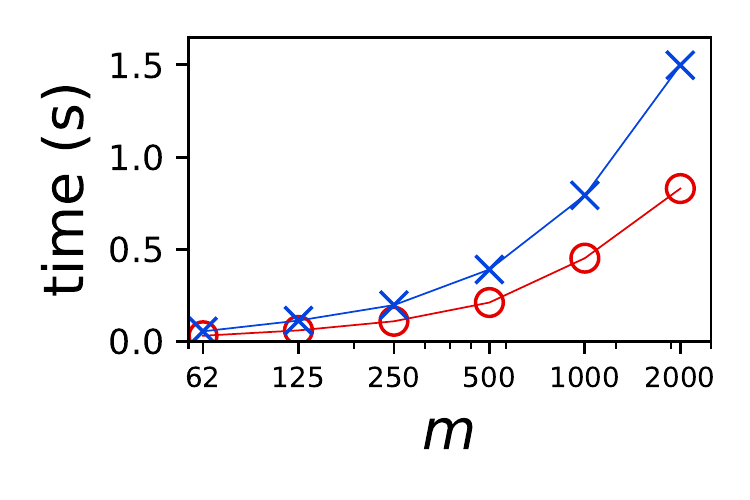}
        \caption{Adult (G+R)}
    \end{subfigure}
    \hfill
    \begin{subfigure}[b]{0.195\textwidth}
        \centering
        \includegraphics[width=\textwidth]{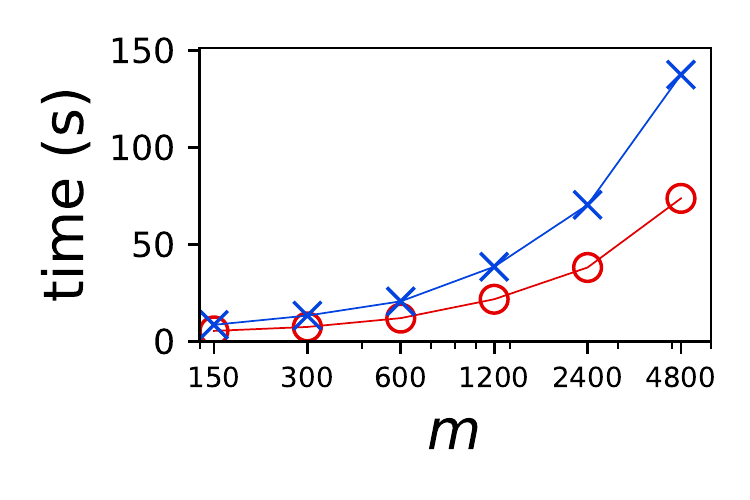}
        \caption{AntiCor\_6D}
    \end{subfigure}
    \hfill
    \begin{subfigure}[b]{0.195\textwidth}
        \centering
        \includegraphics[width=\textwidth]{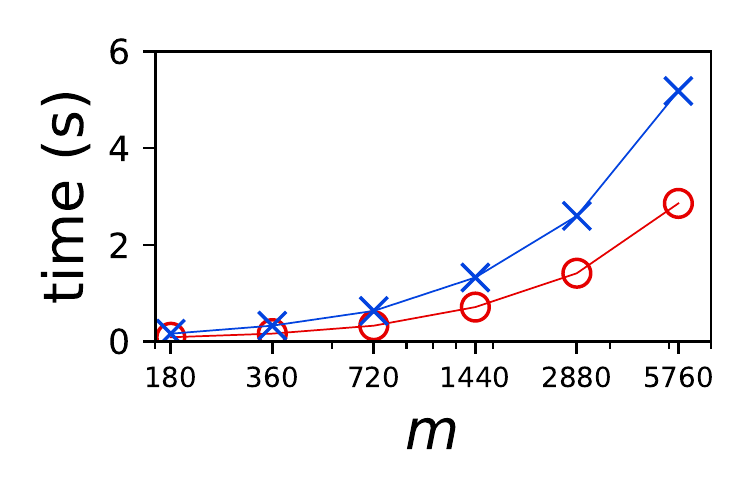}
        \caption{Compas (Gender)}
    \end{subfigure}
    \begin{subfigure}[b]{0.195\textwidth}
        \centering
        \includegraphics[width=\textwidth]{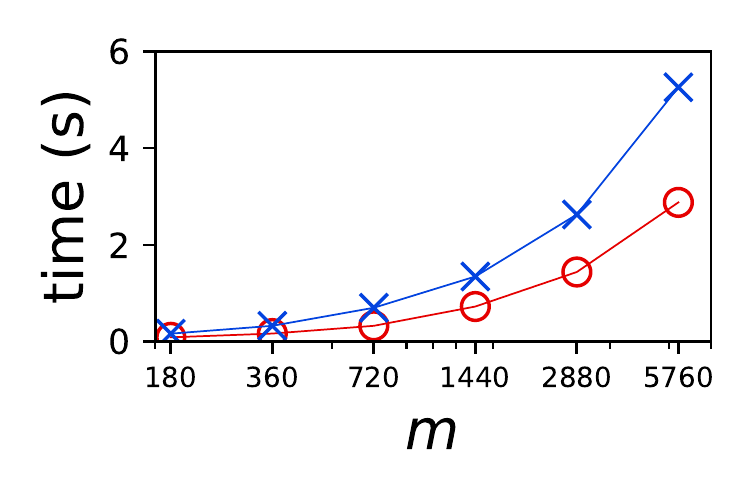}
        \caption{Compas (isRecid)}
    \end{subfigure}
    \hfill
    \begin{subfigure}[b]{0.195\textwidth}
        \centering
        \includegraphics[width=\textwidth]{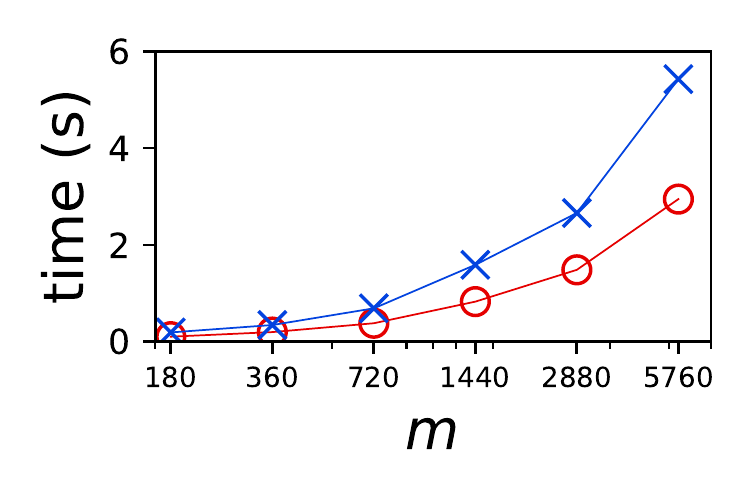}
        \caption{Compas (G+iR)}
    \end{subfigure}
    \hfill
    \begin{subfigure}[b]{0.195\textwidth}
        \centering
        \includegraphics[width=\textwidth]{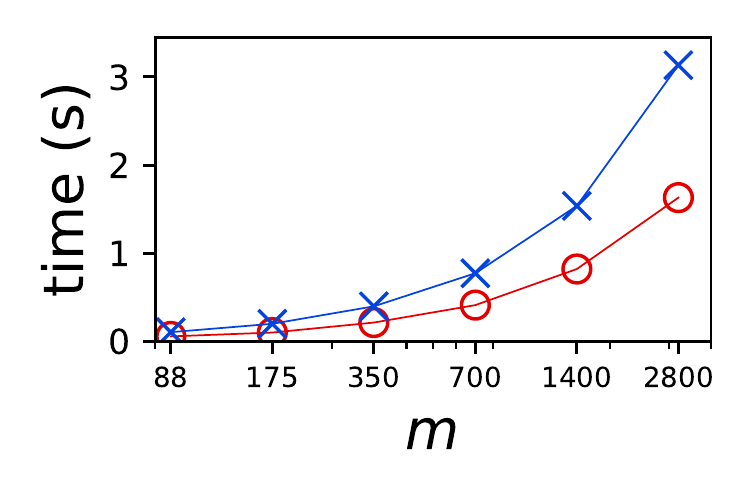}
        \caption{Credit (Job)}
    \end{subfigure}
    \hfill
    \begin{subfigure}[b]{0.195\textwidth}
        \centering
        \includegraphics[width=\textwidth]{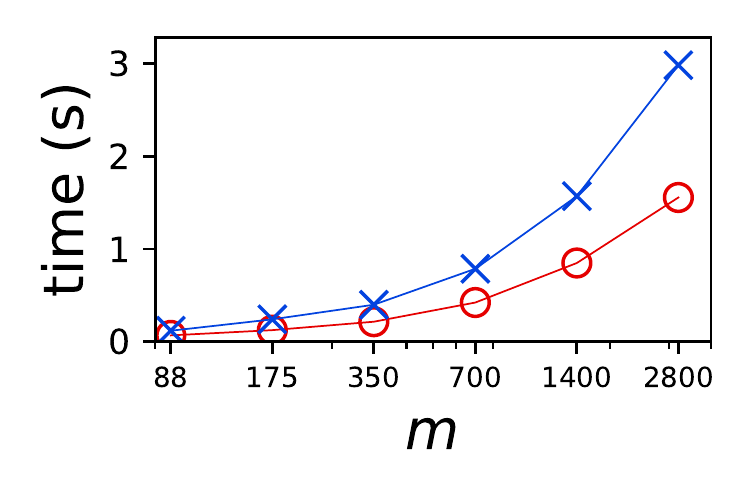}
        \caption{Credit (Housing)}
    \end{subfigure}
    \hfill
    \begin{subfigure}[b]{0.195\textwidth}
        \centering
        \includegraphics[width=\textwidth]{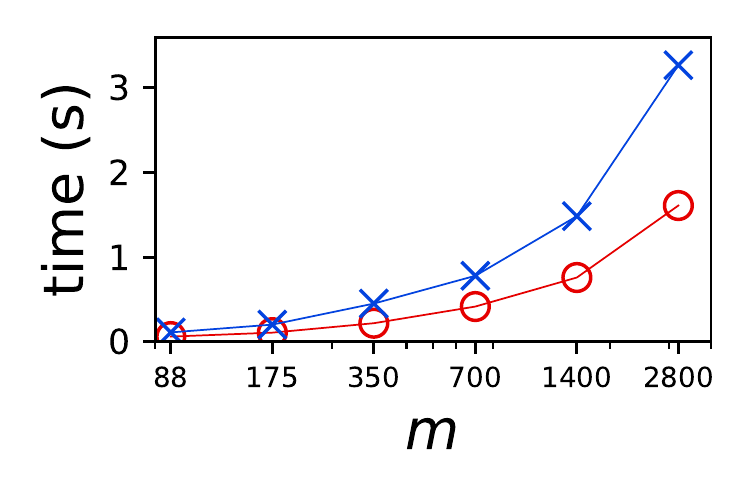}
        \caption{Credit (WY)}
    \end{subfigure}
    \caption{Results for the running time of \AlgBG and \AlgIBG by varying sample size $m$ or maximum sample size $M$.}
    \Description{experimental results}
    \label{fig:delta:time}
\end{figure*}

\section{Missing Proofs}\label{app:proofs}

In this section, we provide the proofs of theorems and lemmas omitted from the paper.

\subsection{Proof of Lemma~\ref{lm:delta:net}}

\noindent\textsc{Lemma~\ref{lm:delta:net}.}\;\;Given a $\delta$-net $\mathcal{N} \subset \mathbb{S}^{d-1}_{+}$, a database $\mathcal{D}$, and a subset $S \subseteq \mathcal{D}$, it holds that $mhr(S) \leq mhr(S | \mathcal{N}) \leq mhr(S) + \frac{2 \delta d}{1 + \delta d}$.

\begin{proof}
    First of all, since $\mathcal{N} \subset \mathbb{S}^{d-1}_{+}$, it is obvious that $mhr(S) = \min_{u \in \mathbb{S}^{d-1}_{+}} hr(u, S) \leq \min_{u \in \mathcal{N}} hr(u, S) = mhr(S | \mathcal{N})$. Then, it suffices to show that $mhr(S | \mathcal{N}) \leq mhr(S) + \frac{2 \delta d}{1 + \delta d}$. For any vector $u \in \mathbb{S}^{d-1}_{+}$, there exists a vector $v \in \mathcal{N}$ with $\langle u, v \rangle \geq \cos \delta$ by the definition of $\delta$-net and we get
    \begin{align*}
      \| u - v \| & = \sqrt{\sum_{i = 1}^{d} (u[i] - v[i])^2} = \sqrt{2 - 2 \sum_{i = 1}^{d} u[i] \cdot v[i]} \\
                  & \leq \sqrt{2 -\cos\delta} = 2 \sin(\delta/2) \leq \delta
    \end{align*}
    based on trigonometric equations. For any point $p \in [0,1]^d$, according to the Cauchy-Schwarz inequality~\cite{Agarwal:2017,Wang:2021b}, we have
    \begin{equation}
    \label{eq:net}
        |\langle u, p \rangle - \langle v, p \rangle| \leq \|u - v\| \cdot \|p\| \leq \delta \|p\| \leq \delta\sqrt{d}
    \end{equation}
    Thus, for any subset $S \subseteq \mathcal{D}$, we suppose that the MHR of $S$ over $\mathcal{D}$ is reached on vector $u'$, i.e., $u' = \argmin_{u \in \mathbb{S}^{d-1}_{+}} hr(u, S)$ and there is a vector $v' \in \mathcal{N}$ with $\langle u', v' \rangle \geq \cos \delta$. By adapting the analysis in~\cite{Agarwal:2017,Wang:2021b} from RMS to HMS, we have the following results for the vectors $u'$ and $v'$:
    \begin{align*}
      mhr(S) & = \frac{\max_{p \in S} \langle u',p \rangle}{\max_{p \in \mathcal{D}} \langle u',p \rangle} \geq \frac{\max_{p \in S} \langle v',p \rangle - \delta\sqrt{d}}{\max_{p \in \mathcal{D}} \langle v',p \rangle + \delta\sqrt{d}} \\
      & \geq \frac{mhr(S|\mathcal{N}) - \frac{\delta\sqrt{d}}{\max_{p \in \mathcal{D}} \langle v',p \rangle}}{1 + \frac{\delta\sqrt{d}}{\max_{p \in \mathcal{D}} \langle v',p \rangle}} \geq \frac{mhr(S|\mathcal{N}) - \delta d}{1 + \delta d} \\
      & \geq mhr(S|\mathcal{N}) - \frac{2 \delta d}{1 + \delta d}
    \end{align*}
    where the first inequality is based on Equation~\ref{eq:net}, the second inequality is obtained from the fact that $mhr(S|\mathcal{N}) \leq \frac{\max_{p \in S} \langle v',p \rangle}{\max_{p \in \mathcal{D}} \langle v',p \rangle}$, the third inequality is because $\max_{p \in \mathcal{D}} \langle v',p \rangle \geq 1/\sqrt{d}$, and the fourth inequality is acquired from $mhr(S|\mathcal{N}) \leq 1$.
\end{proof}

\subsection{Proof of Lemma~\ref{lm:inapprox}}

\noindent\textsc{Lemma~\ref{lm:inapprox}.}\;\;There does not exist any polynomial-time algorithm to approximate the reduced FairHMS problem defined on a set $\mathcal{N}$ of vectors in $\mathbb{S}^{d-1}_{+}$ with a factor of $(1-\varepsilon) \cdot \log{m}$, where $m = |\mathcal{N}|$, for any parameter $\varepsilon > 0$ unless P=NP.

\begin{proof}
	We prove this theorem by reducing from the \textsc{SetCover} problem, a classic NP-Hard problem~\cite{Feige:1998}, to the reduced FairHMS problem when $C = 1$ (\ie without fairness constraints). Given a ground set $\mathcal{E} = \{e_1, \ldots, e_m\}$, a set collection $\mathcal{S} = \{S_1, \ldots, S_n\}$ where $S_j \subseteq \mathcal{E}$, the \textsc{SetCover} decision problem is to determine whether there exists a size-$k$ subset $\mathcal{S}' \subseteq \mathcal{S}$ such that $\bigcup_{S \in \mathcal{S}'} S = \mathcal{E}$. Then, for each element $e_i \in \mathcal{E}$, we define an $m$-dimensional unit vector $u_i = (u_{i}[1], \ldots, u_{i}[m])$, where $u_i[j] = 1$ if $i = j$ and $u_i[j] = 0$ otherwise for each $i, j \in [1,m]$. For each set $S_j \in \mathcal{S}$, we define an $m$-dimensional point $p_j = (p_{j}[1], \ldots, p_{j}[m])$ where $p_{j}[i] = 1$ if $e_i \in S_j$ and $p_{j}[i] = 0$ otherwise for each $i \in [1,m]$ and $j \in [1,n]$. Intuitively, if a sub-collection $\mathcal{S}'$ covers $u_i$, there will exist a point $p$ w.r.t.~some $S \in \mathcal{S}'$ such that $\langle p, u_i \rangle = 1$ and vice versa. In this way, we reduce a \textsc{SetCover} instance on $\mathcal{E}$ and $\mathcal{S}$ to an HMS instance on $U = \{u_i : i \in [1, m]\}$ and $\mathcal{D} = \{p_j : j \in [1, n]\}$. Next, we show that there exists a sub-collection $\mathcal{S}'$ of $k$ sets to cover $\mathcal{E}$ \textsc{if and only if} the minimum happiness ratio $mhr(P | U)$ of the size-$k$ point set $P$ corresponding to $\mathcal{S}'$ defined on $U$ is $1$.
	\begin{itemize}[leftmargin=*]
		\item \textsc{If Direction:} If $mhr(P|U) = 1$, there exists a point $p_j \in P$ such that $\langle p_j, u_i \rangle = 1$ for any $u_i \in U$ and thus $p_j[i] = 1$, which implies that there is a set $S_j \in \mathcal{S}'$ such that $e_i \in S_j$ for every $e_i \in \mathcal{E}$. Thus, $\mathcal{S}'$ is a solution for the \textsc{SetCover} problem.
		\item \textsc{Only if Direction:} If $\bigcup_{S \in \mathcal{S}'} S = \mathcal{E}$, then, for each $u_i \in U$, there is a set $S_j \in \mathcal{S}'$ whose corresponding point $p_j$ satisfies that $\langle p_j, u_i \rangle = 1$. Moreover, since $\max_{p \in \mathcal{D}} \langle p, u_i \rangle = 1$ for any $u_i \in U$, we get $hr(u_i, P) = 1/1 = 1$ and thus $mhr(P) = 1$, where $P = \{p_j : S_j \in \mathcal{S}'\}$.
	\end{itemize}
	Since the above reduction procedure is performed in $O(nm)$ time, we prove that HMS is NP-hard from the NP-hardness of \textsc{SetCover}. Because HMS is a special case of FairHMS when $C=1$, the reduced FairHMS problem is NP-hard as well. Furthermore, it is known that there does not exist any polynomial-time algorithm that approximates the \textsc{SetCover} problem within a factor of $(1-\varepsilon) \log m$ for any $\varepsilon > 0$ unless $\mathrm{NP} \subset \mathrm{TIME}(n^{O(\log{\log n})})$ (\cite{Feige:1998}, \textsc{Theorem} 4.4). Formally, suppose that $\mathcal{S}^*$ be the smallest sub-collection of $\mathcal{S}$ that covers $\mathcal{E}$, no polynomial-time algorithm can guarantee to find a sub-collection of $\mathcal{S}'$ that covers $\mathcal{E}$ with $|\mathcal{S}'| \leq  (1-\varepsilon) \log m \cdot |\mathcal{S}^*|$ for any $\varepsilon > 0$ unless P=NP. The hardness result holds for HMS/FairHMS as well, \ie there is no polynomial-time algorithm that guarantees to find a solution $S_{\tau}$ of size $k_{\tau} = (1-\varepsilon) \log m \cdot k^*_{\tau}$ for any $\varepsilon > 0$ with $mhs(S_{\tau}) = \tau$ unless P=NP, where $k^*_{\tau}$ is the size of the smallest subset $S^*_{\tau}$ for HMS with $mhs(S^*_{\tau}) = \tau$ for any $\tau \in (0, 1)$.
\end{proof}

\subsection{Proof of Lemma~\ref{lem:capped}}

\noindent\textsc{Lemma~\ref{lem:capped}.}\;\;$mhr(S|\mathcal{N}) \geq \tau$ if and only if $mhr_{\tau}(S|\mathcal{N}) = \tau$.

\begin{proof}
    On the one hand, since $\min_{u \in \mathcal{N}} hr(u, S) \geq \tau$, we have $hr_{\tau}(u, S) = \tau$ for any $u \in \mathcal{N}$ and thus $mhr_{\tau}(S|\mathcal{N}) = \frac{m \tau}{m} = \tau$.
    
    On the other hand, we assume that there exists some $u' \in \mathcal{N}$ such that $hr(u', S) < \tau$ and thus $hr_{\tau}(u', S) = \min\{hr(u', S), \tau\} < \tau$. Since $hr_{\tau}(u, S) \leq \tau$, $mhr_{\tau}(S|\mathcal{N}) \leq \frac{(m-1)\tau + hr(u', S)}{m} < \tau$, which contradicts with $mhr_{\tau}(S|\mathcal{N}) = \tau$. Hence, we will have $hr(u, S) \geq \tau$ for each $u \in \mathcal{N}$ and $mhr(S|\mathcal{N}) \geq \tau$ if $mhr_{\tau}(S|\mathcal{N}) = \tau$.
\end{proof}

\subsection{Proof of Lemma~\ref{lem:greedy}}

\noindent\textsc{Lemma~\ref{lem:greedy}.}\;\;Let $\tau^*$ be the largest $\tau$ for which a feasible solution $S^*$ with $mhr_{\tau}(S^*|\mathcal{N}) = \tau$ exists and $S_i$ be the solution returned by the greedy algorithm at the $i$-th round on $\mathcal{D} \setminus \bigcup_{j=1}^{i-1} S_j$ for function $mhr_{\tau^*}(\cdot|\mathcal{N})$. If $\gamma \geq \lceil \log_{2}\frac{m}{\varepsilon} \rceil$, then, for $S = \bigcup_{i=1}^{\gamma} S_i$, it holds that $mhr(S|\mathcal{N}) \geq (1-\varepsilon) \cdot \tau^*$.

\begin{proof}
    First, according to Theorem 3 of~\cite{Anari:2019} generalizing from the statement in~\cite{Fisher:1978}, we have
    \begin{displaymath}
    mhr_{\tau^*}(S|\mathcal{N}) \geq (1 - 2^{-\gamma}) \cdot \tau^* \geq (1 - \frac{\varepsilon}{m}) \cdot \tau^*
    \end{displaymath}
    Suppose that there exists some $u \in \mathcal{N}$ s.t.~$hr_{\tau^*}(u, S)$ $< (1-\varepsilon) \cdot \tau^*$. Because it is known that $hr_{\tau^*}(u, S) \leq \tau^*$ for every $u \in \mathcal{N}$, we have
    \begin{align}
      mhr_{\tau}(S|\mathcal{N}) & \leq \frac{1}{m} \cdot hr_{\tau^*}(u, S) + \frac{m-1}{m} \cdot \tau^* \nonumber\\
                      & < \frac{1-\varepsilon}{m} \cdot \tau^* + \frac{m-1}{m} \cdot \tau^* = (1-\frac{\varepsilon}{m}) \cdot \tau^* \nonumber
    \end{align}
    which contradicts with $mhr_{\tau^*}(S|\mathcal{N}) \geq (1-\frac{\varepsilon}{m}) \cdot \tau^*$. Therefore, we have $hr_{\tau^*}(u, S) \geq (1-\varepsilon) \cdot \tau^*$ for $u \in \mathcal{N}$ and thus $mhr(S|\mathcal{N}) \geq (1-\varepsilon) \cdot \tau^*$ as claimed.
\end{proof}

\begin{figure*}[t]
    \centering
    \begin{subfigure}[b]{0.195\textwidth}
        \centering
        \includegraphics[width=\textwidth]{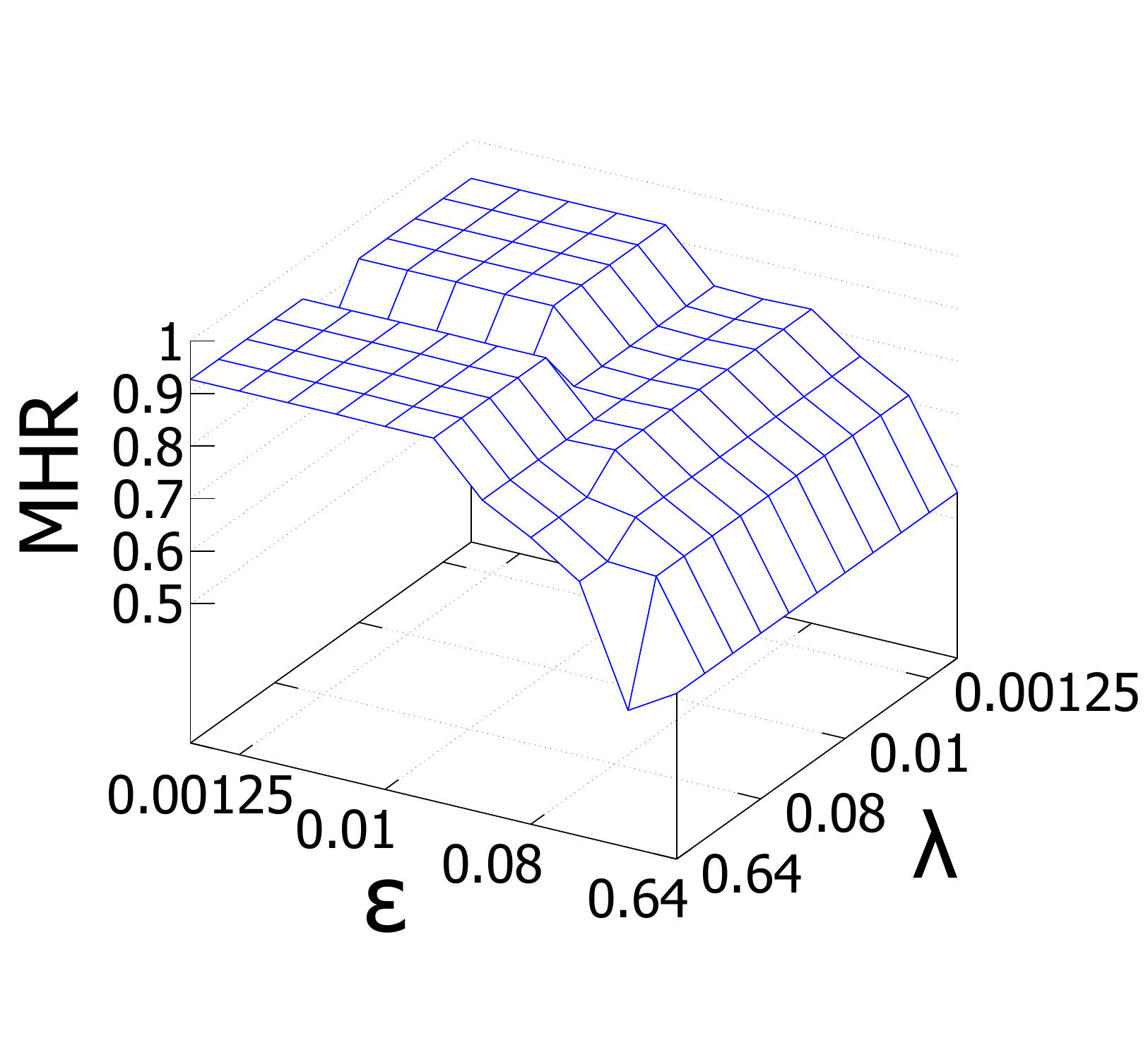}
        \caption{Adult (Gender)}
    \end{subfigure}
    \hfill
    \begin{subfigure}[b]{0.195\textwidth}
        \centering
        \includegraphics[width=\textwidth]{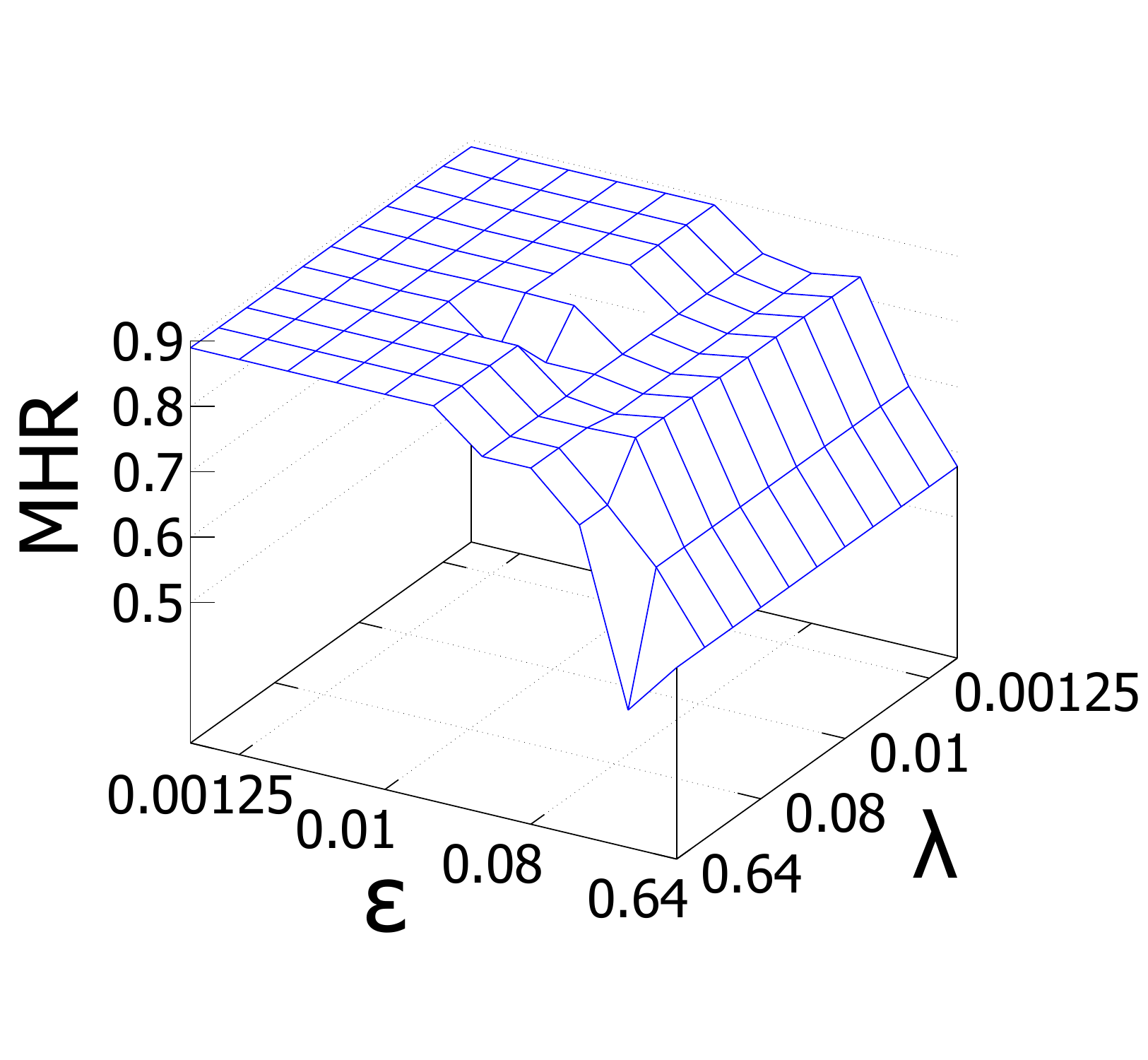}
        \caption{Adult (Race)}
    \end{subfigure}
    \hfill
    \begin{subfigure}[b]{0.195\textwidth}
        \centering
        \includegraphics[width=\textwidth]{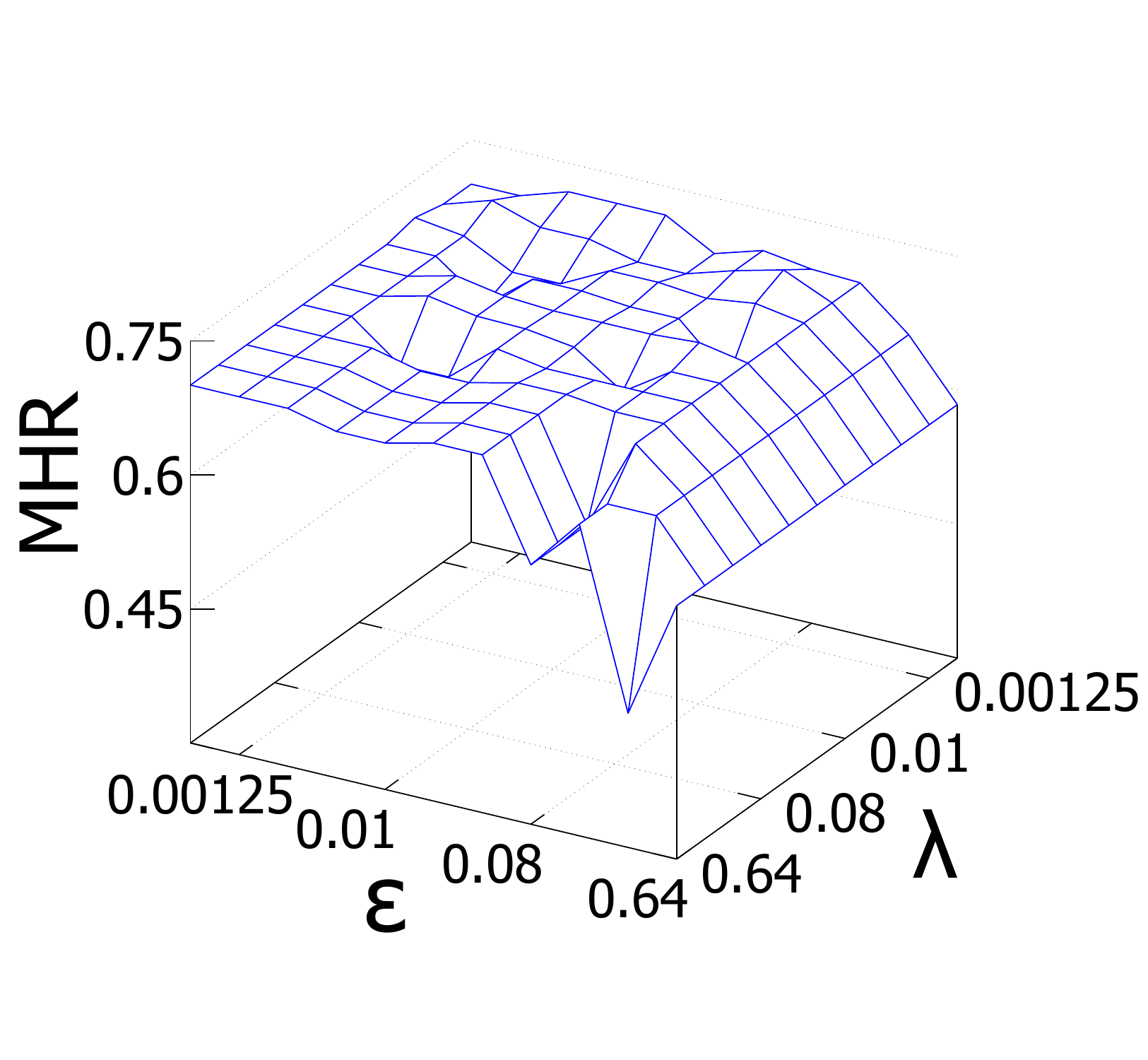}
        \caption{Adult (G+R)}
    \end{subfigure}
    \hfill
    \begin{subfigure}[b]{0.195\textwidth}
        \centering
        \includegraphics[width=\textwidth]{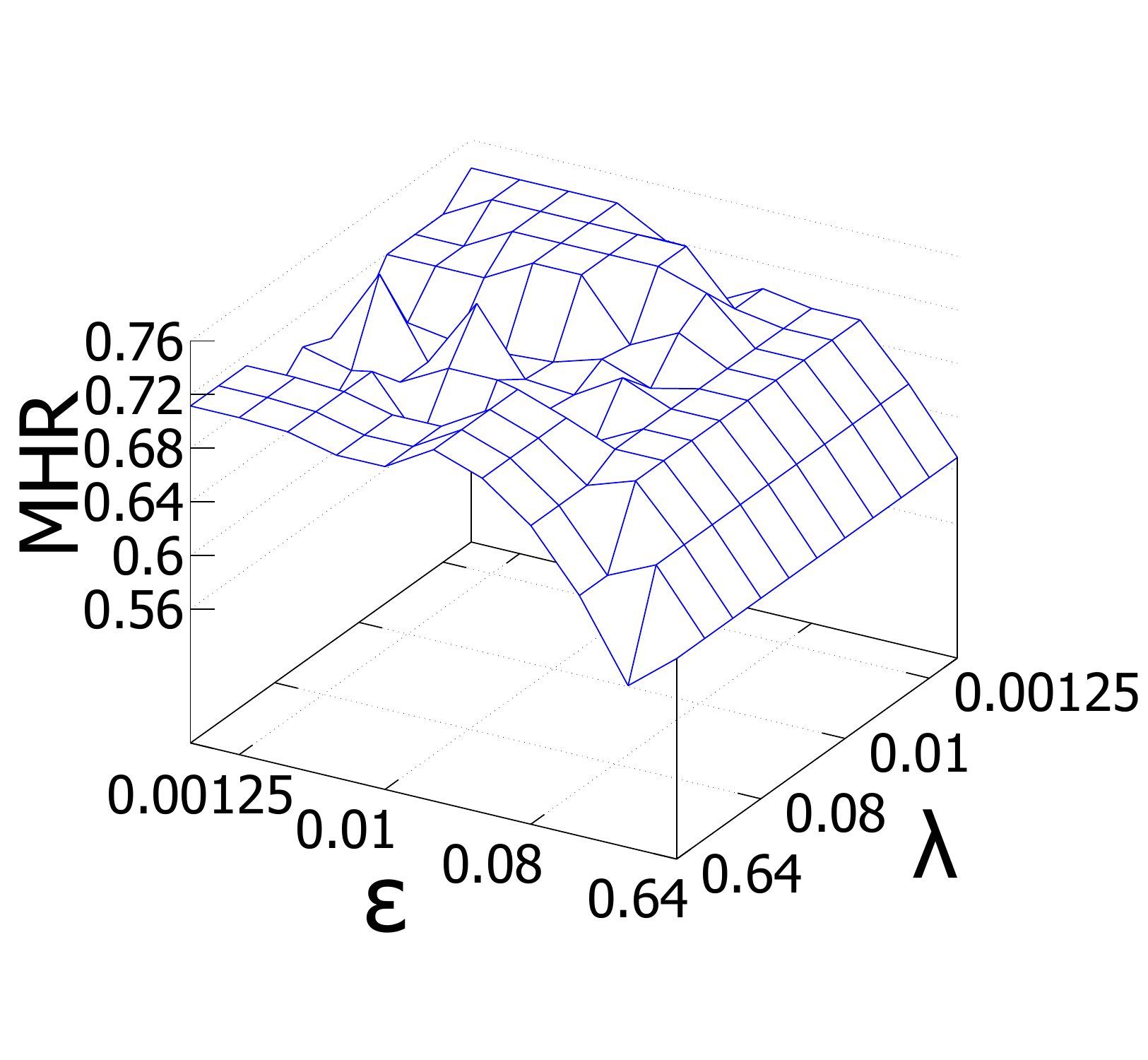}
        \caption{AntiCor\_6D}
    \end{subfigure}
    \hfill
    \begin{subfigure}[b]{0.195\textwidth}
        \centering
        \includegraphics[width=\textwidth]{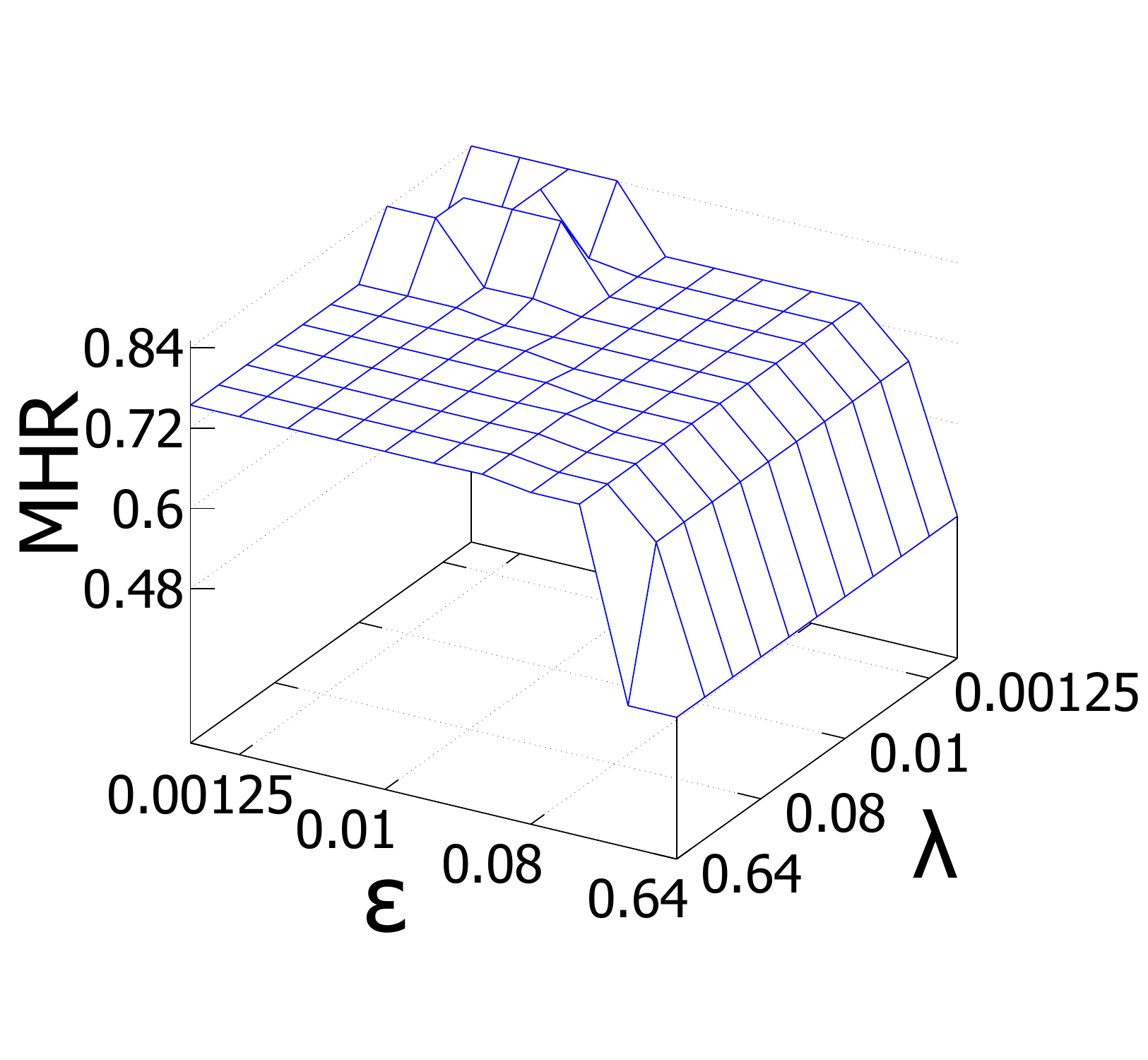}
        \caption{Compas (Gender)}
    \end{subfigure}
    \begin{subfigure}[b]{0.195\textwidth}
        \centering
        \includegraphics[width=\textwidth]{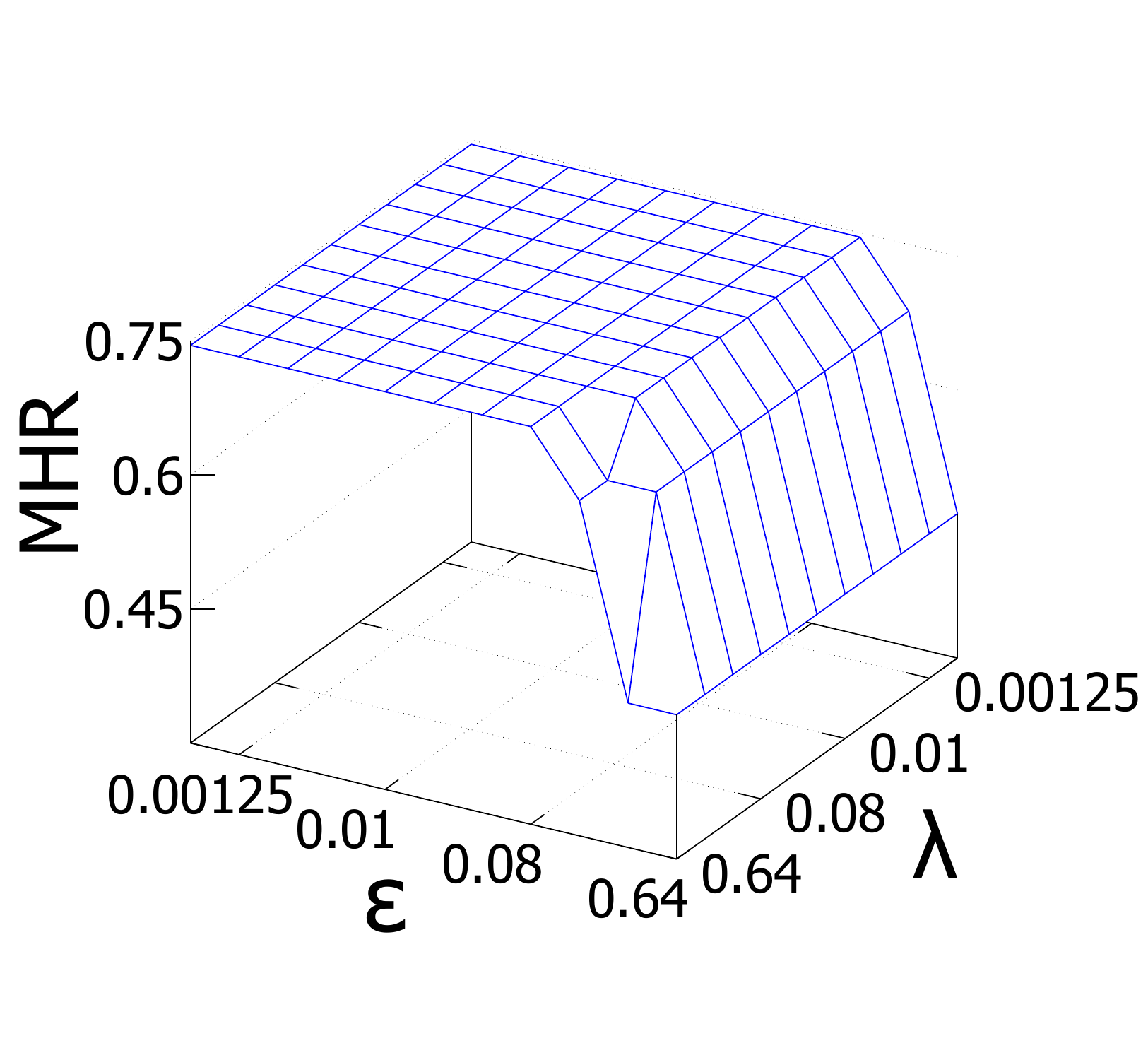}
        \caption{Compas (isRecid)}
    \end{subfigure}
    \hfill
    \begin{subfigure}[b]{0.195\textwidth}
        \centering
        \includegraphics[width=\textwidth]{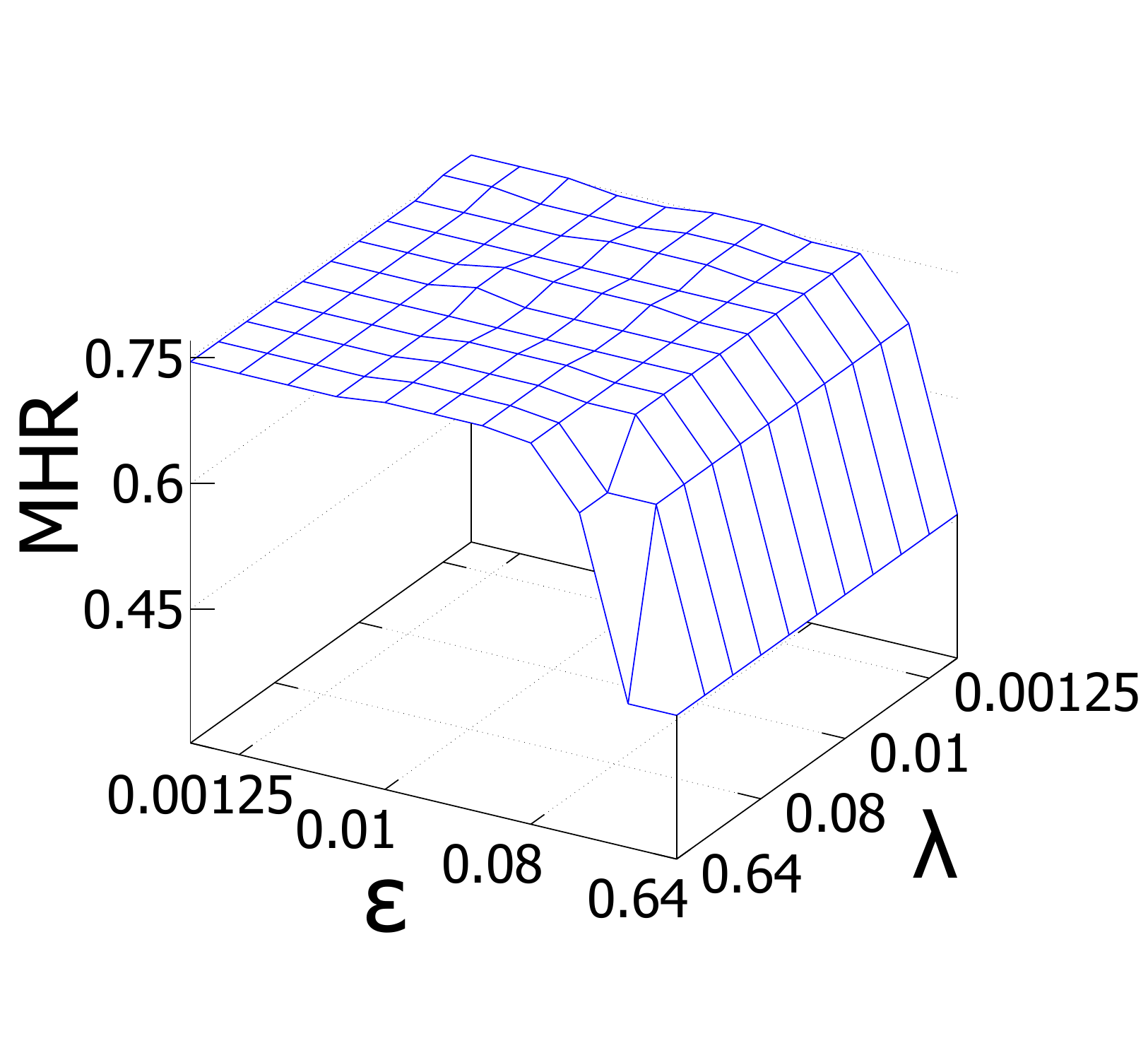}
        \caption{Compas (G+iR)}
    \end{subfigure}
    \hfill
    \begin{subfigure}[b]{0.195\textwidth}
        \centering
        \includegraphics[width=\textwidth]{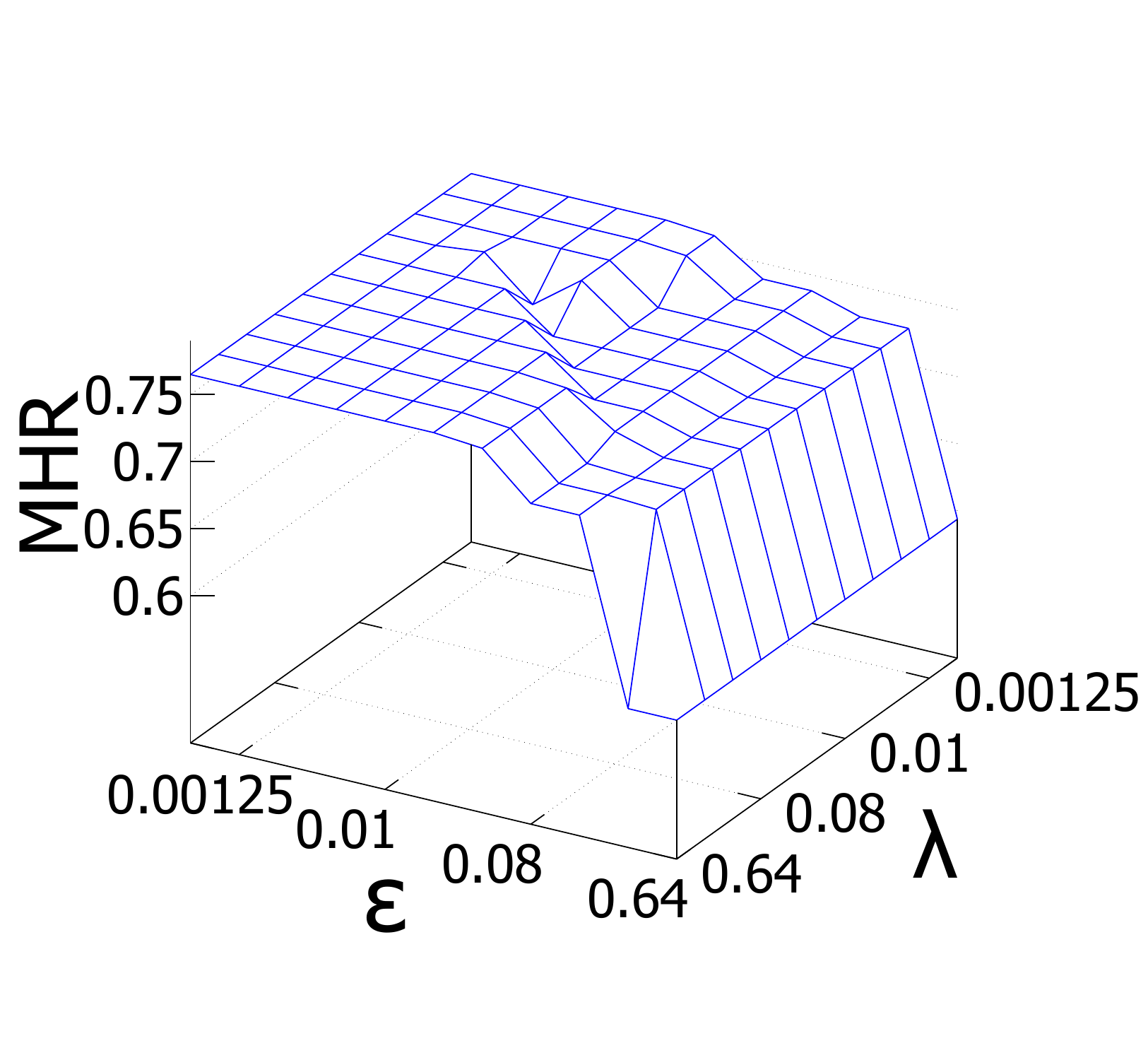}
        \caption{Credit (Job)}
    \end{subfigure}
    \hfill
    \begin{subfigure}[b]{0.195\textwidth}
        \centering
        \includegraphics[width=\textwidth]{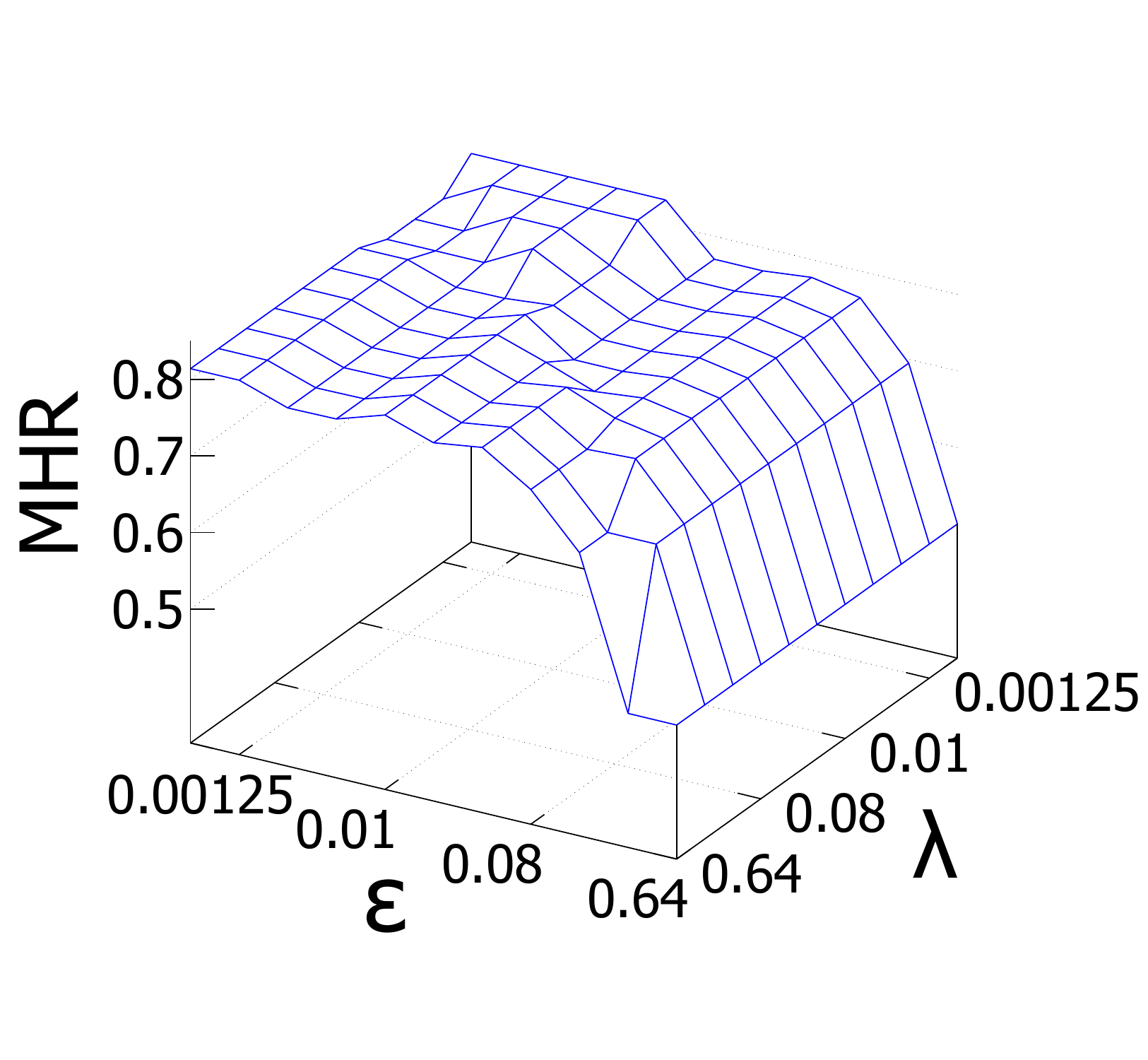}
        \caption{Credit (Housing)}
    \end{subfigure}
    \hfill
    \begin{subfigure}[b]{0.195\textwidth}
        \centering
        \includegraphics[width=\textwidth]{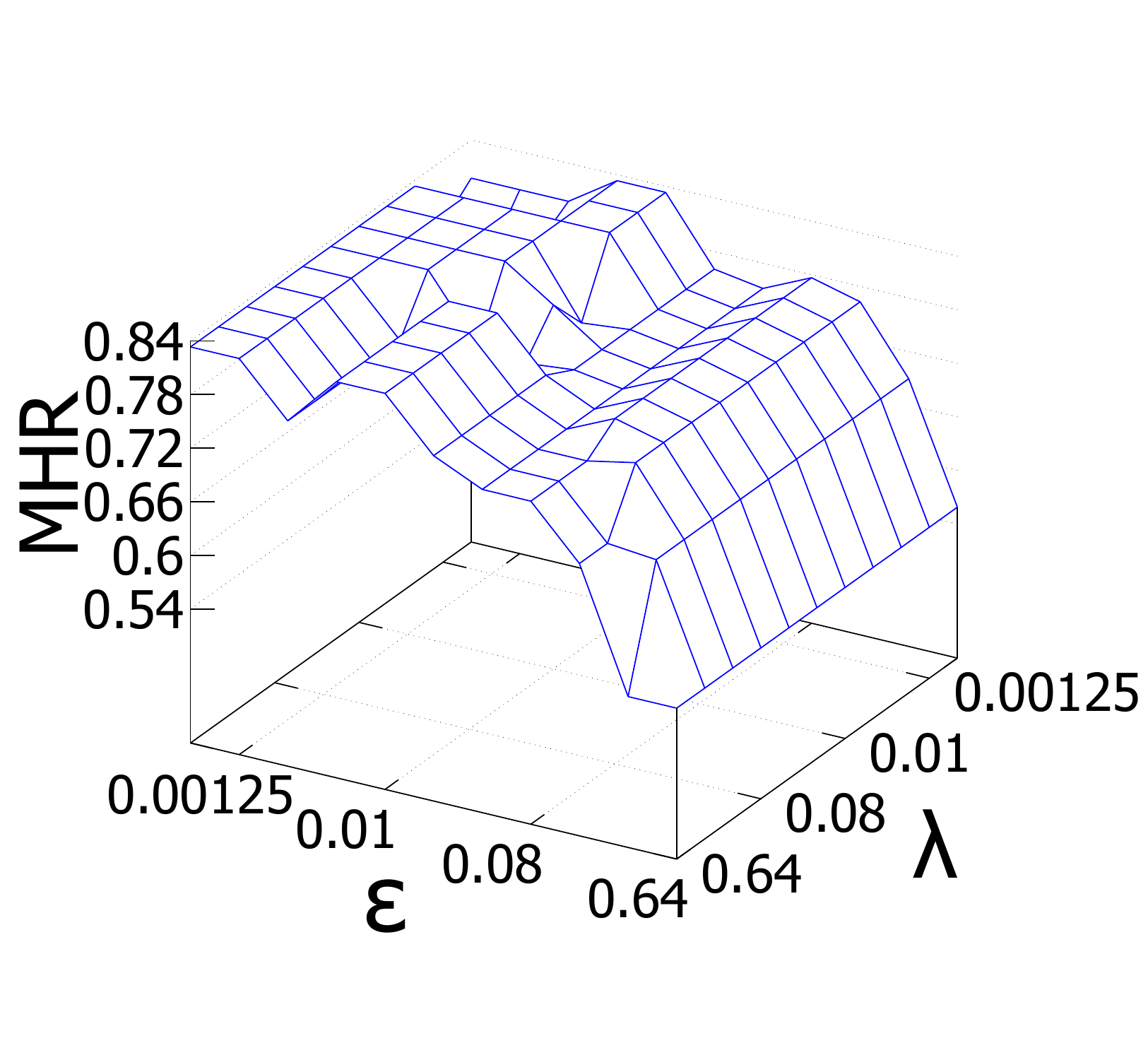}
        \caption{Credit (WY)}
    \end{subfigure}
    \caption{Results for the MHRs of \AlgIBG by varying $\varepsilon$ and $\lambda$.}
    \Description{experimental results}
    \label{fig:el:mhr}
\end{figure*}

\begin{figure*}[t]
    \centering
    \begin{subfigure}[b]{0.195\textwidth}
        \centering
        \includegraphics[width=\textwidth]{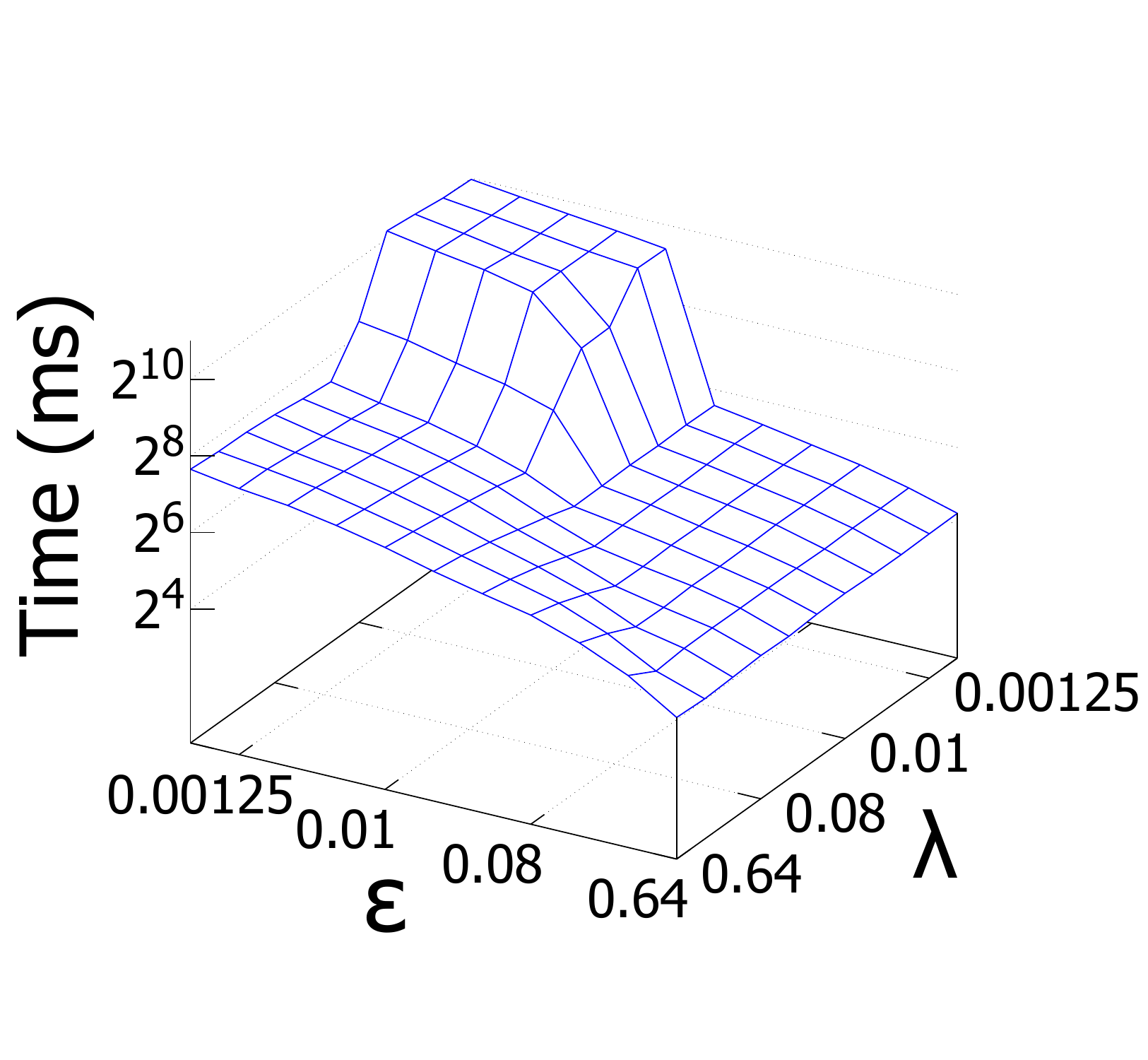}
        \caption{Adult (Gender)}
    \end{subfigure}
    \hfill
    \begin{subfigure}[b]{0.195\textwidth}
        \centering
        \includegraphics[width=\textwidth]{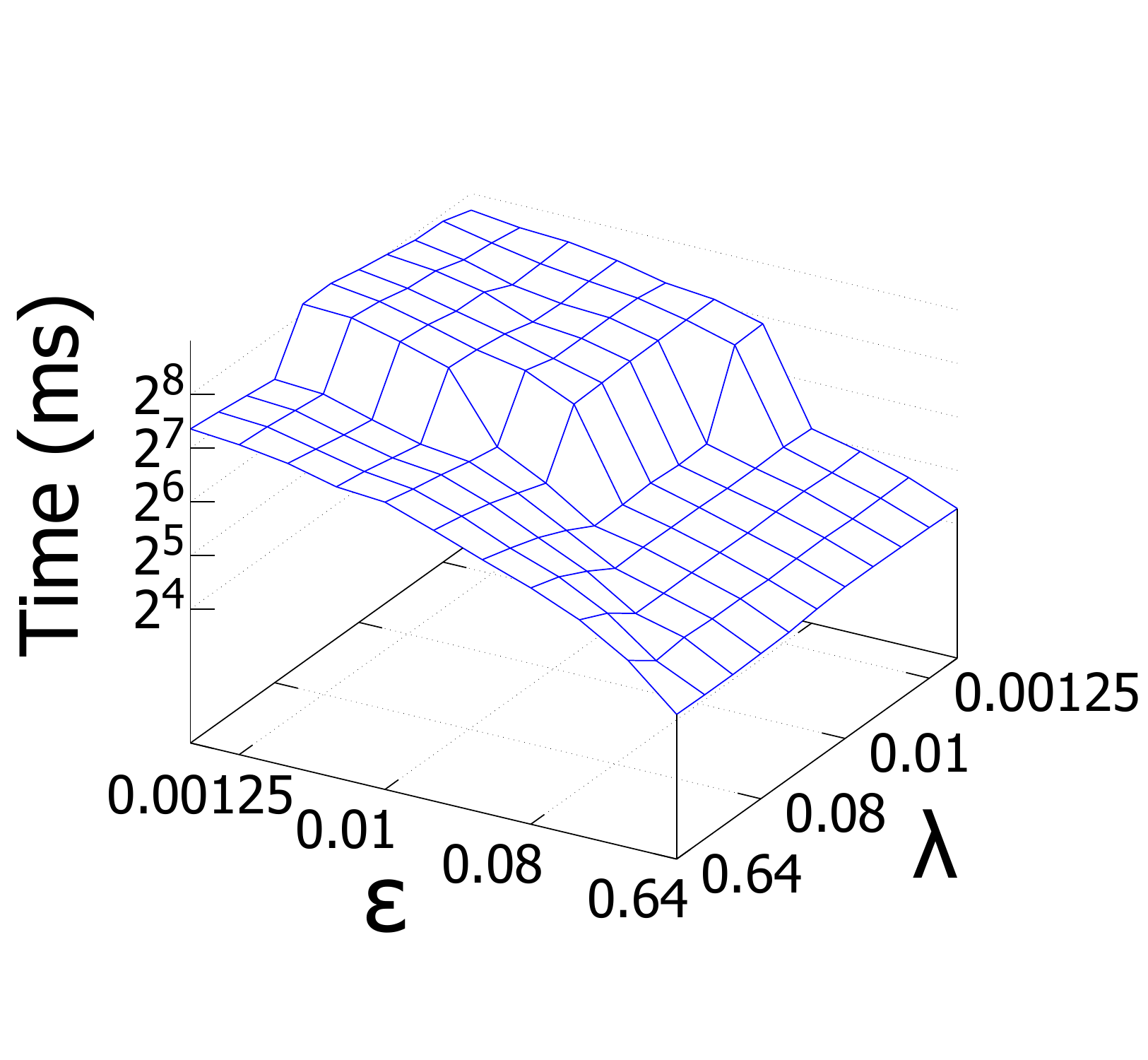}
        \caption{Adult (Race)}
    \end{subfigure}
    \hfill
    \begin{subfigure}[b]{0.195\textwidth}
        \centering
        \includegraphics[width=\textwidth]{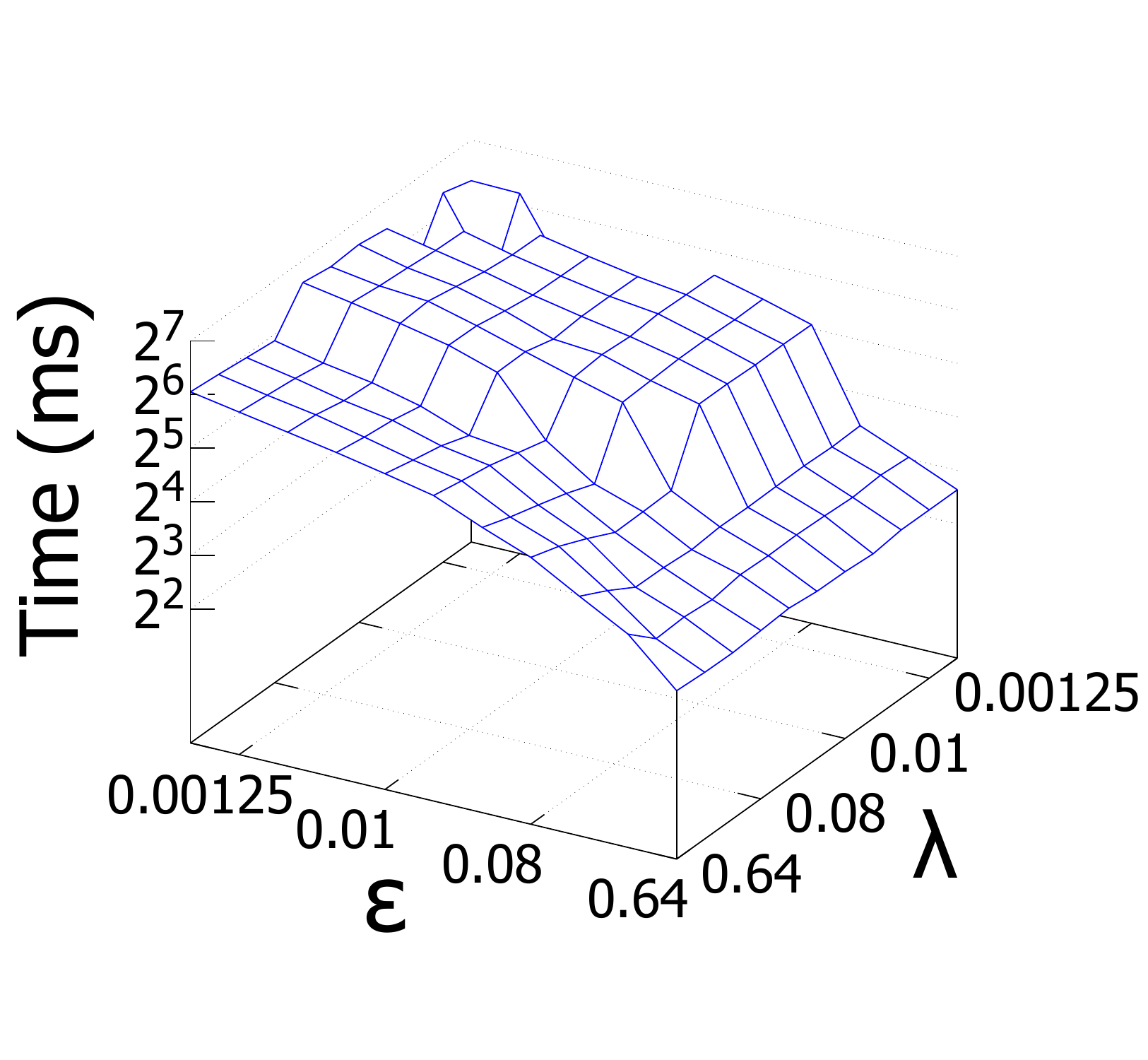}
        \caption{Adult (G+R)}
    \end{subfigure}
    \hfill
    \begin{subfigure}[b]{0.195\textwidth}
        \centering
        \includegraphics[width=\textwidth]{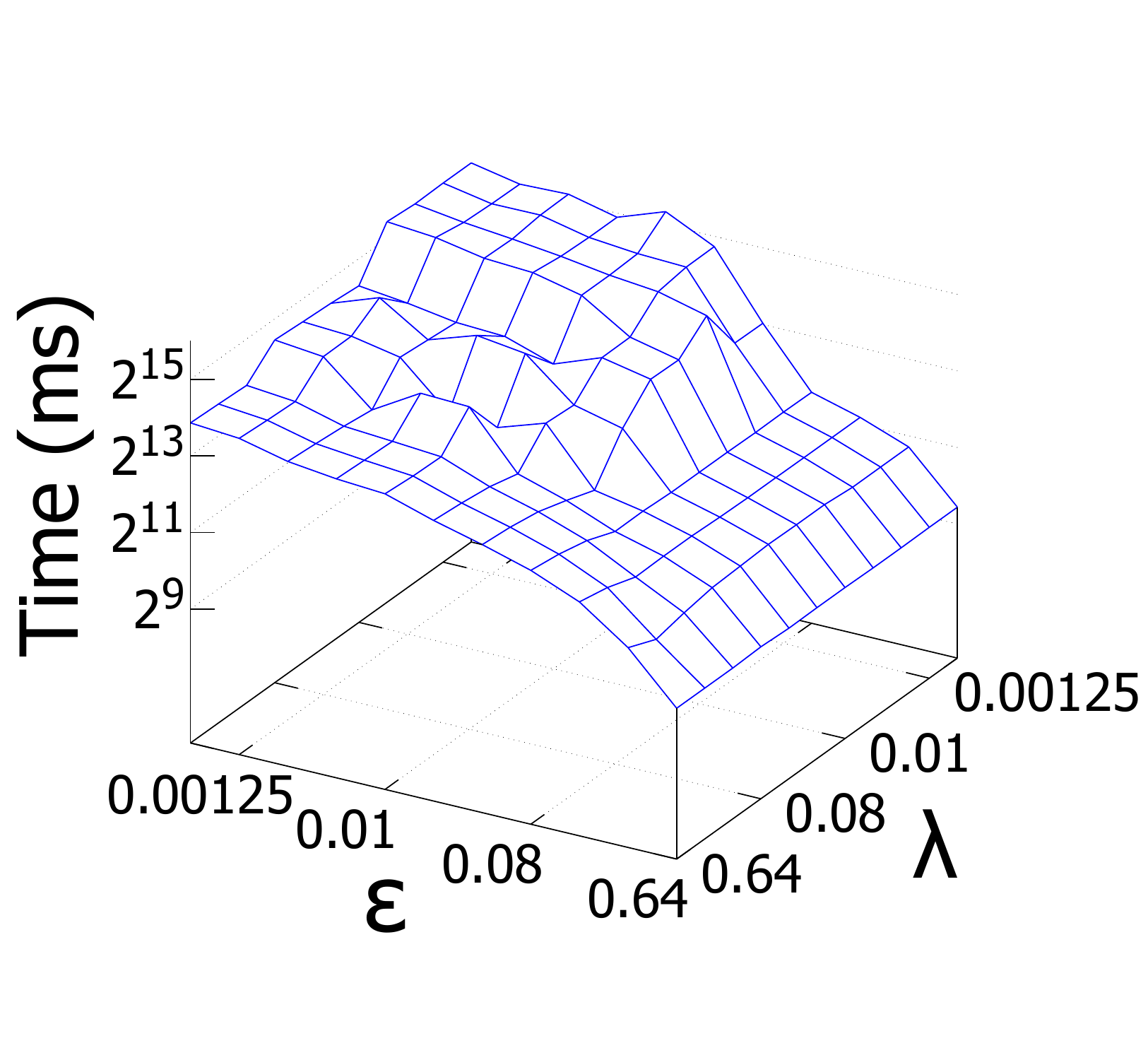}
        \caption{AntiCor\_6D}
    \end{subfigure}
    \hfill
    \begin{subfigure}[b]{0.195\textwidth}
        \centering
        \includegraphics[width=\textwidth]{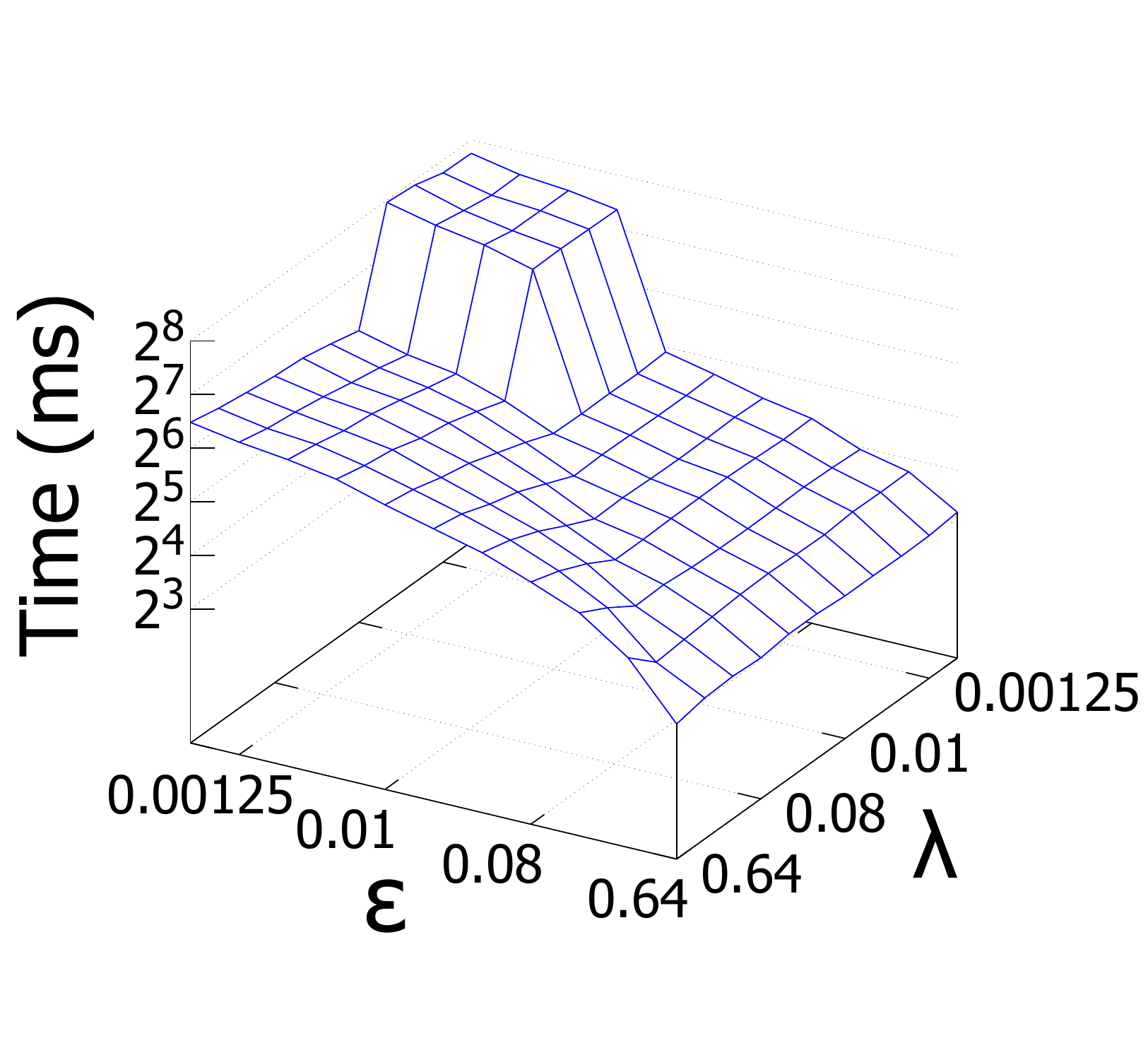}
        \caption{Compas (Gender)}
    \end{subfigure}
    \begin{subfigure}[b]{0.195\textwidth}
        \centering
        \includegraphics[width=\textwidth]{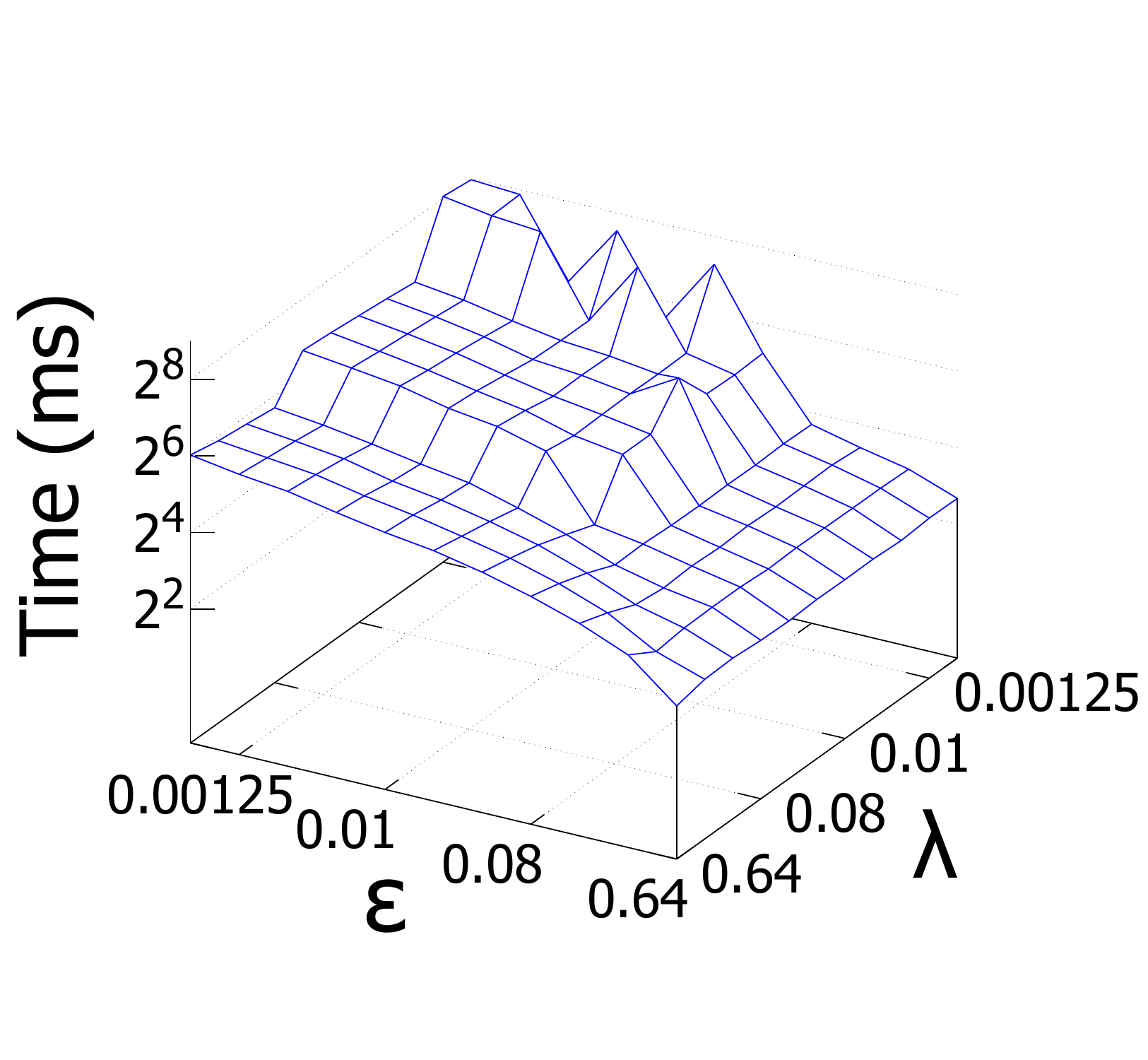}
        \caption{Compas (isRecid)}
    \end{subfigure}
    \hfill
    \begin{subfigure}[b]{0.195\textwidth}
        \centering
        \includegraphics[width=\textwidth]{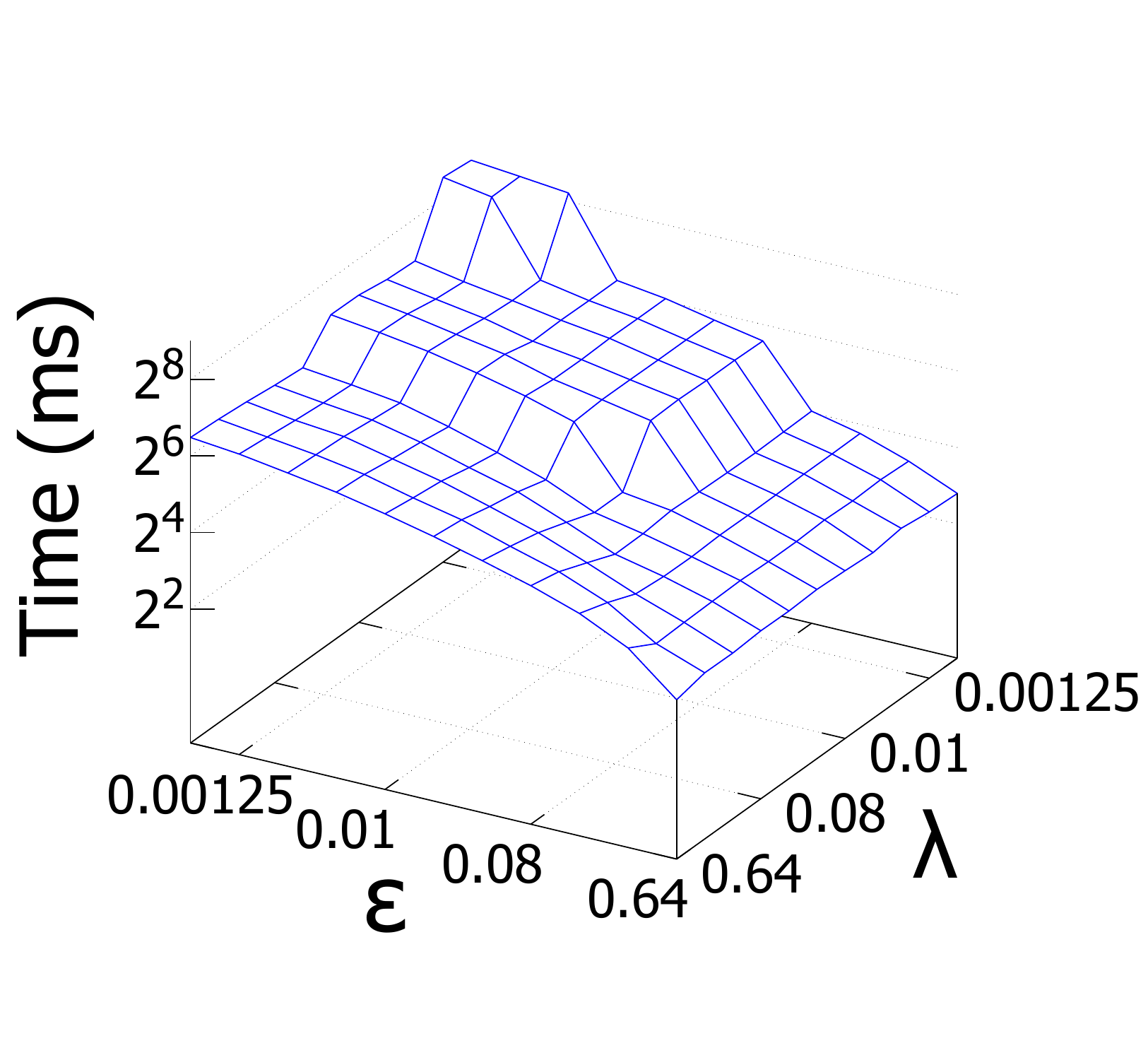}
        \caption{Compas (G+iR)}
    \end{subfigure}
    \hfill
    \begin{subfigure}[b]{0.195\textwidth}
        \centering
        \includegraphics[width=\textwidth]{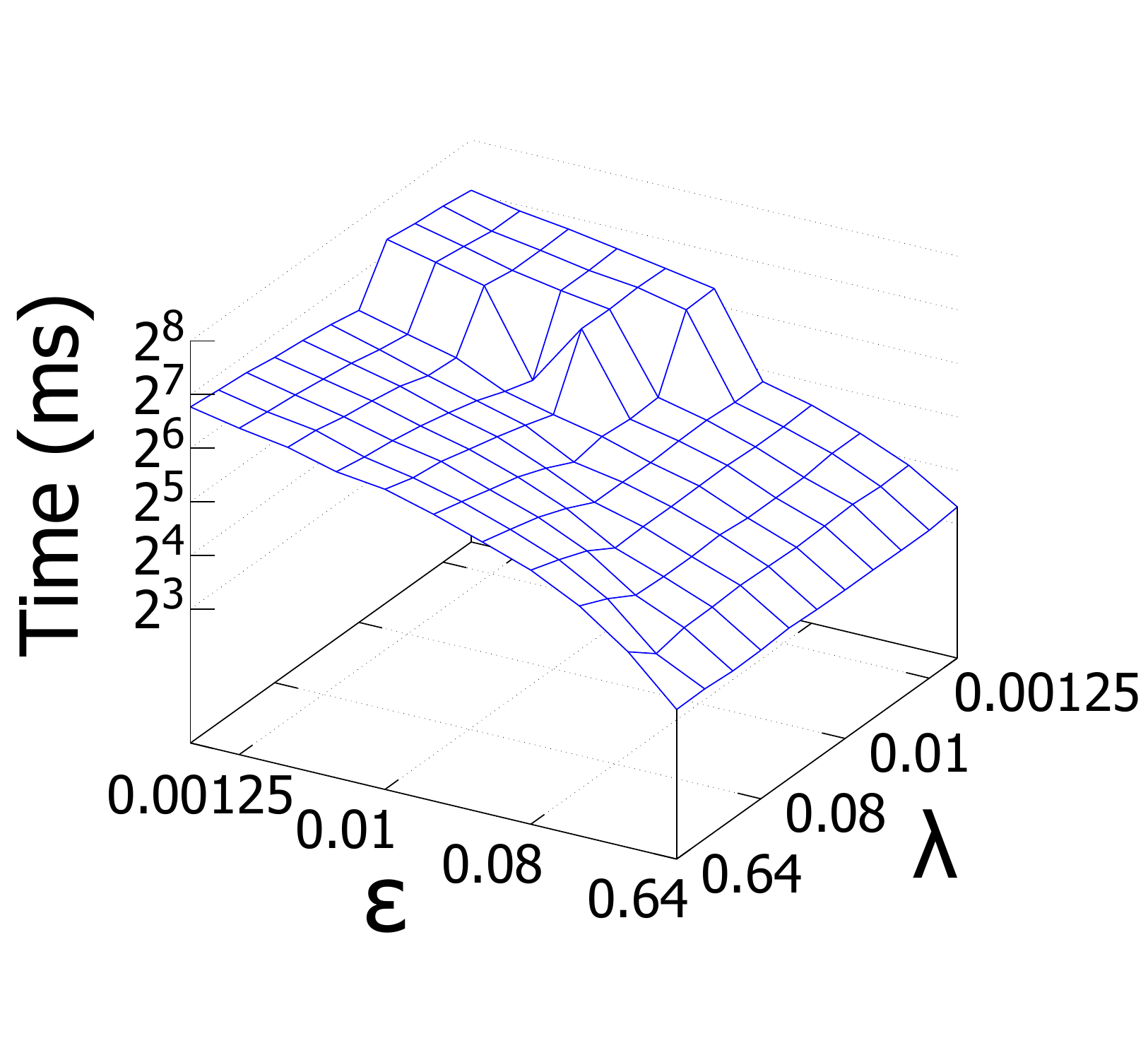}
        \caption{Credit (Job)}
    \end{subfigure}
    \hfill
    \begin{subfigure}[b]{0.195\textwidth}
        \centering
        \includegraphics[width=\textwidth]{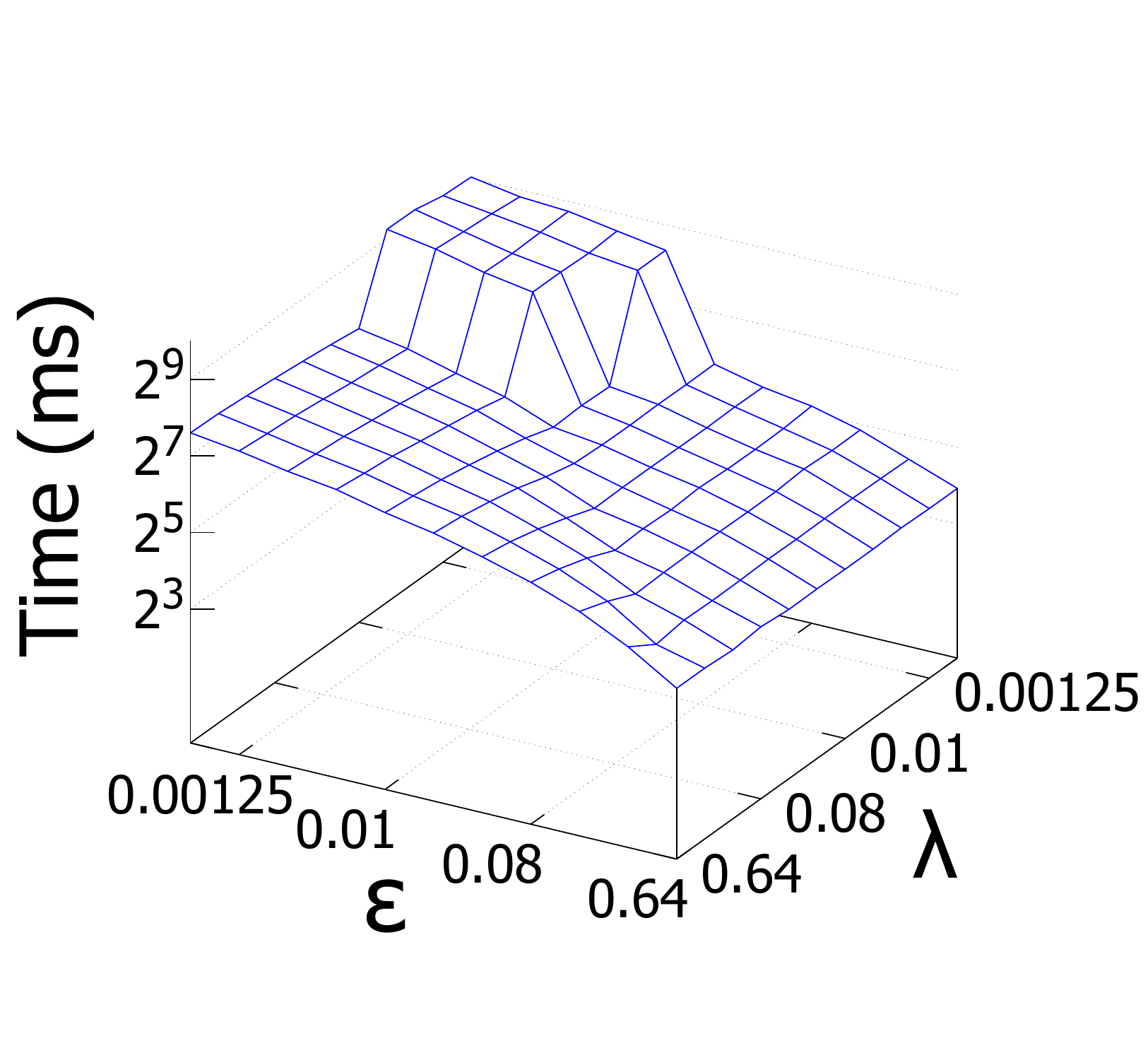}
        \caption{Credit (Housing)}
    \end{subfigure}
    \hfill
    \begin{subfigure}[b]{0.195\textwidth}
        \centering
        \includegraphics[width=\textwidth]{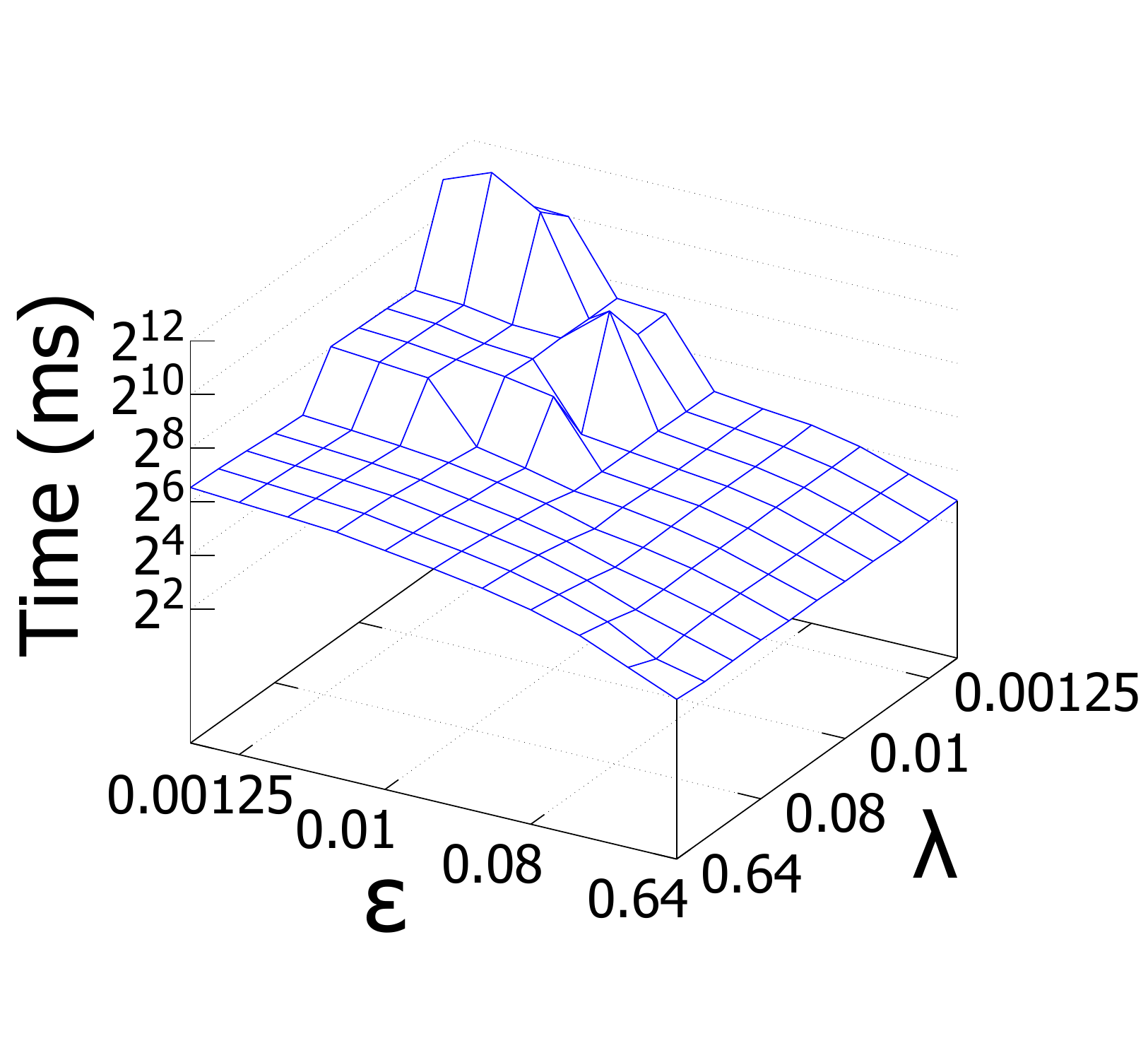}
        \caption{Credit (WY)}
    \end{subfigure}
    \caption{Results for the running time of \AlgIBG by varying  $\varepsilon$ and $\lambda$.}
    \Description{experimental results}
    \label{fig:el:time}
\end{figure*}

\section{Additional Experiments}\label{app:exp}

In this section, we provide the additional experimental results on the impact of parameters $\delta$, $\varepsilon$, and $\lambda$ on the performance of our proposed algorithms \AlgBG and \AlgIBG.

Figures~\ref{fig:delta:mhr} and~\ref{fig:delta:time} present the MHRs and running time of \AlgBG and \AlgIBG by varying the sample size $m$ of the $\delta$-net in \AlgBG or the maximum sample size $M$ in \AlgIBG.
Since the sample size $m$ of a $\delta$-net grows exponentially with respect to $d$ when the value of $\delta$ is fixed, setting a small value of $\delta$ is impractical for high dimensionality $d$.
Following existing studies on using $\delta$-nets for RMS problems~\cite{Agarwal:2017,Kumar:2018,Wang:2021b}, we consider using greater values of $\delta$ for larger $d$'s and testing the corresponding values of $m$ which guarantee that a high-quality solution can be computed in reasonable time.
In practice, we use $m = 10kd$ for \AlgBG and $M = 10kd$ and $m_0 = 0.05M$ for \AlgIBG by default in the experiments.
Here, we further vary $m$ and $M$ in $\{1.25kd, 2.5kd, 5kd, 10kd, 20kd, 40kd\}$ to show their effects on the performance of \AlgBG and \AlgIBG.
First, we observe that the MHRs of both algorithms increase in most cases when the values of $m$ and $M$ grow as the errors in MHR estimations become smaller.
However, such a trend does not strictly follow because \textsc{MRGreedy} is an approximation algorithm and might provide worse solutions even when the estimations for MHRs are more accurate.
Second, the running time of both algorithms increases nearly linearly with $m$ and $M$.
Here, \AlgIBG runs slower than \AlgBG because the setting of $\lambda$ in this set of experiments ensures that \AlgIBG reaches $m_i \geq M$ for testing whether its solution quality still increases when $M$ is larger.
In other experiments, \AlgIBG runs much faster than \AlgBG as $\lambda = 0.04$ is used and \AlgIBG terminates with smaller $m_i$ than $m$.
Based on the above results, we find that the MHRs of the solutions of both algorithms hardly increase in most cases when $m,M > 10kd$ and thus confirm the default setting of $m = 10kd$ in the experiments.

Figures~\ref{fig:el:mhr} and~\ref{fig:el:time} present the MHRs and running time of \AlgIBG for the parameters $\varepsilon \in \{0.00125, 0.0025, \ldots, 0.64\}$ and $\lambda \in \{0.00125, 0.0025, \ldots, 0.64\}$.
In terms of solution quality, we observe that the MHRs increase significantly when the values of $\varepsilon$ and $\lambda$ decrease from $0.64$ to $0.08$ while becoming steady when the values of $\varepsilon$ and $\lambda$ are smaller.
In terms of time efficiency, we find that the decreases of $\varepsilon$ and $\lambda$ both incur higher computational overhead.
The trends in both MHRs and running time are attributed to the facts that smaller $\varepsilon$ leads to more $\tau$ values attempted and smaller $\lambda$ causes larger sample sizes before termination.
Based on these results, we validate that the parameter settings of $\varepsilon=0.02$ and $\lambda=0.04$ in the experiments lead to reasonable trade-offs between efficiency and effectiveness across different datasets.

\end{document}